\documentclass[12pt]{article}%
\usepackage{amssymb}
\usepackage{amsfonts}
\usepackage{amsmath}
\usepackage[dvips]{graphicx}
\usepackage{sw20elba}
\usepackage[bookmarksnumbered=true]{hyperref}
\usepackage{version}
\usepackage[dvipsnames,usenames]{color}%
\providecommand{\U}[1]{\protect\rule{.1in}{.1in}}
\newtheorem{theorem}{Theorem}

\newtheorem{condition}[theorem]{Condition}

\newtheorem{corollary}[theorem]{Corollary}

\newtheorem{definition}[theorem]{Definition}
\newtheorem{example}[theorem]{Example}

\newtheorem{lemma}[theorem]{Lemma}
\newtheorem{notation}[theorem]{Notation}

\newtheorem{proposition}[theorem]{Proposition}
\newtheorem{remark}[theorem]{Remark}

\newenvironment{proof}[1][Proof]{\textbf{#1.} }{\ \rule{0.5em}{0.5em}}
\begin{document}

\title{On overdamping phenomena in gyroscopic systems composed of high-loss and
lossless components }
\author{Alexander Figotin\\Department of Mathematics\\University of California at Irvine\\Irvine, CA 92697
\and Aaron Welters\\Department of Mathematical Sciences\\Florida Institute of Technology\\Melbourne, FL, 32901}
\date{}
\maketitle

\begin{abstract}
Using a Lagrangian framework, we study overdamping phenomena in gyroscopic
systems composed of two components, one of which is highly lossy and the other
is lossless. The losses are accounted by a Rayleigh dissipative function. As
we have shown previously, for such a composite system the modes split into two
distinct classes, high-loss and low-loss, according to their dissipative
behavior. A principal result of this paper is that for any such system a
rather universal phenomenon of \emph{selective overdamping} occurs. Namely,
first of all the high-loss modes are all overdamped, i.e., non-oscillatory, as
are an equal number of low-loss modes. Second of all, the rest of the low-loss
modes remain oscillatory (i.e., the \emph{underdamped modes}) each with an
extremely high quality factor (Q-factor) that actually increases as the loss
of the lossy component increases. We prove that selective overdamping is a
generic phenomenon in Lagrangian systems with gyroscopic forces and give an
analysis of the overdamping phenomena in such systems. Moreover, using
perturbation theory, we derive explicit formulas for upper bound estimates on
the amount of loss required in the lossy component of the composite system for
the selective overdamping to occur in the generic case, and give Q-factor
estimates for the underdamped modes. Central to the analysis is the
introduction of the notion of a \textquotedblleft dual\textquotedblright%
\ Lagrangian system and this yields significant improvements on some results
on modal dichotomy and overdamping. The effectiveness of the theory developed
here is demonstrated by applying it to an electric circuit with a gyrator
element and a high-loss resistor.

\end{abstract}

\section{Introduction\label{sintro}}

In this paper we use the Lagrangian framework introduced in \cite{FigWel2} to
study the dissipative properties and overdamping phenomena of two-component
composite systems composed of a high-loss and lossless components, when the
system also possesses gyroscopic properties. This study applies to any
finite-dimensional linear Lagrangian system, with gyroscopic and dissipative
forces, provided (i) it has a nonnegative Hamiltonian, and (ii) losses are
accounted by a Rayleigh dissipative function, \cite[Sec. 10.11, 10.12]{Pars},
\cite[Sec. 8, 9, 46]{Gantmacher}. Such physical systems include, in
particular, many different types of rotating damped mechanical systems such as
fly wheels \cite[\S 345]{Kelv88I}, MEMS vibratory gyroscopes \cite{AcSh08},
\cite{ApoTay05}, electric networks with gyrators \cite{Tell48},
\cite{CarGio64}, or, in electrodynamics, a moving point charge driven by the
Lorentz force due to a static electromagnetic field \cite{Goldstein}.

The rest of the paper is organized as follows. In the remainder of this
introduction, we will first introduce in Subsection \ref{sinmodel} a model for
a two-component composite system with a high-loss and a lossless component
based on our Lagrangian framework introduced in \cite{FigWel2}, which is
overviewed in Subsections \ref{sinmodel} and \ref{sinframe}. We introduce then
in Subsection \ref{sinsubod} the definition of overdamped and underdamped
modes which is followed by a brief discussion on examples illustrating some of
the subtleties of overdamping phenomena in gyroscopic-dissipative systems. We
motivate our approach to overdamping in Subsection \ref{sinmot} by indicating
its relevance in the development of a theory of broadband absorption
suppression in magnetic composites. We give then an overview of the selective
overdamping phenomenon, which was first introduced in \cite{FigWel2}, and
discuss its potential as a mechanism for significant broadband absorption
suppression in composites. Finally, in Subsection \ref{sinresults} we give a
brief summary of the main results of this paper on modal dichotomy and
overdamping phenomena in gyroscopic-dissipative systems.

In Section \ref{scircuits}, we illustrate our main results based on a simple
example of an electric circuit with a resistor (lossy element) and a gyrator
(gyroscopic element). Using this example we examine analytically and
numerically the modal dichotomy and overdamping phenomena. Next, in Section
\ref{smodel}, we introduce the notion of the \textquotedblleft
dual\textquotedblright\ of a Lagrangian system which plays a key role in the
study of the modal dichotomy and overdamping. Then we discuss the spectral
problems that arise in studying the dissipative properties of eigenmodes of
Lagrangian systems. Finally, Section \ref{smainr} is devoted to the precise
formulation of all significant results in this paper in the form of theorems,
propositions, etc. and their proofs.

\subsection{Overview of our model\label{sinmodel}}

The general Euler-Lagrange equations of motion of the gyroscopic-dissipative
(Lagrangian) systems considered in this paper are of the form%
\begin{gather}
\alpha\ddot{Q}+\left(  2\theta+\beta R\right)  \dot{Q}+\eta Q=0\qquad
\text{(evolution equations),}\label{sintro1}\\
0\leq\beta\qquad\text{(loss parameter),} \label{sintrobeta}%
\end{gather}
where $\dot{Q}=\partial_{t}Q$, $\ddot{Q}=\partial_{t}^{2}Q$, $\beta$ is a
scalar perturbation parameter ($\beta$ is a dimensionless \emph{loss
parameter} which we introduce to scale the intensity of dissipation), and the
$N\times N$ matrices $\alpha$, $\eta$, $\theta$, $R$ have the properties that
their \textit{matrix entries are real} and
\begin{equation}
\alpha^{\mathrm{T}}=\alpha>0,\text{ }\eta^{\mathrm{T}}=\eta\geq0,\text{
}\theta^{\mathrm{T}}=-\theta,\text{ \ }R^{\mathrm{T}}=R\geq0, \label{sintro1a}%
\end{equation}
($\mathrm{T}$ denotes the transpose of a matrix). We also assume the rank
$N_{R}$ of the matrix $R$ is positive:%
\begin{equation}
0<N_{R}=\operatorname{rank}R \label{sintro1b}%
\end{equation}
(i.e., the dimension $N_{R}$ of the range of $R$ is positive). We will refer
to this dissipative system with equations of motion (\ref{sintro1}) as
\textit{gyroscopic} if $\theta\not =0$ and \textit{non-gyroscopic} if
$\theta=0$.

Here the terms involving $\beta R$ and $\theta$ correspond respectively to
dissipative and gyroscopic forces of the Lagrangian system, in which the
Lagrangian $\mathcal{L}$ and the Rayleigh dissipation function $\mathcal{R}$
are the following quadratic forms
\begin{gather}
\mathcal{L}=\mathcal{L}\left(  Q,\dot{Q}\right)  =\frac{1}{2}\left[
\begin{array}
[c]{l}%
\dot{Q}\\
Q
\end{array}
\right]  ^{\mathrm{T}}\left[
\begin{array}
[c]{ll}%
\alpha & \theta\\
\theta^{\mathrm{T}} & -\eta
\end{array}
\right]  \left[
\begin{array}
[c]{l}%
\dot{Q}\\
Q
\end{array}
\right]  \qquad\text{(the Lagrangian),}\label{sintro2}\\
\mathcal{R}=\mathcal{R}\left(  \dot{Q}\right)  =\frac{1}{2}\dot{Q}%
^{\mathrm{T}}\beta R\dot{Q}\qquad\text{(the Rayleigh dissipation function).}
\label{sintro2a}%
\end{gather}
Eqs. (\ref{sintro1}) are the Euler-Lagrange (EL) equations with the
dissipative forces $\frac{\partial\mathcal{R}}{\partial\dot{Q}}$, namely,%
\begin{equation}
\frac{d}{dt}\left(  \frac{\partial\mathcal{L}}{\partial\dot{Q}}\right)
-\frac{\partial\mathcal{L}}{\partial Q}=-\frac{\partial\mathcal{R}}%
{\partial\dot{Q}}\qquad\text{(EL equations with dissipative forces),}
\label{sintro2b}%
\end{equation}
where the generalized coordinates $Q$ and velocities $\dot{Q}$ take values in
the Euclidean space $%
\mathbb{R}
^{N}$. The Hamiltonian $\mathcal{H}\geq0$ corresponding to the Lagrangian
$\mathcal{L}$ can be represented as a function $Q$ and $\dot{Q}$ in the
following form:%
\begin{equation}
\mathcal{L}=\mathcal{T}-\mathcal{V},\qquad0\leq\mathcal{H}=\mathcal{T}%
+\mathcal{V=}\frac{1}{2}\dot{Q}^{\mathrm{T}}\alpha\dot{Q}+\frac{1}%
{2}Q^{\mathrm{T}}\eta Q\qquad\text{ (Hamiltonian),} \label{sintro3a}%
\end{equation}
where $\mathcal{T}$ and $\mathcal{V}$ are respectively \emph{the kinetic and
the potential energies} of the form%
\begin{align}
\mathcal{T}  &  =\mathcal{T}(\dot{Q},Q)=\frac{1}{2}\dot{Q}^{\mathrm{T}}%
\alpha\dot{Q}+\frac{1}{2}\dot{Q}^{\mathrm{T}}\theta Q,\label{sintro3b}\\
\mathcal{V}  &  =\mathcal{V}(\dot{Q},Q)=\frac{1}{2}Q^{\mathrm{T}}\eta
Q-\frac{1}{2}\dot{Q}^{\mathrm{T}}\theta Q.\nonumber
\end{align}
The solutions of Eq. (\ref{sintro1}) satisfy the energy balance equation:
\begin{equation}
-\partial_{t}\mathcal{H}=2\mathcal{R}\geq0\qquad\text{(energy balance
equation),} \label{sintro3}%
\end{equation}
which expresses the energy lost per unit time, where the system energy (or
stored energy) is represented by the Hamiltonian $\mathcal{H}\geq0$, the
dissipated power is $2\mathcal{R}\geq0$.

The model of a two-component composite system (TCCS) made of a lossy and a
lossless components incorporates losses represented by the Rayleigh
dissipation matrix $R$ and the loss fraction parameter%
\begin{equation}
\delta_{R}=\frac{N_{R}}{N}\qquad\text{(loss fraction).} \label{sintro4}%
\end{equation}
The lossy component of system can be roughly characterized by the range
$\operatorname{Ran}R$ of the matrix $R$ with the lossless component being its
nullspace $\operatorname{Ker}R$. The loss fraction $\delta_{R}$ defined by
(\ref{sintro4}) is then interpreted as the ratio of the degrees of freedom
susceptible to losses (i.e., $N_{R}=\dim\operatorname{Ran}R$) to the degrees
of freedom of the entire system (i.e., $N$). When considering a TCCS model we
assume that the following condition is satisfied%
\begin{equation}
0<\delta_{R}<1\qquad\text{(loss fraction condition),} \label{sintro4a}%
\end{equation}
that is, the nonzero matrix $R$ does not have full rank (i.e., $R$ is rank deficient).

A function $Q=Q\left(  t\right)  =Q\left(  t,\beta\right)  $ is a solution of
Eq. (\ref{sintro1}) if $Q$, $\dot{Q}=\partial_{t}Q$, and $\ddot{Q}%
=\partial_{t}^{2}Q$ are continuous functions of the independent variable $t$
into $%
\mathbb{C}
^{N}$ and satisfy (\ref{sintro1}) for all $t\in%
\mathbb{R}
$. The eigenmodes of the Lagrangian system are solutions of Eq. (\ref{sintro1}%
) of the form
\begin{equation}
Q\left(  t\right)  =qe^{-\mathrm{i}\zeta t},\text{ }0\not =q\in%
\mathbb{C}
^{N}\qquad\text{(eigenmode).} \label{sintro5}%
\end{equation}
Its \emph{frequency} $\omega$ and \emph{damping factor} $\gamma$ are defined
in terms of the real and imaginary part of $\zeta$, i.e.,
\begin{equation}
\omega=\operatorname{Re}\zeta\quad\text{(frequency),}\qquad0\leq
\gamma=-\operatorname{Im}\zeta\quad\text{(damping factor).} \label{sintro5a}%
\end{equation}
The damping factor is nonnegative due to the fact that for such a mode the
energy balance equation (\ref{sintro3}) still holds, but now\ in the complex
inner product $\left(  a,b\right)  =a^{\ast}b$ for $a$,$b\in%
\mathbb{C}
^{N}$, where $\ast$ denotes the complex conjugate transpose of vectors or matrices.

An important figure-of-merit, which characterizes the performance of the
dissipative system (\ref{sintro1}), is the \emph{quality factor (Q-factor)}
that can be naturally introduced in a few not entirely equivalent ways (see,
for instance, \cite[pp. 47, 70, and 71]{Pain}). When the system is in the
time-harmonic state (\ref{sintro5}), with frequency and damping factor
(\ref{sintro5a}), the quality factor $Q_{\zeta}$ is most commonly defined as
the reciprocal of the relative rate of energy dissipation per temporal cycle,
that is,%
\begin{equation}
Q_{\zeta}=2\pi\frac{\text{energy stored}}{\text{energy lost per cycle}%
}=\left\vert \omega\right\vert \frac{\mathcal{H}}{-\partial_{t}\mathcal{H}%
}=\frac{1}{2}\frac{\left\vert \omega\right\vert }{\gamma}\text{\quad
(Q-factor),} \label{sintro6}%
\end{equation}
with the convention $Q_{\zeta}=+\infty$ if $\gamma=0$ and $\omega\not =0$ and
$Q_{\zeta}=0$ if $\zeta=0$.

\subsection{The subtleties of overdamping phenomena\label{sinsubod}}

For the purposes of this paper, the following definitions of an overdamped and
an underdamped mode will be sufficient.

\begin{definition}
[overdamped mode]\label{defodm}Any eigenmode (\ref{sintro5}) of the Lagrangian
system (\ref{sintro1}) with time-dependency $e^{-\mathrm{i}\zeta\left(
\beta\right)  t}$ for which there exists a $\beta^{\prime}\geq0$ such that its
frequency $\omega$ has the property
\begin{equation}
\omega=\operatorname{Re}\zeta\left(  \beta\right)  =0,\quad\text{for all
}\beta>\beta^{\prime}\qquad\text{(overdamped),} \label{sintro7}%
\end{equation}
or
\begin{equation}
\omega=\operatorname{Re}\zeta\left(  \beta\right)  \not =0,\quad\text{for all
}\beta>\beta^{\prime}\qquad\text{\ (underdamped),} \label{sintro7a}%
\end{equation}
will be called an overdamped mode (and is said to be overdamped) or
underdamped mode (and is said to be underdamped), respectively.
\end{definition}

In order to appreciate the subtleties of overdamping that we want to study in
this paper, we will give some simple examples and recall some previous results
on overdamping.

\begin{example}
[spring-mass-damper]\label{esmdod}For the simplest mechanical (non-gyroscopic)
system of a spring-mass-damper system with one degree-of-freedom ($N=1$), the
equations of motion of this Lagrangian system (\ref{sintro1}) in standard form
is%
\[
m\ddot{x}+\beta R\dot{x}+kx=0,
\]
where $\alpha=m>0$ is the mass, $\beta R$ is the damping (with $R>0$,
$\beta\geq0$, $\delta_{R}=N_{R}/N=1$), and $\eta=k>0$ is the spring constant,
$Q=x$ is the displacement from equilibrium at $x=0$, $\dot{Q}=\dot{x}$ is its
velocity, and $\theta=0$. This mechanical system has the Lagrangian,
Hamiltonian, and Rayleigh dissipation function:%
\[
\mathcal{L}=\mathcal{T}-\mathcal{V},\text{ }\mathcal{H}=\mathcal{T}%
+\mathcal{V},\text{ }\mathcal{R}=\frac{1}{2}\beta R\left\vert \dot
{x}\right\vert ^{2},\text{ }\mathcal{T}=\frac{1}{2}m\left\vert \dot
{x}\right\vert ^{2},\text{ }\mathcal{V}=\frac{1}{2}k\left\vert x\right\vert
^{2},
\]
where $\mathcal{T}$, $\mathcal{V}$\ are the kinetic and potential energy,
respectively. The eigenmodes of the system have time-dependency
$e^{-\mathrm{i}\zeta_{j}\left(  \beta\right)  t}$, $j=1,2$ with%
\[
\zeta_{j}\left(  \beta\right)  =-\mathrm{i}\frac{\beta R}{2m}+\left(
-1\right)  ^{j}\sqrt{\frac{k}{m}-\left(  \frac{\beta R}{2m}\right)  ^{2}%
},\quad j=1,2.
\]
Thus, all the modes of this system will be overdamped (according to our
definition \ref{defodm}) once
\[
\beta>\beta^{\prime}\text{, where }\beta^{\prime}=\frac{2\sqrt{mk}}{R}.
\]

\end{example}

The simple example above illustrates a general result on overdamping for
non-gyroscopic systems with only lossy components. The next theorem from
\cite[Theorem 17]{FigWel2} (see also \cite{Duff55}, \cite{BarLan92}) gives a
precise statement of the result.

\begin{theorem}
[complete overdamping]Suppose $\theta=0$ (i.e., a non-gyroscopic system) and
$\delta_{R}=1$ (i.e., $R$ has full rank). Then there exists a $\beta^{\prime
}>0$ such that if $\beta>\beta^{\prime}$ then all the eigenmodes of the
Lagrangian system with equations of motion (\ref{sintro1}) are overdamped. In
particular, we can take%
\[
\beta^{\prime}=2\frac{\omega_{\max}}{b_{\min}},
\]
where
\[
\omega_{\max}=\sqrt{\max\sigma\left(  \alpha^{-1}\eta\right)  },\text{
}b_{\min}=\min\sigma\left(  \alpha^{-1}R\right)
\]
and $\sigma\left(  M\right)  $ denotes spectrum of a square matrix $M$, i.e.,
the set of its eigenvalues.
\end{theorem}

\begin{remark}
\label{RemCompOD}Although it may not be immediately obvious, the spectrums
$\sigma\left(  \alpha^{-1}\eta\right)  $ and $\sigma\left(  \alpha
^{-1}R\right)  $ are subsets of $[0,\infty)$ since $\alpha^{-1}\eta$ and
$\alpha^{-1}R$ are similar to positive semidefinite matrices:
\begin{gather*}
\alpha^{-1}\eta=\sqrt{\alpha}^{-1}\left(  \sqrt{\alpha}^{-1}\eta\sqrt{\alpha
}^{-1}\right)  \sqrt{\alpha},\text{ }\alpha^{-1}R=\sqrt{\alpha}^{-1}\left(
\sqrt{\alpha}^{-1}R\sqrt{\alpha}^{-1}\right)  \sqrt{\alpha},\\
\sqrt{\alpha}^{-1}\eta\sqrt{\alpha}^{-1}\geq0,\sqrt{\alpha}^{-1}R\sqrt{\alpha
}^{-1}\geq0.
\end{gather*}
In particular, this implies $\beta^{\prime}\geq0$ in the previous theorem.
\end{remark}

The next example, which we will discuss in more detail later in this paper
(see Example \ref{nodex}), shows that, unlike for non-gyroscopic systems, in
gyroscopic systems it is entirely possible that all the modes can be
underdamped when the loss fraction condition (\ref{sintro4a}) fails to be satisfied.

\begin{example}
[no overdamping]If $\alpha=\eta=R=\mathbf{1}$ (where $\mathbf{1}$ denotes the
$N\times N$ identity matrix and hence $\delta_{R}=1$) and $0\not \in
\sigma\left(  \theta\right)  $ then all the eigenmodes of the Lagrangian
system with equations of motion (\ref{sintro1}) are underdamped for $\beta>0$.
\end{example}

Notice that in the mentioned examples and results the loss fraction condition
(\ref{sintro4a}) is not satisfied, namely $\delta_{R}=1$, and hence the
dissipative (Lagrangian) system consists only of lossy components. But the
question we are most interested in is: what overdamping phenomena can occur
for a two-component composite system with a lossy and a lossless component
when the loss fraction condition (\ref{sintro4a}), that is, $0<\delta_{R}<1$,
is satisfied? The answer is that (generically) some of the modes of the system
will be overdamped and some will be underdamped, and we refer to this
phenomenon as \emph{selective overdamping}. In the next subsection, we will
give a brief description of this phenomenon along with our motivation for its
study, and in the subsection afterwards give an overview of our main results.

\subsection{Motivation\label{sinmot}}

An important motivation for our studies of two component dissipative
gyroscopic system is the development of a theory of broadband absorption
suppression in magnetic composites. Such a theory, we believe, can provide
guiding principles for the design of broadband low-loss magnetic composites
with functionality comparable to bulk magnetic materials. The development of
theory requires a deeper understanding of the interplay between losses and
magnetism manifested as gyroscopic effects. There are numerious applications
of low loss magnetic materials. For instance, they are crucial components in
many microwave, infrared, and optical devices \cite{FMW55}, \cite{Hogan52},
\cite{ILB13}, \cite{Pozer12}, \cite{ZveKot97}. Detrimental to the performance
of many such devices are the high losses associated with the magnetic
materials in frequency ranges of interest, \cite{Hogan52}, \cite{ZveKot97},
and this is a major problem with many natural and synthetic magnetic materials.

The discussion above raises a question if such broadband absorption
suppression in composites even possible? Quite remarkably, the answer is yes.
This result was firmly established in \cite{FigVit8}, \cite{FigVit10},
\cite{SmCh11}, \cite{SmCa13}. For instance, in \cite{FigVit8} an example was
given of a two-component dielectric medium composed of a high-loss and
lossless components, namely, a magnetophotonic crystal (MPCs) consisting of a
finite stack of alternating lossy magnetic and lossless dielectric layers.
They showed that the magnetic composite could reduce the absorption (losses)
by two orders of magnitude in the chosen frequency range compared to those of
the uniform bulk magnetic material while simultaneously enhancing one of its
desired magnetic properties, namely, nonreciprocal Faraday rotation. That
example demonstrated that it is possible to design a composite material/system
which can have a desired property comparable with a naturally occurring bulk
substance but with significantly reduced losses.

In addition to this, an interesting and rather counterintuitive idea arose,
which was first introduced in \cite{FigVit8}, and also recently noticed
independently in \cite{IOKS11} for MPCs. It is the idea that reduction of
losses in the magnetic composite and enhancement of the magnetic
properties/functionality might actually be more substantial when the lossy
magnetic component is replaced by another with even higher losses.

What is the origin of that seemingly counterintuitive behavior in composites?
In order to understand the general mechanism for this behavior, we developed
in \cite{FigWel1} a model, based on the linear response theory from
\cite{FigSch1} and \cite{FigSch2}, for two-component composite systems with a
high-loss and lossless component and introduced in \cite{FigWel2} a Lagrangian
framework, based on the Lagrangian-Hamiltonian formulation of classical
mechanics, in order to account for the physical properties of the composite
system. We showed that for such composite systems the losses of the entire
system become small provided that the lossy component is sufficiently lossy.
This behavior can be explained by two important phenomena, namely, the
\emph{modal dichotomy} and \emph{overdamping}.

As the focus of this paper is on the study of these two phenomena in
gyroscopic-dissipative systems, we will provide a brief explaination from our
studies in \cite{FigWel1}, \cite{FigWel2} on how these phenomena contribute to
the loss suppression. We introduce first a dimensionless loss parameter
$\beta\geq0$ which scales the dissipation in the lossy component of the
system. We consider then the system eigenmodes, i.e., the states of the system
in the absence of external forces with exponential time dependency of the form
$e^{-\mathrm{i}\zeta t}=e^{-\frac{t}{T}}e^{-\mathrm{i}\operatorname{Re}\zeta
t}$, where $\operatorname{Re}\zeta=\omega$ is the frequency,
$-\operatorname{Im}\zeta=\gamma\geq0$ is the damping factor, and $T_{\zeta
}=\frac{1}{-\operatorname{Im}\zeta}$ is the relaxation time. To any such mode
is associated its quality factor (Q-factor) $Q_{\zeta}=-\frac{1}{2}%
\frac{\left\vert \omega\right\vert }{\gamma}$, which is an important figure of
merit that helps to characterize the performance of the dissipative composite
system. Now as the losses in the lossy component of the composite system
become sufficiently large, i.e., $\beta\gg1$, the entire set of eigenmodes of
the composite system splits into two classes, \emph{high-loss} and
\emph{low-loss modes}, based on their dissipative properties. We refer to this
phenomenon as the \emph{modal dichotomy}. One important feature of this
dichotomy is that the high-loss modes decay exponentially in time with both an
extremely small relaxation time $T_{\zeta}$ and Q-factor $Q_{\zeta}$ that
decrease with $T_{\zeta}\rightarrow0$ and $Q_{\zeta}\rightarrow0$ as
$\beta\rightarrow\infty$. On the other hand, the low-loss modes have an
extremely large relaxation time $T_{\zeta}$ which increases with $T_{\zeta
}\rightarrow\infty$ as $\beta\rightarrow\infty$, whereas the Q-factor
$Q_{\zeta}$ either decreases or increases with $Q_{\zeta}\rightarrow0$ or
$Q_{\zeta}\rightarrow\infty$, respectively, as $\beta\rightarrow\infty$ (as to
this behavior of the Q-factor and whether such low-loss high-Q modes even
exist, we address this in the next paragraph). Moreover, in Lagrangian
systems, when the loss of the high-loss component exceeds a finite critical
value, i.e., $\beta>\beta_{0}$, the frequencies of the all the high-loss
eigenmodes become exactly zero, i.e., $\operatorname{Re}\zeta=0$ for
$\beta>\beta_{0}$, a phenomenon known as overdamping. Consequently, when the
composite is excited by external forces at frequencies ranges well separated
from zero, the high-loss modes hardly respond to these excitations because
they are overdamped with extremely small relaxation time, and hence do not
contribute much to the entire composite losses.

This analysis leads to the important question: do such high-Q modes even exist
in systems with a high-loss component? As discussed in \cite{FigWel1},
\cite{FigWel2} the answer is yes, but not always and composites with a
high-loss component $\beta\gg1$ are key to selectively suppressing low-Q modes
and enhancing high-Q modes. More precisely, one of the main result of our
studies in \cite{FigWel2} is that a rather universal phenomenon, called
\emph{selective overdamping}, occurs for non-gyroscopic composite systems
whenever the lossy component of the composite is sufficiently lossy $\beta
\gg1$. In fact, we proved in \cite[Theorems 25 and 26]{FigWel2} that for a
Lagrangian system governed by evolution equations (\ref{sintro1}), that it
will occur for $\beta\gg1$ whenever $\theta=0$, $0<\delta_{R}<1$, and
$\operatorname{Ker}\eta\cap\operatorname{Ker}R=\left\{  0\right\}  $. The term
\textquotedblleft\emph{selective}\textquotedblright\ was used to refer to the
fact that only a fraction, namely, the loss fraction $\delta_{R}>0$, of the
system's eigenmodes are overdamped, specifically, all the high-loss modes and
an equal number of low-loss modes, whereas the remaining \emph{positive
fraction, namely, }$1-$\emph{ }$\delta_{R}>0$\emph{, of modes are low-loss
oscillatory modes (i.e, the underdamped modes) with high quality factor that
actually becomes higher the more lossy the lossy component becomes in the
system}.

Since overdamping phenomenon has a potential to be a mechanism for significant
broadband absorption suppression in composites, we are motivated to analyze
and understand it better, especially in gyroscopic-dissipative systems. It
turns out that as the losses in the lossy component increase the overdamped
high-loss modes are more suppressed while all the low-loss oscillatory modes
are more enhanced with increasingly high quality factor. \emph{This provides a
mechanism for selective enhancement of these high quality factor, low-loss
oscillatory modes (the underdamped modes) and selective suppression of the
high-loss non-oscillatory modes}.

\subsection{Overview of results\label{sinresults}}

The main goal of this paper is to understand if the selective overdamping
phenomenon can occur in gyroscopic systems, and if so whether it as universal
of a phenomenon as for non-gyroscopic systems. One of the major achievements
of this paper, we think, is that we have found sufficient conditions for
overdamping to occur for the high-loss modes, have derived uppper bounds on
the amount of loss required, and have given estimates on the frequencies,
damping factors, and Q-factors for the underdamped modes. In addition to that,
a simple example is given in Section \ref{scircuits} of an electric circuit
with a resistor and a gyrator which illustrates our ideas, methods, and
results both analytically and numerically.

In this section, we will give an overview of the main results of this paper,
which are formulated precisely and proven in Section \ref{smainr}. In
particular, in Section \ref{sevmd} on the modal dichotomy we have Theorems
\ref{tmddI} and \ref{tllspmd} along with their corollaries \ref{cmddI} and
\ref{cllspmdd}. In Section \ref{smdhlr} on the asymptotics of the eigenmodes
in the high-loss regime (i.e., as $\beta\rightarrow\infty$) including the
asymptotics on the frequencies, damping factors, and quality factors, we have
Theorems \ref{tmddII} and \ref{tmqf} and Corollary \ref{cdasym} along with
Propositions \ref{pspredl} and \ref{pspredl2} from Section \ref{stvspen}. And
in Section \ref{sodgd} on overdamping phenomenon we have Theorem \ref{tsogc}
and Corollary \ref{csogc} on selective overdamping in the generic case (along
with Corollary \ref{cllspmdd} in Sec. \ref{sevmd}). In the nongeneric case, we
have an interesting example, Example \ref{sodngc}, which shows an extreme case
of what can happen for dissipative systems which are gyroscopic (i.e.,
$\theta\not =0$).

We will begin by introducing some notation. After this we will discuss the
modal dichtotomy in Section \ref{sbsbsecmd} and then, in Section
\ref{sbsbsecsod}, conclude with a description of the overdamping phenomenon in
terms of the modal dichotomy. Consider the Lagrangian system with equations of
motion (\ref{sintro1}) and recall the definitions of the frequency $\omega$
and damping factor $\gamma$ in (\ref{sintro5a}) of an eigenmode (\ref{sintro5}%
) of this system. Let $\omega_{\max}$ and $\omega_{\min}$ denote the maximum
and minimum positive frequencies, respectively, of the eigenmodes of system
(\ref{sintro1}) with $\beta=0$. For the system (\ref{sintro1}) with $\beta=1$,
$\theta=0$, and $\eta=0$, denote the smallest of the nonzero damping factors
of the eigenmodes by $b_{\min}$. As these terms play a key role in describing
the modal dichotomy and overdamping phenomena, we provide a way to calculate
them (as described in Sections \ref{sinframe} and \ref{stvspen}) using
spectral theory:%
\begin{gather}
\omega_{\max}=\max\left\{  \omega\in(0,\infty):\det\left(  \omega^{2}%
\alpha+2\omega\mathrm{i}\theta-\eta\right)  =0\right\}  ,\text{ }%
\label{sintmr00}\\
\omega_{\min}=\min\left\{  \omega\in(0,\infty):\det\left(  \omega^{2}%
\alpha+2\omega\mathrm{i}\theta-\eta\right)  =0\right\}  ,\label{sintmr01}\\
b_{\min}=\min\left[  \sigma\left(  \alpha^{-1}R\right)  \setminus\left\{
0\right\}  \right]  >0. \label{sintmr02}%
\end{gather}
Next, to describe our results we assume that the following condition holds:

\begin{condition}
The duality condition is the assumption that%
\begin{equation}
\eta>0\text{.} \label{cnddl}%
\end{equation}

\end{condition}

The reason this is called the duality condition is that under this condition
there is a \textquotedblleft dual\textquotedblright\ Lagrangian system to the
Lagrangian system with evolution equations (\ref{sintro1}), which has the same
evolution equations except $\alpha$ and $\eta$ are interchanged, i.e., the
equations of motion (\ref{dradis2a}).

\begin{remark}
[duality]This \textquotedblleft duality\textquotedblright\ is discussed in
more detail in Section \ref{subsDualSys}. Its importance lies in the fact that
it allows us to achieve more complete and sharper results in describing the
modal dichotomy (see Theorems \ref{tmddI}, \ref{tmqf} and Corollaries
\ref{cmddI}, \ref{cllspmdd}, and \ref{cdasym}) and overdamping (see
Corollaries \ref{csogc} and \ref{cbeta0}). This is a consequence of the
relationship between the eigenmodes (and their quality factors) of the
Lagrangian system and its dual [cf. (\ref{dradis4_2}) and (\ref{dradis4_3})].
Our main results on this relationship is contained in Propositions
\ref{ppfspdl}, \ref{pspredl} and \ref{pspredl2} which connects the spectral
theory associated with the eigenmodes of each system together.
\end{remark}

For this dual Lagrangian system (\ref{dradis2a}), we define $\omega_{\max
}^{\flat}$ and $b_{\min}^{\flat}$ similar to $\omega_{\max}$ and $b_{\min}$ as
follows: $\omega_{\max}^{\flat}$ is the maximum positive frequency of the
eigenmodes of (\ref{dradis2a}) with $\beta=0$ and $b_{\min}^{\flat}$ is the
smallest nonzero damping factor\ of the eigenmodes of (\ref{dradis2a}) with
$\beta=1$, $\theta=0$, and $\alpha=0$. In particular, it follows from
Proposition \ref{pspredl} that%
\begin{gather}
\omega_{\max}^{\flat}=\frac{1}{\omega_{\min}}=\max\left\{  \omega\in
(0,\infty):\det\left(  \omega^{2}\eta+2\omega\mathrm{i}\theta-\alpha\right)
=0\right\}  ,\label{sintmr0a}\\
b_{\min}^{\flat}=\min\left[  \sigma\left(  \eta^{-1}R\right)  \setminus
\left\{  0\right\}  \right]  >0. \label{sintmr0b}%
\end{gather}
We next define the decreasing functions, $y=c\left(  \beta\right)  $ and its
inverse $\beta=c^{-1}\left(  y\right)  $, by%
\begin{gather}
c\left(  \beta\right)  =\left(  \frac{2\omega_{\max}^{2}}{b_{\min}}\right)
\left[  \beta-\left(  2\frac{\omega_{\max}}{b_{\min}}\right)  \right]
^{-1},\ \ \text{for }\beta>\left(  2\frac{\omega_{\max}}{b_{\min}}\right)
,\label{sintmr1}\\
c^{-1}\left(  y\right)  =\left(  \frac{2\omega_{\max}^{2}}{b_{\min}}\right)
y^{-1}+\left(  2\frac{\omega_{\max}}{b_{\min}}\right)  ,\ \ \text{for }y>0
\label{sintmr1a}%
\end{gather}
and introduce the same functions for the dual Lagrangian system%
\begin{gather}
c^{\flat}\left(  \beta\right)  =\left[  2\frac{\left(  \omega_{\max}^{\flat
}\right)  ^{2}}{b_{\min}^{\flat}}\right]  \left[  \beta-\left(  2\frac
{\omega_{\max}^{\flat}}{b_{\min}^{\flat}}\right)  \right]  ^{-1},\ \ \text{for
}\beta>\left(  2\frac{\omega_{\max}^{\flat}}{b_{\min}^{\flat}}\right)
,\label{sintmrd1}\\
\left(  c^{\flat}\right)  ^{-1}\left(  y\right)  =\left[  2\frac{\left(
\omega_{\max}^{\flat}\right)  ^{2}}{b_{\min}^{\flat}}\right]  y^{-1}+\left(
2\frac{\omega_{\max}^{\flat}}{b_{\min}^{\flat}}\right)  ,\ \ \text{for }y>0.
\label{sintmrd1a}%
\end{gather}

Finally, the (nonzero) rank $N_{R}$ of the $N\times N$ matrix $R$, i.e.,
\begin{equation}
N_{R}=\operatorname{rank}R>0, \label{sintrmr1c}%
\end{equation}
plays a key role in the following description of our main results as does the
\emph{configuration space} $\mathbb{M}\left(  \beta\right)  $ and the
corresponding \emph{phase space} $\mathbb{V}\left(  \beta\right)  $ of
(\ref{sintro1}) for each $\beta\geq0$, i.e.,%
\begin{gather}
\mathbb{M}\left(  \beta\right)  =\left\{  Q:Q\text{ is a solution of
(\ref{sintro1})}\right\}  \text{\quad(configuration space),}\label{sintmr1d}\\
\mathbb{V}\left(  \beta\right)  =\left\{  \left[  Q,\dot{Q}\right]
^{\mathrm{T}}:Q\text{ is a solution of (\ref{sintro1}) and }\dot{Q}%
=\partial_{t}Q\right\}  \text{\quad(phase space),}\label{sintmr1e}\\
\text{where }Q=Q\left(  t,\beta\right) \nonumber
\end{gather}
along with the $\mathbb{M}\left(  \beta\right)  $-eigenmodes and the
$\mathbb{V}\left(  \beta\right)  $-eigenmodes, i.e.,%
\begin{align}
\mathbb{M}\left(  \beta\right)  \text{-eigenmode}  &  \text{: }Q\in
\mathbb{M}\left(  \beta\right)  \text{ of the form (\ref{sintro5}%
),}\label{sintmr1f}\\
\mathbb{V}\left(  \beta\right)  \text{-eigenmode}  &  \text{: }\left[
Q,\dot{Q}\right]  ^{\mathrm{T}}\in\mathbb{V}\left(  \beta\right)  \text{ with
}Q\text{ an }\mathbb{M}\left(  \beta\right)  \text{-eigenmode.}
\label{sintmr1g}%
\end{align}

\begin{remark}
[change-of-variables]Although it is simpler and most perspicuous to phrase our
main results in this overview in terms of the configuration space
$\mathbb{M}\left(  \beta\right)  $ and the phase space $\mathbb{V}\left(
\beta\right)  $ for the Lagrangian system with equations of motion
(\ref{sintro1}) (a system of linear second-order ODEs), it is actually better
(in terms of the analysis and precision in the statement of results in Section
\ref{smainr}) to first make a change-of-variables (see \ref{chvar}) from the
generalized coordinates and generalized velocities, i.e., $\left[  Q,\dot
{Q}\right]  ^{\mathrm{T}}$, to a new variable $v$ which satisfies the
canonical evolution equations (\ref{ceveqs}) (a system of linear first-order
ODEs). The evolution of this canonical system is governed by a contraction
semigroup $e^{-\mathrm{i}A\left(  \beta\right)  t}$ in which the (system)
operator $A\left(  \beta\right)  $ is an analytic matrix-valued function of
the loss parameter $\beta$ with the fundamental properties (\ref{ceveqs1}) for
$\beta\geq0$. The key advantage of this is it allows us to study the modal
dichotomy and the overdamping phenomenon using linear perturbation theory by
considering the standard eigenvalue problem (\ref{sevp}) of $A\left(
\beta\right)  $ and the splitting of its spectrum as a function of $\beta$. A
brief description of this framework that we use to study the modal dichotomy,
overdamping phenomena, and the associated spectral problems is discussed below
in Section \ref{sinframe}.
\end{remark}

\subsubsection{The modal dichotomy\label{sbsbsecmd}}

The\emph{ phenomenon of modal dichotomy} can be described, as we have done
below, as occurring in \emph{four stages (i)-(iv)} with increasing $\beta$. To
begin with, the phase space $\mathbb{V}\left(  \beta\right)  $ of
(\ref{sintro1}) is a $2N$-dimensional vector space over $%
\mathbb{C}
$ for each $\beta\geq0$. Moreover, $\mathbb{V}\left(  \beta\right)  $ is
spanned by a basis of $\mathbb{V}\left(  \beta\right)  $-eigenmodes for every
$\beta$ with only a finite number of exceptions (a consequence of Proposition
\ref{pdiagsop} and Corollary \ref{cpfsp}).

Now in the description of each stage (i)-(iii) we provide bounds on the
frequencies, damping factors, and quality factors (Q-factor) for the
eigenmodes of the Lagrangian system (\ref{sintro1}) with stage (iv) providing
a description of their asymptotics as the loss parameter $\beta\rightarrow
\infty$. The main point of these bounds is that it allows us at each of these
stages to give the following dissipative characterization of the splitting of
the phase space $\mathbb{V}\left(  \beta\right)  $: (i) into the direct sum of
a high-loss subspace $\mathbb{V}_{h\ell}\left(  \beta\right)  $, whose
$\mathbb{V}\left(  \beta\right)  $-eigenmodes in it will have large damping
factors and low Q-factors, and its complement $\mathbb{V}_{\ell\ell}\left(
\beta\right)  $; (ii) the splitting of $\mathbb{V}_{\ell\ell}\left(
\beta\right)  $ into the direct sum of a low-loss/low-Q subspace
$\mathbb{V}_{\ell\ell,0}\left(  \beta\right)  $, whose $\mathbb{V}\left(
\beta\right)  $-eigenmodes in it will have small damping factors and low
Q-factors, and its complement $\mathbb{V}_{\ell\ell,1}\left(  \beta\right)  $;
(iii) the low-loss/high-Q subspace $\mathbb{V}_{\ell\ell,1}\left(
\beta\right)  $, whose $\mathbb{V}\left(  \beta\right)  $-eigenmodes in it
will have small damping factors and high Q-factors; (iv) a basis of
$\mathbb{V}\left(  \beta\right)  $-eigenmodes in each of these subspaces and
the asymptotics for their frequencies, damping factors, and Q-factors as
$\beta\rightarrow\infty$.

Let us now describe these four stages of the modal dichotomy more precisely
using quantities defined in (\ref{sintmr00})-(\ref{sintmr02}), (\ref{sintmr0a}%
), (\ref{sintmr0b}), and (\ref{sintmr1})-(\ref{sintrmr1c}).

\textbf{(i)} In the first stage of modal dichotomy (Theorem \ref{tmdic}), if
$\beta>2\frac{\omega_{\max}}{b_{\min}}$ then the space $\mathbb{V}\left(
\beta\right)  $ splits into the direct sum of subspaces
\[
\mathbb{V}\left(  \beta\right)  =\mathbb{V}_{h\ell}\left(  \beta\right)
\oplus\mathbb{V}_{\ell\ell}\left(  \beta\right)
\]
which have dimensions%
\[
\dim\mathbb{V}_{h\ell}\left(  \beta\right)  =N_{R}\text{, }\dim\mathbb{V}%
_{\ell\ell}\left(  \beta\right)  =2N-N_{R}%
\]
and the properties that for any $\mathbb{M}\left(  \beta\right)  $-eigenmode
$Q=Q\left(  t\right)  =qe^{-\mathrm{i}\zeta t}$ of (\ref{sintro1}), the
$\mathbb{V}\left(  \beta\right)  $-eigenmode $\left[  Q,\dot{Q}\right]
^{\mathrm{T}}$ belongs to either $\mathbb{V}_{h\ell}\left(  \beta\right)  $ or
$\mathbb{V}_{\ell\ell}\left(  \beta\right)  $ with the following estimates
holding (by Corollary \ref{cmdic}):\bigskip

\textbf{a.} If $\left[  Q,\dot{Q}\right]  ^{\mathrm{T}}\in\mathbb{V}_{h\ell
}\left(  \beta\right)  $ then $-\operatorname{Im}\zeta\geq\beta b_{\min
}-\omega_{\max}$, $\left\vert \operatorname{Re}\zeta\right\vert \leq
\omega_{\max}$, and $0\leq Q_{\zeta}\leq\frac{1}{2}\frac{\omega_{\max}}{\beta
b_{\min}-\omega_{\max}}<\frac{1}{2}$.\bigskip

\textbf{b.} If $\left[  Q,\dot{Q}\right]  ^{\mathrm{T}}\in\mathbb{V}_{\ell
\ell}\left(  \beta\right)  $ then $0\leq-\operatorname{Im}\zeta\leq
\omega_{\max}$ and $\left\vert \operatorname{Re}\zeta\right\vert \leq
\omega_{\max}$.\bigskip

\textbf{(ii)} In the second stage (Theorem \ref{tmddI}), if $\beta
>\max\left\{  2\frac{\omega_{\max}}{b_{\min}},2\frac{\omega_{\max}^{\flat}%
}{b_{\min}^{\flat}}\right\}  $ then $\left(  2N-N_{R}\right)  $-dimensional
space $\mathbb{V}_{\ell\ell}\left(  \beta\right)  $ splits into the direct sum%
\[
\mathbb{V}_{\ell\ell}\left(  \beta\right)  =\mathbb{V}_{\ell\ell,0}\left(
\beta\right)  \oplus\mathbb{V}_{\ell\ell,1}\left(  \beta\right)
\]
with dimensions
\[
\dim\mathbb{V}_{\ell\ell,0}\left(  \beta\right)  =N_{R}\text{, }\dim
\mathbb{V}_{\ell\ell,1}\left(  \beta\right)  =2N-2N_{R}\text{.}%
\]
Furthermore, for any $\mathbb{M}\left(  \beta\right)  $-eigenmode $Q=Q\left(
t\right)  =qe^{-\mathrm{i}\zeta t}$ of (\ref{sintro1}), if the $\mathbb{V}%
\left(  \beta\right)  $-eigenmode $\left[  Q,\dot{Q}\right]  ^{\mathrm{T}}$
belongs to $\mathbb{V}_{\ell\ell}\left(  \beta\right)  $ then $\left[
Q,\dot{Q}\right]  ^{\mathrm{T}}\in\mathbb{V}_{\ell\ell,0}\left(  \beta\right)
$ or $\left[  Q,\dot{Q}\right]  ^{\mathrm{T}}\in\mathbb{V}_{\ell\ell,1}\left(
\beta\right)  $ with the following estimates holding (Theorem \ref{tmddI} and
Corollary \ref{cmddI}):\bigskip

\textbf{a.} If $\left[  Q,\dot{Q}\right]  ^{\mathrm{T}}\in\mathbb{V}_{\ell
\ell,0}\left(  \beta\right)  $ then $\left\vert \zeta\right\vert \leq\frac
{1}{\beta b_{\min}^{\flat}-\omega_{\max}^{\flat}}<\omega_{\min}$ and $0\leq
Q_{\zeta}\leq\frac{1}{2}\frac{\omega_{\max}^{\flat}}{\beta b_{\min}^{\flat
}-\omega_{\max}^{\flat}}<\frac{1}{2}$.\bigskip

\textbf{b.} If $\left[  Q,\dot{Q}\right]  ^{\mathrm{T}}\in\mathbb{V}_{\ell
\ell,1}\left(  \beta\right)  $ then $\omega_{\min}\leq\left\vert
\zeta\right\vert \leq\omega_{\max}$.\bigskip

\textbf{(iii)} In the third stage (Theorem \ref{tllspmd}), either $N_{R}=N$
(i.e., $R$ has full rank) and $\mathbb{V}_{\ell\ell,1}\left(  \beta\right)
=\left\{  0\right\}  $ or $N_{R}<N$ (i.e., $R$ is rank deficient) and there
exists an $\rho_{\min}>0$ such that $c^{-1}\left(  \frac{\rho_{\min}}%
{2}\right)  >2\frac{\omega_{\max}}{b_{\min}}$ and if $\beta>\max\left\{
c^{-1}\left(  \frac{\rho_{\min}}{2}\right)  ,2\frac{\omega_{\max}^{\flat}%
}{b_{\min}^{\flat}}\right\}  $ then for any $\mathbb{M}\left(  \beta\right)
$-eigenmode $Q=Q\left(  t\right)  =qe^{-\mathrm{i}\zeta t}$ of (\ref{sintro1})
whose $\mathbb{V}\left(  \beta\right)  $-eigenmode $\left[  Q,\dot{Q}\right]
^{\mathrm{T}}\ $belongs to $\mathbb{V}_{\ell\ell}\left(  \beta\right)  $ will
have either $\left[  Q,\dot{Q}\right]  ^{\mathrm{T}}\in\mathbb{V}_{\ell\ell
,0}\left(  \beta\right)  $ or $\left[  Q,\dot{Q}\right]  ^{\mathrm{T}}%
\in\mathbb{V}_{\ell\ell,1}\left(  \beta\right)  $ and the following estimates
hold (Corollary \ref{cllspmd}):\bigskip

\textbf{a.} If $\left[  Q,\dot{Q}\right]  ^{\mathrm{T}}\in\mathbb{V}_{\ell
\ell,0}\left(  \beta\right)  $ then $0\leq-\operatorname{Im}\zeta\leq c\left(
\beta\right)  $ and $\left\vert \operatorname{Re}\zeta\right\vert \leq
c\left(  \beta\right)  $.\bigskip

\textbf{b.} If $\left[  Q,\dot{Q}\right]  ^{\mathrm{T}}\in\mathbb{V}_{\ell
\ell,1}\left(  \beta\right)  $ then $0\leq-\operatorname{Im}\zeta\leq c\left(
\beta\right)  $, $\left\vert \operatorname{Re}\zeta\right\vert \geq\rho_{\min
}-c\left(  \beta\right)  $, and $Q_{\zeta}\geq\frac{1}{2}\frac{\rho_{\min
}-c\left(  \beta\right)  }{c\left(  \beta\right)  }>\frac{1}{2}$. In
particular, $\operatorname{Re}\zeta\not =0$ and all the $\mathbb{M}\left(
\beta\right)  $-eigenmodes in $\mathbb{V}_{\ell\ell,1}\left(  \beta\right)  $
are underdamped.\bigskip

\textbf{(iv)} In the fourth stage (Theorems \ref{tmdicII}, \ref{tmddII},
Corollary \ref{cdasym}, Section \ref{smdhlr}, and Propositions \ref{pspredl},
\ref{pspredl2}), if $\beta$ is sufficiently large (i.e., $\beta\gg1$) then the
space $\mathbb{V}\left(  \beta\right)  $ is spanned by a basis of
$\mathbb{V}\left(  \beta\right)  $-eigenmodes $\left[  Q_{j},\partial_{t}%
Q_{j}\right]  ^{\mathrm{T}}$, where $Q_{j}=Q_{j}\left(  t,\beta\right)
=q_{j}\left(  \beta\right)  e^{-\mathrm{i}\zeta_{j}\left(  \beta\right)  t}$,
$j=1,\ldots,2N$ which split into two distinct classes%
\begin{align*}
\text{high-loss}  &  \text{: }Q_{j}(\beta),\text{ }1\leq j\leq N_{R};\\
\text{low-loss}  &  \text{: }Q_{j}(\beta),\text{ }N_{R}+1\leq j\leq2N
\end{align*}
with the following properties:\bigskip

\textbf{a.} These $\mathbb{V}\left(  \beta\right)  $-eigenmode split the space
$\mathbb{V}\left(  \beta\right)  $ into the direct sum of subspaces
\[
\mathbb{V}\left(  \beta\right)  =\mathbb{V}_{h\ell}\left(  \beta\right)
\oplus\mathbb{V}_{\ell\ell,0}\left(  \beta\right)  \oplus\mathbb{V}_{\ell
\ell,1}\left(  \beta\right)
\]
in which
\begin{align*}
\mathbb{V}_{h\ell}\left(  \beta\right)   &  =\operatorname*{span}\left\{
\left[  Q_{j},\partial_{t}Q_{j}\right]  ^{\mathrm{T}}:\text{ }1\leq j\leq
N_{R}\right\}  ,\\
\mathbb{V}_{\ell\ell,0}\left(  \beta\right)   &  =\operatorname*{span}\left\{
\left[  Q_{j},\partial_{t}Q_{j}\right]  ^{\mathrm{T}}:\text{ }N_{R}+1\leq
j\leq2N_{R}\right\}  ,\\
\mathbb{V}_{\ell\ell,1}\left(  \beta\right)   &  =\operatorname*{span}\left\{
\left[  Q_{j},\partial_{t}Q_{j}\right]  ^{\mathrm{T}}:\text{ }2N_{R}+1\leq
j\leq2N\right\}  ,
\end{align*}
where $\operatorname*{span}\left\{  \cdot\right\}  $ denotes the span of a set
$\left\{  \cdot\right\}  $, i.e., all linear combinations of elements of
$\left\{  \cdot\right\}  $ over $%
\mathbb{C}
$.\bigskip

\textbf{b.} The frequencies $\operatorname{Re}\zeta_{j}\left(  \beta\right)
$, damping factors $-\operatorname{Im}\zeta_{j}\left(  \beta\right)  $, and
Q-factors $Q_{\zeta_{j}\left(  \beta\right)  }$ have the following asymptotic
expansions as $\beta\rightarrow\infty$:%
\begin{gather*}
\text{high-loss:}-\operatorname{Im}\zeta_{j}\left(  \beta\right)  =b_{j}%
\beta+O\left(  \beta^{-1}\right)  ,\text{ }\operatorname{Re}\zeta_{j}\left(
\beta\right)  =\rho_{j}+O\left(  \beta^{-2}\right)  ,\\
Q_{\zeta_{j}\left(  \beta\right)  }=\frac{\left\vert \rho_{j}\right\vert
}{b_{j}}\beta^{-1}+O\left(  \beta^{-3}\right)  \text{, for }1\leq j\leq N_{R},
\end{gather*}
where $0<b_{1}\leq\cdots\leq b_{N_{R}}$ are all the nonzero eigenvalues of
$\alpha^{-1}R$ listed in increasing order and repeated according to their
multiplicities;
\begin{gather*}
\text{low-loss, low-Q:}-\operatorname{Im}\zeta_{j}\left(  \beta\right)
=\frac{1}{b_{j-N_{R}}^{\flat}}\beta^{-1}+O\left(  \beta^{-3}\right)  ,\text{
}\operatorname{Re}\zeta_{j}\left(  \beta\right)  =O\left(  \beta^{-2}\right)
,\\
Q_{\zeta_{j}\left(  \beta\right)  }=O\left(  \beta^{-1}\right)  \text{, for
}N_{R}+1\leq j\leq2N_{R},
\end{gather*}
where $0<b_{1}^{\flat}\leq\cdots\leq b_{N_{R}}^{\flat}$ are all the nonzero
eigenvalues of $\eta^{-1}R$ listed in increasing order and repeated according
to their multiplicities;%
\begin{gather*}
\text{low-loss, high-Q:}-\operatorname{Im}\zeta_{j}\left(  \beta\right)
=d_{j}\beta^{-1}+O\left(  \beta^{-3}\right)  ,\text{ }\operatorname{Re}%
\zeta_{j}\left(  \beta\right)  =\rho_{j}+O\left(  \beta^{-2}\right)  ,\\
Q_{\zeta_{j}\left(  \beta\right)  }=\frac{\left\vert \rho_{j}\right\vert
}{d_{j}}\beta+O\left(  \beta^{-1}\right)  \text{ (provided }d_{j}%
\not =0\text{), for }2N_{R}+1\leq j\leq2N,
\end{gather*}
where the limiting frequencies $\rho_{j}$, $2N_{R}+1\leq j\leq2N$ are all the
nonzero eigenvalues of a self-adjoint operator $\Omega_{1}$, defined in
Proposition \ref{pspredl2} [see also (\ref{pod7})], and repeated according to
their multiplicities. \bigskip

In the third stage (iii) of the modal dichotomy described above, the value
$\rho_{\min}$ can be taken to be
\begin{equation}
\rho_{\min}=\min\left\{  \rho_{j}:2N_{R}+1\leq j\leq2N,\text{ }\rho
_{j}>0\right\}  . \label{sintmr0g0}%
\end{equation}
Now define the value $\rho_{\max}$ defined by%
\begin{equation}
\rho_{\max}=\max\left\{  \rho_{j}:2N_{R}+1\leq j\leq2N\right\}
\label{sintmr0g1}%
\end{equation}
Also define $\rho_{\min}^{\flat}$ and $\rho_{\max}^{\flat}$ similarly for the
dual Lagrangian system (\ref{dradis2a}) as we defined $\rho_{\min}$ and
$\rho_{\max}$ above for the Lagrangian system (\ref{sintro1}). It follows from
Proposition \ref{pspredl2} and the remark below that%
\begin{equation}
\rho_{\min}^{\flat}=\rho_{\max}^{-1}\text{, }\rho_{\max}^{\flat}=\rho_{\min
}^{-1}. \label{sintmr0g2}%
\end{equation}

\begin{remark}
[alternative spectral characterization]\label{raltspc}Proposition
\ref{pspredl2} (which complements Proposition \ref{pspredl}) in Sec.
\ref{smdhlr} and the perturbation analysis described in Sec. \ref{smdhlr}
gives an important alternative spectral characterization of the limiting
frequencies $\rho_{j}$, $2N_{R}+1\leq j\leq2N$ of the low-loss, high-Q modes,
which can be used to calculate\ explicitly these values as we have done, for
instance, in Section \ref{scircuits} for an electric circuit example.
Moreover, Proposition \ref{pspredl2} together with Remark \ref{pspredl2} gives
an interpretation (within the Lagrangian framework introduced in
\cite{FigWel2}) of these limiting frequencies as being the frequencies of the
eigenmodes of a certain conservative Lagrangian system whose Lagrangian is
also a quadratic form similar to (\ref{sintro2}) but associated with
$\operatorname{Ker}R$.
\end{remark}

\subsubsection{Selective overdamping\label{sbsbsecsod}}

Now we willl describe the selective overdamping phenomenon in terms of the
above modal dichotomy. First, we need to define the generic condition which is
the assumption that the nonzero eigenvalues of $\alpha^{-1}R$ and $\eta^{-1}R$
[in particular, $\sigma\left(  \alpha^{-1}R\right)  \setminus\left\{
0\right\}  =\left\{  b_{1},\ldots,b_{N_{R}}\right\}  $ and $\sigma\left(
\eta^{-1}R\right)  \setminus\left\{  0\right\}  =\left\{  b_{1}^{\flat}%
,\ldots,b_{N_{R}}^{\flat}\right\}  $] are simple (that is, their geometric
multiplicity is one), i.e.,%
\begin{equation}
b_{i}\not =b_{j}\text{, }b_{i}^{\flat}\not =b_{j}^{\flat}\text{, if }%
i\not =j\text{, for }1\leq i,\ j\leq N_{R}\text{\quad(generic condition).}
\label{sintmr0d}%
\end{equation}
Next, we define $\beta_{0}$, $\beta_{1}$, and $\beta_{2}$ as
\begin{align}
\beta_{0}  &  =\frac{2\omega_{\max}}{d},\text{ where }d=\min_{0\leq i,j\leq
N_{R},i\not =j}\left\vert b_{i}-b_{j}\right\vert ,\label{sintmr0e}\\
\beta_{1}  &  =\max\left\{  \beta_{0},\frac{2\omega_{\max}^{\flat}}{d^{\flat}%
}\right\}  ,\text{ where }d^{\flat}=\min_{0\leq i,j\leq N_{R},i\not =%
j}\left\vert b_{i}^{\flat}-b_{j}^{\flat}\right\vert ,\label{sintmr0f}\\
\beta_{2}  &  =\max\left\{  \min\left\{  c^{-1}\left(  \rho_{\min}/2\right)
,\left(  c^{\flat}\right)  ^{-1}\left(  \rho_{\min}^{\flat}/2\right)
\right\}  ,2\frac{\omega_{\max}}{b_{\min}},2\frac{\omega_{\max}^{\flat}%
}{b_{\min}^{\flat}}\right\}  , \label{sintmr0g}%
\end{align}

One of the most important facts we prove in this paper is that \emph{selective
overdamping is a generic phenomenon}. It will occur when $\beta$ is
sufficiently large, i.e., $\beta\gg1$, provided $N_{R}<N$ and the generic
condition is satisfied (\ref{sintmr0d}). Under these conditions and in terms
of the modal dichotomy describe above, the selective overdamping phenomenon
can be described as occurring in the following \emph{three stages (i)-(iii)}
with increasing $\beta$:

\textbf{(i)} If $\beta>\beta_{0}$ (Theorem \ref{tbeta0}) then $\beta
>2\frac{\omega_{\max}}{b_{\min}}$ and the $N_{R}$-dimensional subspace
$\mathbb{V}_{h\ell}\left(  \beta\right)  $ is spanned by overdamped
$\mathbb{V}\left(  \beta\right)  $-eigenmodes and, in particular, if $\left[
Q,\dot{Q}\right]  ^{\mathrm{T}}\in\mathbb{V}_{h\ell}\left(  \beta\right)  $,
where $Q=Q\left(  t\right)  =qe^{-\mathrm{i}\zeta t}$ is an eigenmode of
(\ref{sintro1}) then%
\[
\operatorname{Re}\zeta=0\text{.}%
\]

\textbf{(ii)} If $\beta>\beta_{1}$ (Corollary \ref{cbeta0}) then $\beta
>\max\left\{  2\frac{\omega_{\max}}{b_{\min}},2\frac{\omega_{\max}^{\flat}%
}{b_{\min}^{\flat}}\right\}  $ and\ the $N_{R}$-dimensional subspace
$\mathbb{V}_{\ell\ell,0}\left(  \beta\right)  $ is spanned by overdamped
$\mathbb{V}\left(  \beta\right)  $-eigenmodes and, in particular, if $\left[
Q,\dot{Q}\right]  ^{\mathrm{T}}\in\mathbb{V}_{\ell\ell,0}\left(  \beta\right)
$, where $Q=Q\left(  t\right)  =qe^{-\mathrm{i}\zeta t}$ is an eigenmode of
(\ref{sintro1}) then%
\[
\operatorname{Re}\zeta=0\text{.}%
\]

\textbf{(iii)} If $\beta>\beta_{2}$ (Corollary \ref{cllspmdd}; see also
Theorem \ref{tllspmd} and Corollary \ref{cllspmd}) then $\beta>\max\left\{
2\frac{\omega_{\max}}{b_{\min}},2\frac{\omega_{\max}^{\flat}}{b_{\min}^{\flat
}}\right\}  $ and the $\left(  2N-2N_{R}\right)  $-dimensional subspace
$\mathbb{V}_{\ell\ell,1}\left(  \beta\right)  $ is spanned by underdamped
$\mathbb{V}\left(  \beta\right)  $-eigenmodes and, in particular, if $\left[
Q,\dot{Q}\right]  ^{\mathrm{T}}\in\mathbb{V}_{\ell\ell,1}\left(  \beta\right)
$, where $Q=Q\left(  t\right)  =qe^{-\mathrm{i}\zeta t}$ is an eigenmode of
(\ref{sintro1}) then%
\[
\operatorname{Re}\zeta\not =0\text{.}%
\]

\begin{remark}
[selective overdamping estimates]One of the main goals of the paper has been
achieved, namely, we have given explicit formulas in terms of $\beta$ for
upper bound estimates on the amount of loss required in order that the lossy
component of a composite system, as modeled by our Lagrangian system
(\ref{sintro1}) when $R$ is rank deficient, for the selective overdamping to
occur in the generic case [in terms of $\beta_{0}$, $\beta_{1}$, and
$\beta_{2}$ as occuring in the stages (i)-(iii)] and have given Q-factor
estimates for the underdamped modes [in (iii).(a) of the modal dichotomy].
\end{remark}

\subsection{Overview of our framework\label{sinframe}}

Here we give a brief description of our framework we will use in our paper to
study the modal dichotomy, overdamping phenomena, and the associated spectral
problems that arise in this study. Further details on this framework can be
found in \cite{FigWel1}, \cite{FigWel2}.

Consider the Lagrangian system with equations of motion (\ref{sintro1}). The
eigenmodes $Q\left(  t\right)  =qe^{-\mathrm{i}\zeta t}$ of this Lagrangian
system corresponds to eigenpairs $\zeta$, $q$ of the quadratic \emph{matrix
pencil} $C\left(  \zeta,\beta\right)  $, i.e., solutions of the quadratic
eigenvalue problem:
\begin{gather}
C\left(  \zeta,\beta\right)  q=0,\text{\quad}0\not =q\in%
\mathbb{C}
^{N}\qquad\,\text{(quadratic eigenvalue problem),}\label{qevp}\\
C\left(  \zeta,\beta\right)  =\zeta^{2}\alpha+\left(  2\theta+\beta R\right)
\mathrm{i}\zeta-\eta\qquad\text{(quadratic pencil).} \label{qmpen}%
\end{gather}
\qquad Hence, the set of eigenvalues (spectrum) of the pencil $C\left(
\zeta,\beta\right)  $ is the set%
\begin{equation}
\sigma\left(  C\left(  \cdot,\beta\right)  \right)  =\left\{  \zeta\in%
\mathbb{C}
:\det\left(  C\left(  \zeta,\beta\right)  \right)  =0\right\}  \qquad
\text{(pencil spectrum),} \label{qspec}%
\end{equation}
which are exactly those values $\zeta$ for which an eigenmode of the
Lagrangian system with time-dependency $e^{-\mathrm{i}\zeta t}$ exists.

This form of the spectral problem is not suitable for the well-developed
perturbation theory of linear operators \cite{Bau85}, \cite{Kato},
\cite{Wel11}. Thus, we convert the spectral problem to the standard form by
making a change-of-variables from $\left[  Q,\dot{Q}\right]  ^{\mathrm{T}}$ to
a new variable $v$ via
\begin{gather}
v=Ku,\qquad K=%
\begin{bmatrix}
\sqrt{\alpha}^{-1} & 0\\
0 & \sqrt{\eta}%
\end{bmatrix}%
\begin{bmatrix}
\mathbf{1} & -\theta\\
0 & \mathbf{1}%
\end{bmatrix}
,\label{chvar}\\
u=%
\begin{bmatrix}
P\\
Q
\end{bmatrix}
=%
\begin{bmatrix}
\theta & \alpha\\
\mathbf{1} & 0
\end{bmatrix}%
\begin{bmatrix}
Q\\
\dot{Q}%
\end{bmatrix}
\text{ (change-of-variables),}\nonumber
\end{gather}
where $Q$, $P$ are conjugate variables with $P=\alpha\dot{Q}+\theta Q$ the
\emph{conjugate momentum} and $\mathbf{1}$ denotes the $N\times N$ identity
matrix. The variables $\left[  Q,\dot{Q}\right]  ^{\mathrm{T}}$ and $v$ are
related to the system energy, i.e., the Hamiltonian $\mathcal{H}%
=\mathcal{H}\left(  P,Q\right)  $, by%
\begin{equation}
\frac{1}{2}\left(  v,v\right)  =\mathcal{H}\left(  P,Q\right)  =\frac{1}%
{2}\left(  \dot{Q},\alpha\dot{Q}\right)  +\frac{1}{2}\left(  Q,\eta Q\right)
\qquad\text{(system energy),} \label{ennvar}%
\end{equation}
where $\left(  \cdot,\cdot\right)  $ denotes the standard complex inner product.

This change-of-variable takes solutions $Q$ in $%
\mathbb{C}
^{N}$ of the Lagrangian system (\ref{sintro1}) to solutions $v$ in $%
\mathbb{C}
^{2N}$ of the canonical system, i.e, solutions of the canonical evolution
equations%
\begin{gather}
\partial_{t}v=-\mathrm{i}A\left(  \beta\right)  v,\quad A\left(  \beta\right)
=\Omega-\mathrm{i}\beta B,\quad\beta\geq0\quad\text{(canonical evolution
equations),}\label{ceveqs}\\
\text{where }\beta\geq0,\quad v\left(  t\right)  \in H=%
\mathbb{C}
^{2N}.\nonumber
\end{gather}
The evolution of this canonical system is governed by a contraction semigroup
$e^{-\mathrm{i}A\left(  \beta\right)  t}$ in which the system operator
$A\left(  \beta\right)  $ has the important fundamental properties%
\begin{gather}
A\left(  \beta\right)  ^{\ast}=-A\left(  \beta\right)  ^{\mathrm{T}},\text{
}\operatorname{Re}A\left(  \beta\right)  =\Omega\qquad\text{(frequency
operator),}\label{ceveqs1}\\
-\operatorname{Im}A\left(  \beta\right)  =\beta B\geq0\qquad\text{(power
dissipation condition),}\nonumber\\
0<N_{R}=\operatorname{rank}B\leq N\qquad\text{(rank deficient).}\nonumber
\end{gather}
Here $\operatorname{Re}A\left(  \beta\right)  $, $\operatorname{Im}A\left(
\beta\right)  $ denote the real and imaginary part, respectively, of the
$2N\times2N$ matrix $A\left(  \beta\right)  $. The matrices $\Omega$, $B$ are
given in block form by%
\begin{gather}
\Omega=\left[
\begin{array}
[c]{cc}%
\Omega_{\mathrm{p}} & -\mathrm{i}\Phi^{\mathrm{T}}\\
\mathrm{i}\Phi & 0
\end{array}
\right]  ,\qquad B=\left[
\begin{array}
[c]{ll}%
\mathsf{\tilde{R}} & 0\\
0 & 0
\end{array}
\right]  ,\label{ceveqs2}\\
\Omega_{\mathrm{p}}=-\mathrm{i}2K_{\mathrm{p}}\theta K_{\mathrm{p}%
}^{\mathrm{T}},\qquad\Phi=K_{\mathrm{q}}K_{\mathrm{p}}^{\mathrm{T}}%
,\qquad\mathsf{\tilde{R}}=K_{\mathrm{p}}RK_{\mathrm{p}}^{\mathrm{T}%
},\nonumber\\
K_{\mathrm{p}}=\sqrt{\alpha}^{-1},\qquad K_{\mathrm{q}}=\sqrt{\eta},\nonumber
\end{gather}
where $\sqrt{\alpha}$ and $\sqrt{\eta}$ denote the unique positive definite
and positive semidefinite square roots of the matrices $\alpha$ and $\eta$,
respectively. In particular, this implies the real $N\times N$ matrices
$K_{\mathrm{p}}$, $K_{\mathrm{q}}$ have the properties
\begin{equation}
K_{\mathrm{p}}^{\mathrm{T}}=K_{\mathrm{p}}>0,\qquad K_{\mathrm{q}}%
^{\mathrm{T}}=K_{\mathrm{q}}\geq0. \label{ceveqs3}%
\end{equation}

The modal dichotomy and overdamping phenomenon is now studied via the spectral
perturbation analysis in $\beta$ of the system operator $A\left(
\beta\right)  $ and the standard spectral problem%
\begin{equation}
A\left(  \beta\right)  w=\zeta w,0\not =w\in%
\mathbb{C}
^{2N}\qquad\text{(standard eigenvalue problem),} \label{sevp}%
\end{equation}
and, in particular, its spectrum
\begin{equation}
\sigma\left(  A\left(  \beta\right)  \right)  =\left\{  \zeta\in%
\mathbb{C}
:\det\left(  \zeta\mathbf{1}-A\left(  \beta\right)  \right)  =0\right\}
\qquad\text{(system operator spectrum)}. \label{sevspec}%
\end{equation}
The main reason that we can study the standard spectral problem instead of the
quadratic eigenvalue problem is that an eigenmode $Q\left(  t\right)
=qe^{-\mathrm{i}\zeta t}$ of the Lagrangian system corresponds to an eigenmode
$v\left(  t\right)  =we^{-\mathrm{i}\zeta t}$ of the canonical system which
means $\zeta$, $w$ is an eigenpair of $A\left(  \beta\right)  $, i.e., a
solution of the spectral problem. This correspondence is elaborated on in
Corollary \ref{cpfsp}, In particular, by this corollary, we have the equality
of the spectra
\begin{equation}
\sigma\left(  C\left(  \cdot,\beta\right)  \right)  =\sigma\left(  A\left(
\beta\right)  \right)  ,\qquad\beta\geq0\qquad\text{(equivalence of
spectrum).} \label{sevspeq}%
\end{equation}

Finally, based on our perturbation theory developed in \cite{FigWel1}, it
follows that, except for only a finite number of $\beta$ in $[0,\infty)$, the
eigenvalues of $A\left(  \beta\right)  $ are semi-simple and $A\left(
\beta\right)  $ is diagonalizable. We prove this statement now.

\begin{proposition}
[diagonalization]\label{pdiagsop}The system operator $A\left(  \beta\right)  $
is diagonalizable for all $\beta\in\lbrack0,\infty)$ except for a finite set
of positive values of $\beta$.
\end{proposition}

\begin{proof}
As the matrix $A\left(  \beta\right)  =\Omega-\mathrm{i}\beta B$, $\beta\in%
\mathbb{C}
$ is analytic then, by a well-known fact from perturbation theory
\cite[Theorem 3, p. 25 and Theorem 1, p. 225]{Bau85}, its Jordan normal form
is invariant except on a set $S\subseteq$ $%
\mathbb{C}
$ which is closed and isolated. By the proof of our results \cite[Theorem 5,
Sec. IV.A and Theorem 15, Sec. IV.B]{FigWel1}, it follows that there exists
$\beta_{\ell}$, $\beta_{h}>0$ such that $A\left(  \beta\right)  $ is
diagonalizable with invariant Jordan normal form for $0<\left\vert
\beta\right\vert <\beta_{\ell}$ and for $\left\vert \beta\right\vert
>\beta_{h}$. Also, $A\left(  0\right)  =\Omega$ is Hermitian and so it is
diagonalizable. These facts imply this exceptional set $S$ is finite,
$A\left(  \beta\right)  $ is diagonalizable with an invariant Jordan normal
form on $%
\mathbb{C}
\setminus S$, and only in the finite set $\left\{  \beta\in%
\mathbb{C}
:\beta_{\ell}\leq\left\vert \beta\right\vert \leq\beta_{h}\right\}  \cap S$ is
it possible for $A\left(  \beta\right)  $ to not be diagonalizable. This
completes the proof.
\end{proof}

\section{Electric circuit example\label{scircuits}}

Among important applications of methods and results described in this paper
are electric circuits and networks involving resistors (lossy elements) and
gyrators (gyroscopic elements), where the latter are lossless nonreciprocal
circuit elements which was introduced in \cite{Tell48} (see also
\cite{CarGio64}). A general study of electric networks with losses can be
carried out with the help of the Lagrangian approach, and that systematic
study has already been carried out in \cite{FigWel2} and in this paper. For
more on the Lagrangian treatment of electric networks and circuits we refer
the reader to \cite[Sec. 9]{Gantmacher}, \cite[Sec. 2.5]{Goldstein}.

We now will illustrate the idea and give a flavor of the efficiency of our
methods by considering below an example of a rather simple electric circuit
with a gyrator as depicted on the top of Fig. \ref{Figc1} with the assumptions%
\begin{equation}
L_{1},\text{ }L_{2},\text{ }C_{1},\text{ }C_{2},\text{ }C_{12}>0\text{ and
}R_{2},\text{ }G_{12}\geq0\text{.} \label{circ0}%
\end{equation}
This example\ has the essential features of two component systems
incorporating lossy and lossless components. It serves to illustrate how our
theory can be used to calculate explicitly all the terms in (\ref{sintmr00}%
)-(\ref{sintmr02}), (\ref{sintmr0a})-(\ref{sintmrd1a}), (\ref{sintmr0g0}%
)-(\ref{sintmr0g}) for the phenomenon of modal dichotomy and selective
overdamping phenomenon. After we have done this, we will examine numerically
the phenomena using this example.

\begin{figure}
[ptbh]
\begin{center}
\includegraphics[
height=5.3583in,
width=4.8403in
]%
{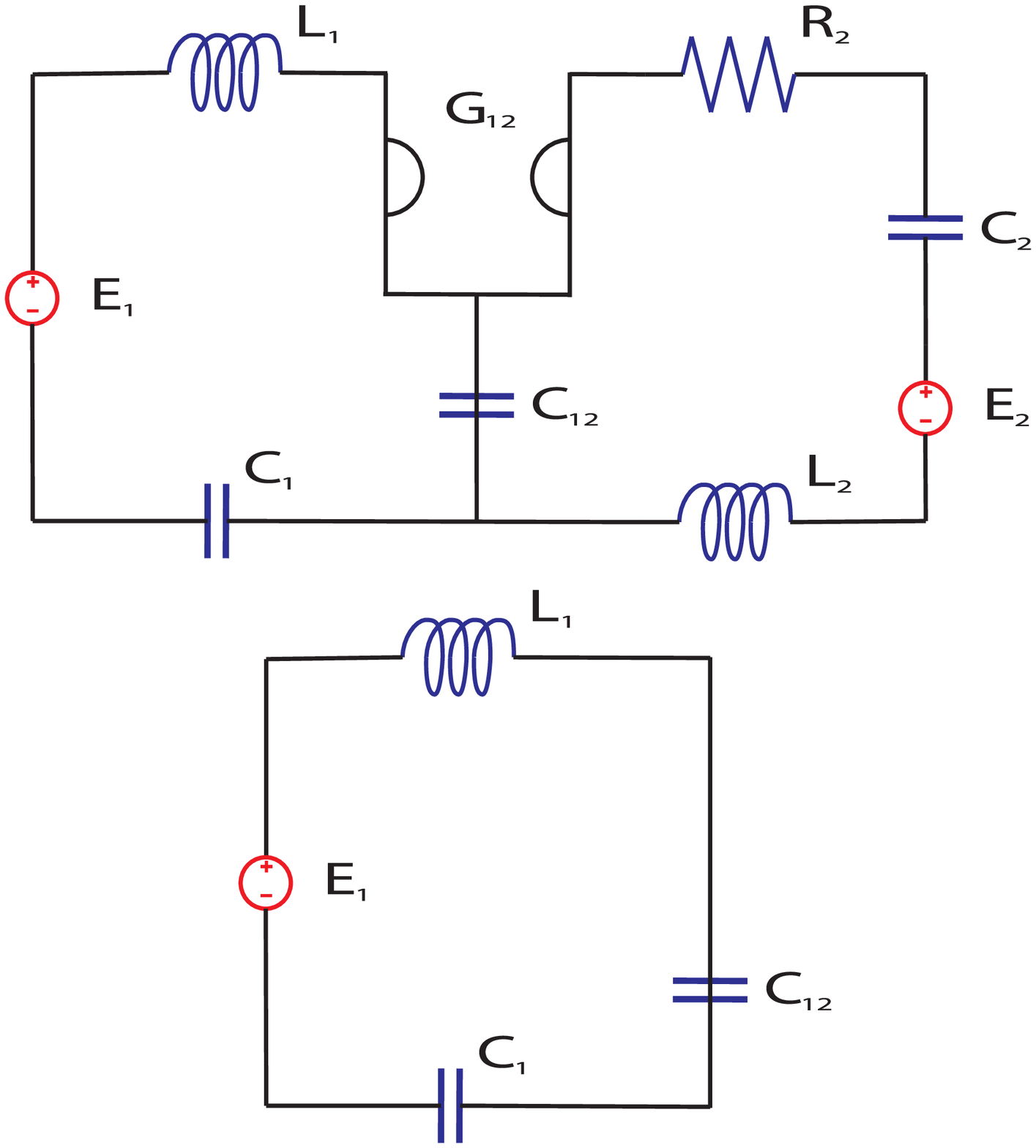}%
\caption{(Top) The electric circuit on top has a gryator $G_{12}$, three
capacitances $C_{1}$, $C_{2}$, $C_{12}$, two inductances $L_{1}$, $L_{2}$, a
resistor $R_{2}$, and two sources $E_{1}$, $E_{2}$. This electric circuit
example fits within the framework of our model.\ Indeed, the resistor $R_{2}$
represents losses, the gyrator $G_{12}$ represents gyroscopy, and this two
component system consists of a lossy component and a lossless component -- the
right and left circuits, respectively. (Bottom) The electric circuit on the
bottom is obtained from the circuit above by eliminating in it the right
circuit with lossy component $R_{2}$. This elimination corresponds to imposing
constraint $q_{2}=0$ on the Lagrangian defined by equation (\ref{circ1}). This
bottom circuit has inductance $L_{1}$ and two capacitances $C_{1}$, $C_{12}$
connected in series to a source $E_{1}$. The frequencies of the eigenmodes of
this circuit are $\rho_{3}$ and $\rho_{4}$ which are the limits of
$\operatorname{Re}\zeta_{3}\left(  \beta\right)  $ and $\operatorname{Re}%
\zeta_{4}\left(  \beta\right)  $ [see the asymptotic expansions in
(\ref{circllhq})], the frequencies of the underdamped eigenmodes of the
circuit above, as $\beta=\frac{R_{2}}{\ell}\rightarrow+\infty$, i.e., as
losses in the lossy component (top right circuit) become infinite, as
predicted by our theory on the selective overdamping phenomenon (see Fig.
\ref{Fig_udmf} and Remark \ref{rspredl2c} for more details on these results
for this electric circuit example and Section \ref{smdhlr}, Proposition
\ref{pspredl2} and Remarks \ref{rspredl2_0}, \ref{rspredl2} for more details
on the general theory).}%
\label{Figc1}%
\end{center}
\end{figure}

\subsection{\textbf{Lagrangian system\label{scircuitLag}}}

To derive evolution equations for the electric circuit with a gyrator in Fig.
\ref{Figc1} we use a general method for constructing Lagrangians for circuits,
\cite[Sec. 9]{Gantmacher}, that yields%
\begin{align}
\mathcal{L}  &  =\mathcal{T}-\mathcal{V}\text{
\ \ \ \ \ \ \ \ \ \ \ \ \ \ \ (circuit Lagrangian),}\label{circ1}\\
\mathcal{T}  &  =\frac{L_{1}}{2}\dot{q}_{1}^{2}+\frac{L_{2}}{2}\dot{q}_{2}%
^{2}+\frac{G_{12}}{4}(q_{1}\dot{q}_{2}-\dot{q}_{1}q_{2}),\quad\nonumber\\
\mathcal{V}  &  =\frac{1}{2C_{1}}q_{1}^{2}+\frac{1}{2C_{12}}\left(
q_{1}-q_{2}\right)  ^{2}+\frac{1}{2C_{2}}q_{2}^{2}-\frac{G_{12}}{4}(q_{1}%
\dot{q}_{2}-\dot{q}_{1}q_{2}),\nonumber\\
\mathcal{R}  &  =\frac{R_{2}}{2}\dot{q}_{2}^{2},\nonumber
\end{align}
where $\mathcal{L}$ is the Lagrangian, $\mathcal{R}$ is the Rayleigh
dissipative function, and $I_{1}=\dot{q}_{1}$, $I_{2}=\dot{q}_{2}$ are the
currents.\ The sources we take to be zero, i.e., $E_{1}=E_{2}=0$. The general
Euler-Lagrange equations of motion with forces are%
\begin{equation}
\frac{\partial}{\partial t}\frac{\partial\mathcal{L}}{\partial\dot{Q}}%
-\frac{\partial\mathcal{L}}{\partial Q}=-\frac{\partial\mathcal{R}}%
{\partial\dot{Q}}, \label{circ2}%
\end{equation}
where $Q$ are the charges%
\begin{equation}
Q=\left[
\begin{array}
[c]{c}%
q_{1}\\
q_{2}%
\end{array}
\right]  , \label{circ3}%
\end{equation}
yielding from (\ref{circ0})--(\ref{circ3}) the following second-order ODEs%
\begin{equation}
\alpha\ddot{Q}+\left(  2\theta+\beta R\right)  \dot{Q}+\eta Q=0, \label{circ4}%
\end{equation}
with the dimensionless loss parameter%
\begin{equation}
\beta=\frac{R_{2}}{\ell}\qquad\text{(where }\ell>0\text{ is fixed and has same
units as }R_{2}\text{)} \label{circ5}%
\end{equation}
that scales the intensity of losses in the system, and%
\begin{gather}
\alpha=\left[
\begin{array}
[c]{cc}%
L_{1} & 0\\
0 & L_{2}%
\end{array}
\right]  ,\qquad\eta=\left[
\begin{array}
[c]{cc}%
\frac{1}{C_{1}}+\frac{1}{C_{12}} & -\frac{1}{C_{12}}\\
-\frac{1}{C_{12}} & \frac{1}{C_{2}}+\frac{1}{C_{12}}%
\end{array}
\right]  ,\label{circ6}\\
\theta=\left[
\begin{array}
[c]{cc}%
0 & -\frac{G_{12}}{2}\\
\frac{G_{12}}{2} & 0
\end{array}
\right]  ,\qquad R=\left[
\begin{array}
[c]{cc}%
0 & 0\\
0 & \ell
\end{array}
\right]  .\nonumber
\end{gather}
Recall, the loss fraction $\delta_{R}$ defined in (\ref{sintro4}) is the ratio
of the rank of the matrix $R$ to the total degrees of freedom $N$ of the
system which in this case is
\begin{equation}
\text{loss fraction:\quad}\delta_{R}=\frac{N_{R}}{N}=\frac{1}{2},\qquad
N=2,\qquad N_{R}=\operatorname{rank}R=1. \label{circ7}%
\end{equation}
Thus, the Lagrangian system (\ref{circ4}) fits with our framework described in
Sec. \ref{sinmodel} with the loss fraction condition (\ref{sintro4a}),
satisfied, i.e., $0<\delta_{R}<1$, and hence is a model of a two-component
composite with a lossy and a lossless component.

\subsection{Modal dichotomy and overdamping}

We now will describe the modal dichotomy and overdamping phenomenon for this
electric circuit following our discussion in Section \ref{sinresults}.

First, the duality condition \ref{cnddl} holds in this example as%
\begin{equation}
\eta>0 \label{circ6_1}%
\end{equation}
and so the equations of motion for the dual Lagrangian system are%
\begin{equation}
\eta\ddot{Q}+\left(  2\theta+\beta R\right)  \dot{Q}+\alpha Q=0.
\label{circ4dual}%
\end{equation}

We now begin by calculating the spectra $\sigma\left(  \alpha^{-1}R\right)  $
and $\sigma\left(  \eta^{-1}R\right)  $ in order to calculate $b_{\min}$ and
$b_{\min}^{\flat}$ in (\ref{sintmr02}) and (\ref{sintmr0b}), respectively:%
\begin{align}
\sigma\left(  \alpha^{-1}R\right)   &  =\left\{  b_{0},b_{1}\right\}  ,\text{
}b_{0}=0,\text{ }b_{\min}=b_{1}=\ell L_{2}^{-1},\label{circ8}\\
\sigma\left(  \eta^{-1}R\right)   &  =\left\{  b_{0}^{\flat},b_{1}^{\flat
}\right\}  ,\text{ }b_{0}^{\flat}=0,\text{ }b_{\min}^{\flat}=b_{1}^{\flat
}=\ell\left(  \frac{1}{C_{2}}+\frac{1}{C_{1}+C_{12}}\right)  ^{-1}.
\label{circ9}%
\end{align}
\qquad Next, we calculate $\omega_{\max}$, $\omega_{\min}$, and $\omega_{\max
}^{\flat}$ in (\ref{sintmr00}), (\ref{sintmr01}), and (\ref{sintmr0a}),
respectively, from the spectrum of the pencil%
\[
C\left(  \zeta,\beta\right)  =\zeta^{2}\alpha+\left(  2\theta+\beta R\right)
\mathrm{i}\zeta-\eta,
\]
at $\beta=0$, i.e.,%
\[
\sigma\left(  C\left(  \cdot,0\right)  \right)  =\left\{  \zeta\in%
\mathbb{C}
:\det\left(  \zeta^{2}\alpha\mathbf{+}2\zeta\mathrm{i}\theta-\eta\right)
=0\right\}  =\left\{  \pm\omega_{\min},\pm\omega_{\max}\right\}  ,
\]
where%
\begin{gather*}
0<\omega_{\min}\leq\omega_{\max},\\
\omega_{\max}=\sqrt{\frac{a}{2}+\sqrt{\left(  \frac{a}{2}\right)  ^{2}-\left(
\det\alpha\right)  ^{-1}\det\eta}},\\
\omega_{\min}=\frac{1}{\omega_{\max}^{\flat}}=\sqrt{\frac{a}{2}-\sqrt{\left(
\frac{a}{2}\right)  ^{2}-\left(  \det\alpha\right)  ^{-1}\det\eta}},\\
a=\left(  \frac{1}{L_{2}C_{12}}+\frac{1}{L_{1}C_{12}}+\frac{1}{L_{2}C_{2}%
}+\frac{1}{L_{1}C_{1}}+\frac{G_{12}^{2}}{L_{1}L_{2}}\right)  ,\text{ }\\
\left(  \det\alpha\right)  ^{-1}\det\eta=\left(  L_{1}L_{2}\right)
^{-1}\left(  \frac{1}{C_{1}C_{2}}+\frac{1}{C_{1}C_{12}}+\frac{1}{C_{2}C_{12}%
}\right)  .
\end{gather*}
Next, we calculate the spectra $\rho_{\min}$, $\rho_{\max}$, and $\rho_{\min
}^{\flat}$ in (\ref{sintmr0g0})-(\ref{sintmr0g2}), which following the Remark
\ref{raltspc}, can be computed using Proposition \ref{pspredl2} as%
\begin{gather*}
\rho_{\min}=\min\left\{  \rho\in(0,\infty):\det\left[  \left(  P_{R}^{\bot
}\alpha^{-1}C\left(  \rho,0\right)  P_{R}^{\bot}\right)  |_{\operatorname{Ker}%
R}\right]  =0\right\}  ,\\
\rho_{\max}=\max\left\{  \rho\in(0,\infty):\det\left[  \left(  P_{R}^{\bot
}\alpha^{-1}C\left(  \rho,0\right)  P_{R}^{\bot}\right)  |_{\operatorname{Ker}%
R}\right]  =0\right\}  ,\\
\rho_{\min}^{\flat}=\rho_{\max}^{-1},
\end{gather*}
where $P_{R}^{\bot}$ is the orthogonal projection onto $\operatorname{Ker}R$
(the nullspace of $R$) and in this example,%
\begin{gather}
P_{R}^{\bot}=\left[
\begin{array}
[c]{cc}%
1 & 0\\
0 & 0
\end{array}
\right]  ,\text{ }P_{R}^{\bot}C\left(  \rho,0\right)  P_{R}^{\bot}=\left[
L_{1}\rho^{2}-\left(  \frac{1}{C_{1}}+\frac{1}{C_{12}}\right)  \right]
P_{R}^{\bot},\nonumber\\
\left\{  \rho\in%
\mathbb{C}
:\det\left[  \left(  P_{R}^{\bot}C\left(  \rho,0\right)  P_{R}^{\bot}\right)
|_{\operatorname{Ker}R}\right]  =0\right\}  =\left\{  \rho_{3},\rho
_{4}\right\}  ,\nonumber\\
\rho_{3}=-\rho_{4}=\rho_{\min}=\rho_{\max}=\frac{1}{\rho_{\min}^{\flat}}%
=\sqrt{L_{1}^{-1}\left(  \frac{1}{C_{1}}+\frac{1}{C_{12}}\right)  }.
\label{circllhqp}%
\end{gather}

\begin{remark}
[limiting frequencies]\label{rspredl2c}In accordance with our theory (see
Section \ref{smdhlr}, Proposition \ref{pspredl2}, and Remarks \ref{rspredl2_0}%
, \ref{rspredl2} for more details), the real numbers $\rho_{3},\rho_{4}$ are
the frequencies of the eigenmodes of a conservative Lagrangian system with
Euler-Lagrange equations%
\[
L_{1}\ddot{q}_{1}+\left(  \frac{1}{C_{1}}+\frac{1}{C_{12}}\right)  q_{1}=0
\]
corresponding to the Lagrangian%
\[
\mathcal{L}_{\operatorname{Ker}R}=\mathcal{L}_{\operatorname{Ker}R}\left(
q_{1},\dot{q}_{1}\right)  =\frac{1}{2}\left[
\begin{array}
[c]{l}%
\dot{q}_{1}\\
q_{1}%
\end{array}
\right]  ^{\mathrm{T}}\left[
\begin{array}
[c]{ll}%
L_{1} & 0\\
0 & \frac{1}{C_{1}}+\frac{1}{C_{12}}%
\end{array}
\right]  \left[
\begin{array}
[c]{l}%
\dot{q}_{1}\\
q_{1}%
\end{array}
\right]  .
\]
In particular, this is the Lagrangian of the electric circuit on the bottom of
Fig. \ref{Figc1} with inductance $L_{1}$ and two capacitances $C_{1}$,
$C_{12}$ connected in series (with no source, i.e., $E_{1}=0$). Notice that
this is the same Lagrangian for a $LC$-circuit with inductor $L_{1}$ and
capacitor $\left(  \frac{1}{C_{1}}+\frac{1}{C_{12}}\right)  ^{-1}$. This makes
sense since it well-known in electric circuit theory that connecting two
capacitors $C_{1}$ and $C_{12}$ in series is the same as having one capacitor
$C$ which is the one-half of the harmonic mean of the two capacitors, i.e.,
$C=$ $\left(  \frac{1}{C_{1}}+\frac{1}{C_{12}}\right)  ^{-1}$.
\end{remark}

Next, as the nonzero eigenvalues of $\alpha^{-1}R$ and $\eta^{-1}R$ are simple
then this implies the generic condition (\ref{cndgc}) is true for both the
Lagrangian system (\ref{circ4}) and its dual system (\ref{circ4dual}). Thus,
the terms (\ref{sintmr0e})-(\ref{sintmr0g}) for the selective overdamping in
this example are \
\begin{align}
\beta_{0}  &  =\frac{2\omega_{\max}}{d},\text{ where }d=\min_{0\leq i,j\leq
N_{R},i\not =j}\left\vert b_{i}-b_{j}\right\vert =b_{1},\label{circbeta0}\\
\beta_{1}  &  =\max\left\{  \beta_{0},\frac{2\omega_{\max}^{\flat}}{d^{\flat}%
}\right\}  ,\text{ where }d^{\flat}=\min_{0\leq i,j\leq N_{R},i\not =%
j}\left\vert b_{i}^{\flat}-b_{j}^{\flat}\right\vert =b_{1}^{\flat
},\label{circbeta1}\\
\beta_{2}  &  =\max\left\{  \min\left\{  c^{-1}\left(  \rho_{\min}/2\right)
,\left(  c^{\flat}\right)  ^{-1}\left(  \rho_{\min}^{\flat}/2\right)
\right\}  ,2\frac{\omega_{\max}}{b_{\min}},2\frac{\omega_{\max}^{\flat}%
}{b_{\min}^{\flat}}\right\}  , \label{circbeta2}%
\end{align}
where the functions $c^{-1}\left(  y\right)  $ and $\left(  c^{\flat}\right)
^{-1}\left(  y\right)  $\ are defined in (\ref{sintmr1a}) and (\ref{sintmrd1a}%
), respectively.

Finally, according to our theory, as $\beta\rightarrow\infty$ there are
eigenmodes $Q_{j}=Q_{j}\left(  t,\beta\right)  =q_{j}\left(  \beta\right)
e^{-\mathrm{i}\zeta_{j}\left(  \beta\right)  t}$, $j=1,2,3,4$ of the
Lagrangian system such that $\left[  Q_{j},\partial_{t}Q_{j}\right]
^{\mathrm{T}}$, $j=1,2,3,4$ is a basis for the phase space $\mathbb{V}\left(
\beta\right)  $, as defined in (\ref{sintmr1e}) for this Lagrangian system
(\ref{circ4})-(\ref{circ6}), and which split into two distinct classes%
\[
\text{high-loss}\text{: }Q_{1}(\beta);\text{ low-loss}\text{: }Q_{j}%
(\beta),\text{ }2\leq j\leq4
\]
with the asymptotic expansions:%
\begin{equation}
\text{high-loss:}-\operatorname{Im}\zeta_{1}\left(  \beta\right)  =b_{1}%
\beta+O\left(  \beta^{-1}\right)  ,\text{ }\operatorname{Re}\zeta_{1}\left(
\beta\right)  =0, \label{circllhq0}%
\end{equation}%
\begin{equation}
\text{low-loss, low-Q:}-\operatorname{Im}\zeta_{2}\left(  \beta\right)
=\frac{1}{b_{1}^{\flat}}\beta^{-1}+O\left(  \beta^{-3}\right)  ,\text{
}\operatorname{Re}\zeta_{2}\left(  \beta\right)  =0, \label{circllhq1}%
\end{equation}%
\begin{gather}
\text{low-loss, high-Q:}-\operatorname{Im}\zeta_{j}\left(  \beta\right)
=d_{j}\beta^{-1}+O\left(  \beta^{-3}\right)  ,\text{ }\operatorname{Re}%
\zeta_{j}\left(  \beta\right)  =\rho_{j}+O\left(  \beta^{-2}\right)
,\label{circllhq}\\
\text{(quality factor) }Q_{\zeta_{j}\left(  \beta\right)  }=\frac{\left\vert
\rho_{j}\right\vert }{d_{j}}\beta+O\left(  \beta^{-1}\right)  \text{ , for
}j=3,4,\nonumber
\end{gather}
where in this example
\begin{equation}
d_{3}=d_{4}=\frac{1}{2}\frac{G_{12}^{2}}{\ell L_{1}}+\frac{1}{2}\frac{1}{\ell
}\frac{C_{1}}{\left(  C_{1}+C_{12}\right)  C_{12}}, \label{circllhqd}%
\end{equation}
and $b_{1}$, $b_{1}^{\flat}$, $\rho_{3}$, $\rho_{4}$ are already calculated above.

Therefore, for this electric circuit example, the four stages (i)-(iv) of the
modal dichotomy as described in Sec. \ref{sbsbsecmd} occur for (i)
$\beta>\frac{2\omega_{\max}}{b_{\min}}=\beta_{0}$; (ii) $\beta>\max\left\{
\frac{2\omega_{\max}}{b_{\min}},\frac{2\omega_{\max}^{\flat}}{b_{\min}^{\flat
}}\right\}  =\beta_{1}$; (iii) $\beta>\max\left\{  c^{-1}\left(  \frac
{\rho_{\min}}{2}\right)  ,2\frac{\omega_{\max}^{\flat}}{b_{\min}^{\flat}%
}\right\}  \geq\beta_{2}\geq\beta_{1}$; (iv) $\beta\gg1$. Moreover, the three
stages (i)-(iii) of overdamping as described in Sec. \ref{sbsbsecsod} occur
for (i) $\beta>\beta_{0}$; (ii) $\beta>\beta_{1}$; (iii) $\beta>\beta_{2}$. In
particular, for this example, the two stages (i), (ii) of selective
overdamping correspond to the two stages (i), (ii) for the modal dichotomy, respectively.

\begin{remark}
The method used to calculate the $d_{j}$'s is found in Sec. \ref{smdhlr}. It
is calculated similar to the example in \cite[Sec. III]{FigWel1} by using the
formula in (\ref{pod14}) below and the system operator $A\left(  \beta\right)
=\Omega-\mathrm{i}\beta B$ for the Lagrangian system (\ref{circ4}) in this example.
\end{remark}

\begin{remark}
Notice that for these lowest order terms only in the low-loss, high-Q modes,
i.e., the underdamped modes, does gyroscopy effect the modes. And more
precisely for the lowest order asympotics ($b_{1}$, $b_{1}^{\flat}$, $\rho
_{j}$, $d_{j}$, $j=3,4$), gyroscopy has no effect on the frequencies
$\operatorname{Re}\zeta_{j}\left(  \beta\right)  $, $j=1,2,3,4$ or damping
factors of the overdamped modes $\operatorname{Im}\zeta_{j}\left(
\beta\right)  $, $j=1,2$, yet gyroscopy does have an effect on the damping
factors of the underdamped modes $\operatorname{Im}\zeta_{j}\left(
\beta\right)  $, $j=3,4$. The effect is proportional to $G_{12}^{2}=\left\Vert
2\theta\right\Vert ^{2}$, where $\left\Vert \cdot\right\Vert $ denotes the
operator norm. As the interplay between losses and gyroscopy in Lagrangian
systems is of significant interest, it would be interesting to derive formulas
for higher order terms for the frequencies and damping factors of the
eigenmodes to see how gyroscopy effects these terms asymptotically as
$\beta\rightarrow\infty$.
\end{remark}

\subsection{Numerical Analysis\label{snuma}}

We will now illustrate the behavior of the eigenmodes $Q_{j}\left(
t,\beta\right)  =q_{j}\left(  \beta\right)  e^{-\mathrm{i}\zeta_{j}\left(
\beta\right)  t}$, $j=1,2,3,4$ for the electric circuit with gyrator in Fig.
\ref{Figc1} as a function of the loss parameter $\beta$ based on the theory
described above and focusing on the overdamping phenomena and the asymptotic
expansions (\ref{circllhq0})-(\ref{circllhqd}) in the high-loss regime as
$\beta\rightarrow\infty$:%
\begin{align}
\text{high-loss}  &  \text{: }\zeta_{1}\left(  \beta\right)  =\zeta_{1}%
^{a}\left(  \beta\right)  +O\left(  \beta^{-1}\right)  ,\text{ }\zeta_{1}%
^{a}\left(  \beta\right)  =-\mathrm{i}b_{1}\beta,\text{ }\operatorname{Re}%
\zeta_{1}\left(  \beta\right)  =0,\label{circllhqe0}\\
\text{low-loss, low-Q}  &  \text{: }\zeta_{2}\left(  \beta\right)  =\zeta
_{2}^{a}\left(  \beta\right)  +O\left(  \beta^{-3}\right)  ,\text{ }\zeta
_{2}^{a}\left(  \beta\right)  =-\mathrm{i}\frac{1}{b_{1}^{\flat}}\beta
^{-1},\text{ }\operatorname{Re}\zeta_{2}\left(  \beta\right)
=0,\label{circllhqe1}\\
\text{low-loss, high-Q}  &  \text{: }\zeta_{j}\left(  \beta\right)  =\zeta
_{j}^{a}\left(  \beta\right)  +O\left(  \beta^{-2}\right)  ,\text{ }\zeta
_{j}^{a}\left(  \beta\right)  =\rho_{j}-\mathrm{i}d_{j}\beta^{-1},\text{
\ }\operatorname{Re}\zeta_{j}\left(  \beta\right)  \not =0, \label{circllhqe2}%
\end{align}
for $j=3,4$.

All the figures below were generated (by Marcus Marinho) with
MATLAB$^{\text{\textregistered}}$ using the framework described in Sec.
\ref{sinframe} by just calculating the eigenvalues of the system operator
$A\left(  \beta\right)  =\Omega-\mathrm{i}\beta B$, $\beta\geq0$ (a $4\times4$
matrix in this circuit example) since, according to our theory, they are the
values $\zeta_{j}\left(  \beta\right)  $, $j=1,2,3,4$.

We fix positive values of the capacitance $C_{1}$, $C_{2}$, $C_{12}$,
inductances $L_{1}$, $L_{2}$, gyration resistance $G_{12}/2$ (the term coined
by Tellegen in \cite{Tell48}), and $\ell$ [where the dimensionless loss
parameter is $\beta=R_{2}/\ell$ in (\ref{circ5}) with resistance $R_{2}>0$].
For the numerical analysis in this section, we choose%
\begin{equation}
C_{1}=C_{2}=25,\text{ }C_{12}=\frac{25}{3},\text{ }L_{1}=10,\text{ }%
L_{2}=\frac{1}{2},\text{ }G_{12}=\frac{5}{2},\text{ }\ell=10. \label{nae1}%
\end{equation}

The values of $\beta_{0}$, $\beta_{1}$, $\beta_{2}$ in Theorem \ref{tbeta0},
Corollary \ref{cbeta0}, and Corollary \ref{cllspmdd}, respectively, where the
high-loss modes are guaranteed to be overdamped for $\beta>\beta_{0}$, where
the overdamped low-loss modes are guaranteed to be overdamped for $\beta
>\beta_{1}$, and where the underdamped low-loss modes are guaranteed to be
underdamped for $\beta>\beta_{2}$ are determined explicitly using the analysis
in Sec. \ref{scircuitLag} which we calculate from Eqs. (\ref{circbeta0}%
)-(\ref{circbeta2}), using the values from (\ref{nae1}), to be%
\begin{align}
\beta_{0}  &  =\frac{2\omega_{\max}}{d}=\frac{2\omega_{\max}}{b_{\min}}%
=\frac{1}{1000}\sqrt{7930+10\sqrt{626609}}\approx0.1258803552,\label{nae2}\\
\beta_{1}  &  =\max\left\{  \beta_{0},\frac{2\omega_{\max}^{\flat}}{d^{\flat}%
}\right\}  =\frac{2\omega_{\max}^{\flat}}{b_{\min}^{\flat}}=\frac{7}%
{5\sqrt{7930-10\sqrt{626609}}}\approx0.3723591130,\nonumber\\
\beta_{2}  &  =\max\left\{  \min\left\{  c^{-1}\left(  \rho_{\min}/2\right)
,\left(  c^{\flat}\right)  ^{-1}\left(  \rho_{\min}^{\flat}/2\right)
\right\}  ,2\frac{\omega_{\max}}{b_{\min}},2\frac{\omega_{\max}^{\flat}%
}{b_{\min}^{\flat}}\right\}  =c^{-1}\left(  \rho_{\min}/2\right) \nonumber\\
&  =\frac{1}{20000}\left(  7930+10\sqrt{626609}\right)  \sqrt{10}+\frac
{1}{1000}\sqrt{7930+10\sqrt{626609}}\nonumber\\
&  \approx2.631331413.\nonumber
\end{align}

In the figures below (see Figs. \ref{Fig_hldp}-\ref{Fig_udmq}), we give a
graphical representation of the effects of increasing losses, i.e., increasing
$\beta=\frac{R_{2}}{\ell}$ (with $\ell$ fixed), in the lossy component of the
electric circuit with gyrator in Fig. \ref{Figc1} on frequencies
$\operatorname{Re}\zeta_{j}\left(  \beta\right)  $, damping factors
$-\operatorname{Im}\zeta_{j}\left(  \beta\right)  $, and quality factor
$Q_{\zeta_{j}\left(  \beta\right)  }=\frac{1}{2}\frac{\left\vert
\operatorname{Re}\zeta_{j}\left(  \beta\right)  \right\vert }%
{-\operatorname{Im}\zeta_{j}\left(  \beta\right)  }$ ($j=1,2,3,4$) of all the
eigenmodes of the dissipative-gyroscopic Lagrangian system (\ref{circ4}) [with
the numerical values (\ref{nae1})].%

\begin{figure}
[pth]
\begin{center}
\includegraphics[
height=4.9476in,
width=6.5956in
]%
{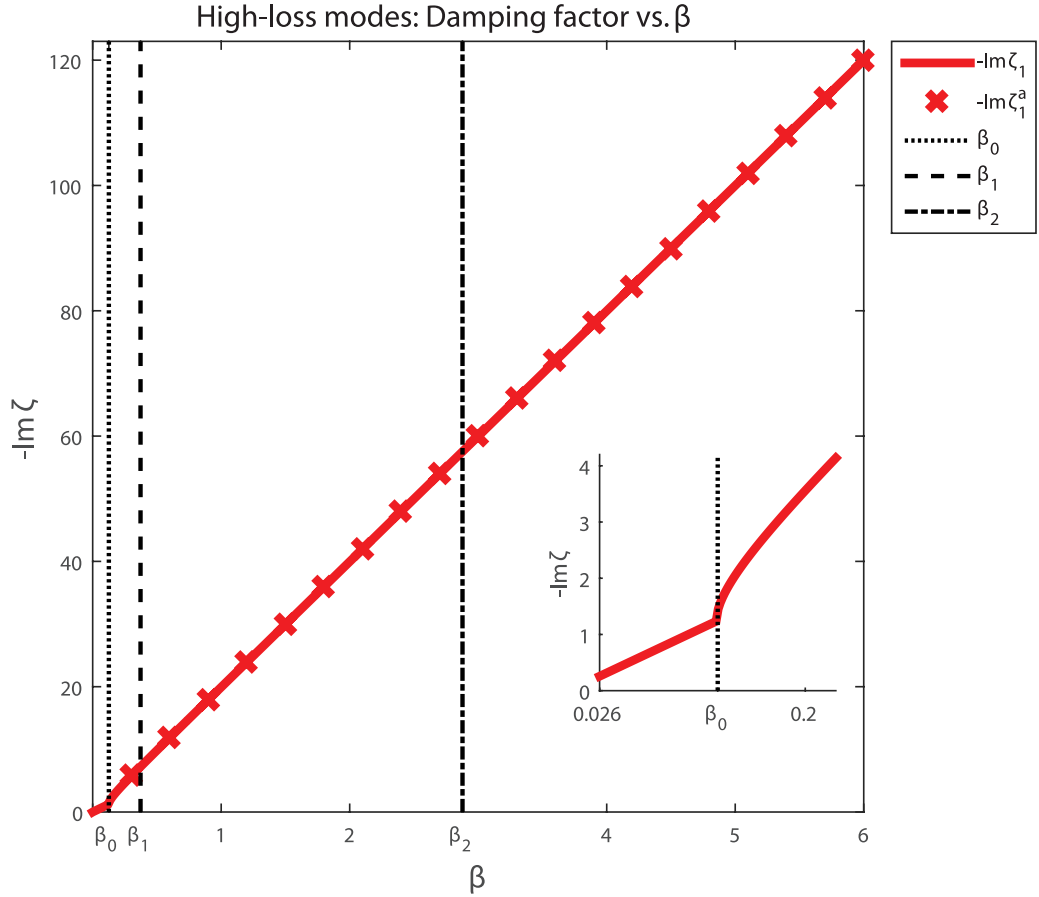}%
\caption{For the electric circuit with gyrator in Fig. \ref{Figc1} with the
numerical values (\ref{nae1}), the graph of damping factor $-\operatorname{Im}%
\zeta_{1}\left(  \beta\right)  $ of the high-loss eigenmodes with the
asymptotic approximation $-\operatorname{Im}\zeta_{1}^{a}\left(  \beta\right)
=b_{1}\beta$ in (\ref{circ8}) and (\ref{circllhqe0}) as a function of the loss
parameter $\beta=\frac{R_{2}}{\ell}$ (with $\ell$ fixed). In the inset, a
close-up of the damping factor in neighborhood of the value $\beta_{0}$ from
Theorem \ref{tbeta0} where the high-loss modes are guaranteed to be overdamped
for $\beta>\beta_{0}$.}%
\label{Fig_hldp}%
\end{center}
\end{figure}
\begin{figure}
[pth]
\begin{center}
\includegraphics[
height=4.9467in,
width=6.5956in
]%
{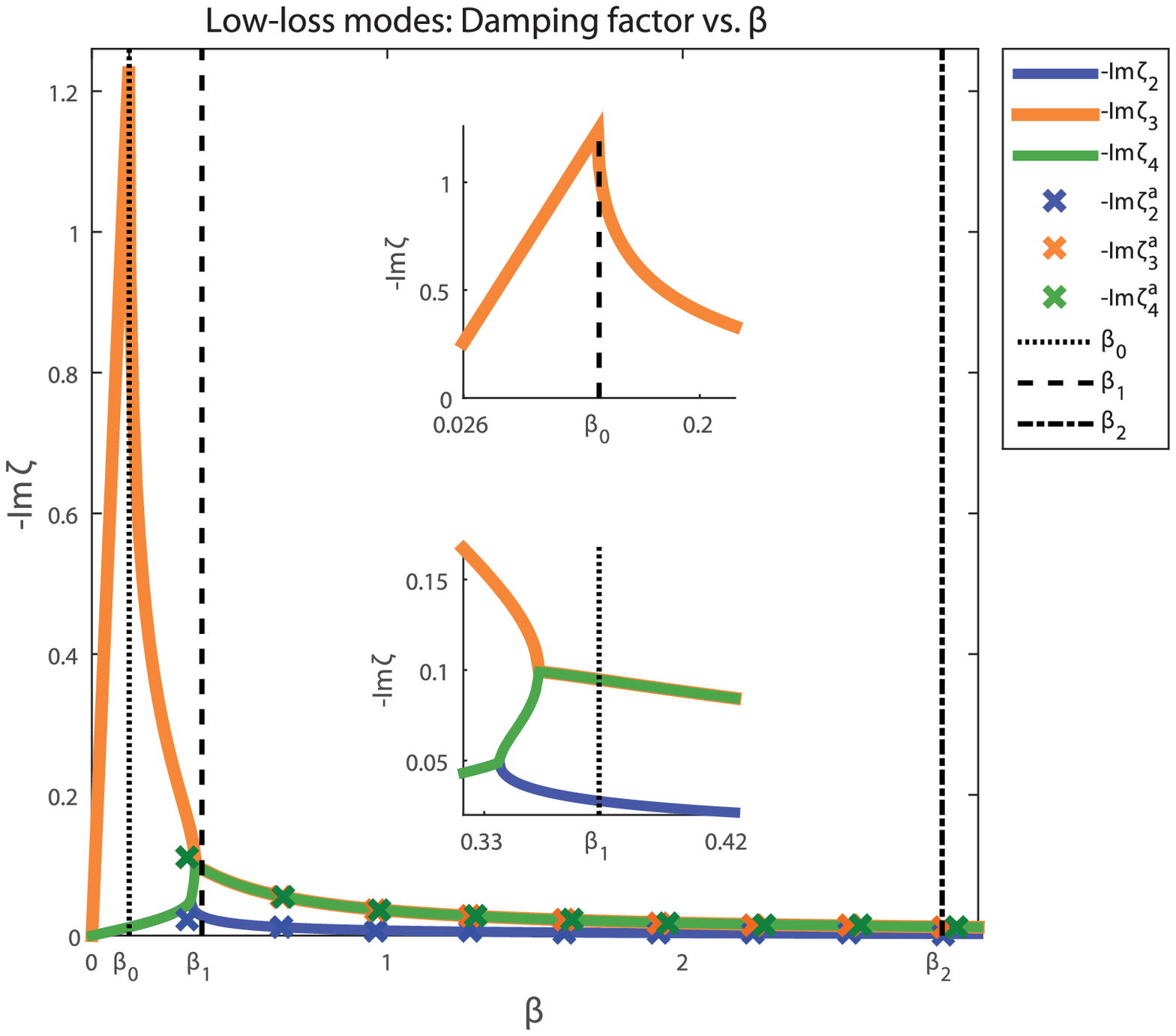}%
\caption{For the electric circuit with gyrator in Fig. \ref{Figc1} with the
numerical values (\ref{nae1}), the graph of damping factors
$-\operatorname{Im}\zeta_{j}\left(  \beta\right)  $, $j=2,3,4$ of the low-loss
eigenmodes with the asymptotic approximation $-\operatorname{Im}\zeta_{2}%
^{a}\left(  \beta\right)  =\frac{1}{b_{1}^{\flat}}\beta^{-1}$ in (\ref{circ9})
and (\ref{circllhqe1}) and $-\operatorname{Im}\zeta_{j}^{a}\left(
\beta\right)  =d_{j}\beta^{-1}$ in (\ref{circllhqd}) and (\ref{circllhqe2})
for $j=3,4$, as a function of the loss parameter $\beta=\frac{R_{2}}{\ell}$
(with $\ell$ fixed). Due to the spectral symmetry described in Proposition
\ref{pssym}, we always have $\left\{  \zeta_{1}\left(  \beta\right)
,\zeta_{2}\left(  \beta\right)  ,\zeta_{3}\left(  \beta\right)  ,\zeta
_{4}\left(  \beta\right)  \right\}  =\left\{  -\overline{\zeta_{1}\left(
\beta\right)  },-\overline{\zeta_{2}\left(  \beta\right)  },-\overline
{\zeta_{3}\left(  \beta\right)  },-\overline{\zeta_{4}\left(  \beta\right)
}\right\}  $ and so there can be intervals where the damping factors overlap,
as seen in this figure. In the bottom inset, a close-up of the damping factor
in neighborhood of the value $\beta_{1}$ from Corollary \ref{cbeta0} where the
overdamped low-loss modes, with damping factor $-\operatorname{Im}\zeta
_{2}\left(  \beta\right)  $ shown as the blue curve, are guaranteed to be
overdamped for $\beta>\beta_{1}$. Comparing this figure and the insets with
that of Fig. \ref{Fig_hldp}, one can see clearly the modal dichotomy near
$\beta=\beta_{0}$, between the high-loss modes with damping factor
$-\operatorname{Im}\zeta_{1}\left(  \beta\right)  $ and the low-loss modes
with damping factors $-\operatorname{Im}\zeta_{j}\left(  \beta\right)  $,
$j=2,3,4$ (as described by Theorem \ref{tmdic} and Corollary \ref{cmdic}) and
near $\beta=\beta_{1}$, between the low-loss, low-Q modes with damping factor
$-\operatorname{Im}\zeta_{2}\left(  \beta\right)  $ and the low-loss mode,
high-Q modes with damping factors $-\operatorname{Im}\zeta_{j}\left(
\beta\right)  $, $j=3,4$ (as described by Theorem \ref{tmddI} and Corollary
\ref{cmddI}).}%
\label{Fig_lldp}%
\end{center}
\end{figure}
\begin{figure}
[pth]
\begin{center}
\includegraphics[
height=4.9467in,
width=6.5956in
]%
{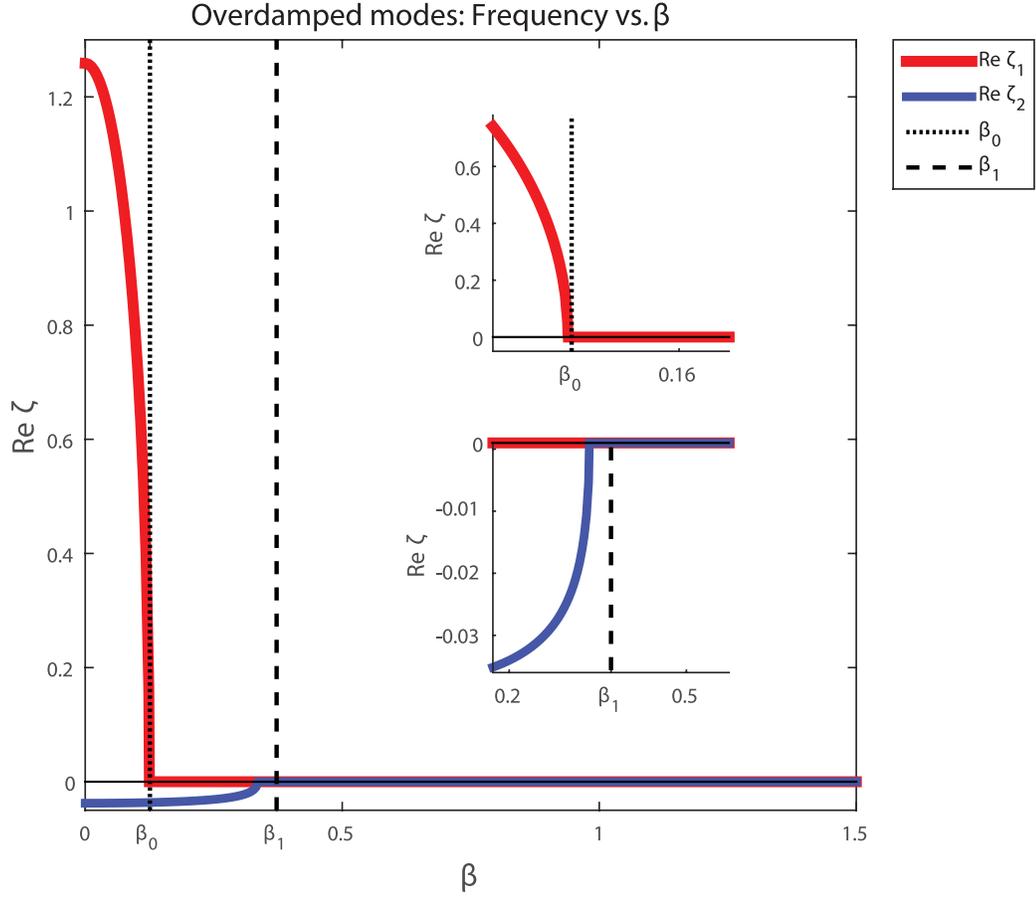}%
\caption{For the electric circuit with gyrator in Fig. \ref{Figc1} with the
numerical values (\ref{nae1}), the graph of frequencies $\operatorname{Re}%
\zeta_{j}\left(  \beta\right)  $, $j=1,2$ of the overdamped eigenmodes as a
function of the loss parameter $\beta=\frac{R_{2}}{\ell}$ (with $\ell$ fixed).
In the top inset, a close-up of the frequency in neighborhood of the value
$\beta_{0}$ from Theorem \ref{tbeta0} where the high-loss modes, with
frequency $\operatorname{Re}\zeta_{1}\left(  \beta\right)  $ shown as the red
curve, are guaranteed to be overdamped for $\beta>\beta_{0}$. This overdamping
is clearly visible since $\operatorname{Re}\zeta_{1}\left(  \beta\right)  =0$
for $\beta>\beta_{0}$ in this figure. In the bottom inset, a close-up of the
frequencies in neighborhood of the value $\beta_{1}$ from Corollary
\ref{cbeta0} where the overdamped low-loss modes, with frequency
$\operatorname{Re}\zeta_{2}\left(  \beta\right)  $ shown as the blue curve,
are guaranteed to be overdamped for $\beta>\beta_{1}$. This overdamping is
clearly visible since $\operatorname{Re}\zeta_{2}\left(  \beta\right)  =0$ for
$\beta>\beta_{1}$ in this figure.}%
\label{Fig_odmf}%
\end{center}
\end{figure}
\begin{figure}
[pth]
\begin{center}
\includegraphics[
height=4.9467in,
width=6.5956in
]%
{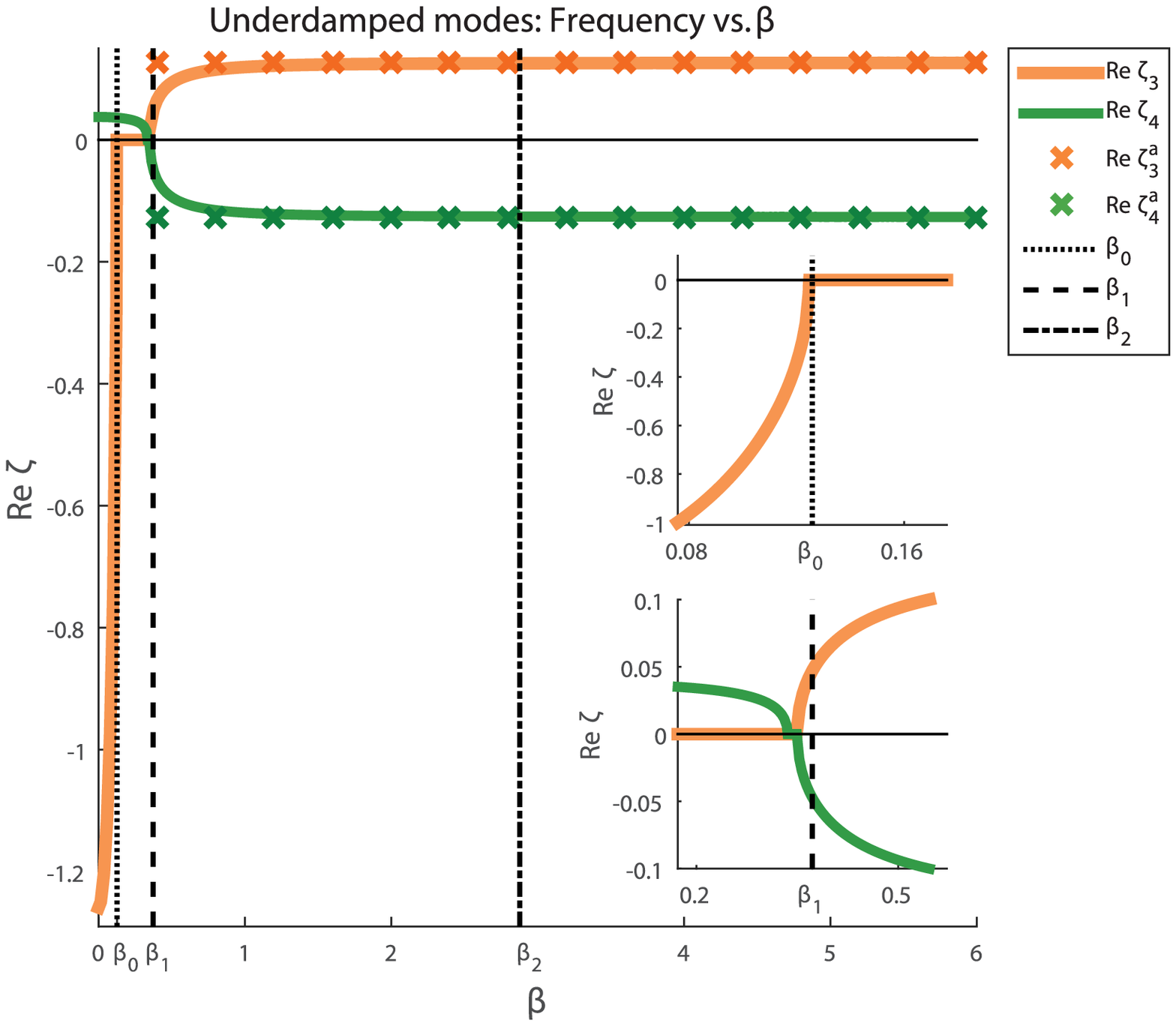}%
\caption{For the electric circuit with gyrator in Fig. \ref{Figc1} with the
numerical values (\ref{nae1}), the graph of frequencies $\operatorname{Re}%
\zeta_{j}\left(  \beta\right)  $, $j=3,4$ of the low-loss, high-Q eigenmodes
with the asymptotic approximation $\operatorname{Re}\zeta_{j}^{a}\left(
\beta\right)  =\rho_{j}$, $j=3,4$ in (\ref{circllhqp}) and (\ref{circllhqe2})
as a function of the loss parameter $\beta$. The low-loss, high-Q eigenmodes
are guaranteed to be underdamped for $\beta>\beta_{2}$ by Corollary
\ref{cllspmdd}, that is, $\operatorname{Re}\zeta_{j}\left(  \beta\right)
\not =0$ for $j=3,4$ as this figure indicates. Although outside the scope of
our studies here, it is interesting to point out some other interesting
phenomena which can be seen in these figures. For instance, in regards to the
low-loss eigenmodes, i.e., those with frequencies $\operatorname{Re}\zeta
_{j}\left(  \beta\right)  $, $j=2,3,4$, there can be open intervals in
$(\beta_{0},\beta_{1})$ where some of the frequencies are identically zero
such as for $\operatorname{Re}\zeta_{3}\left(  \beta\right)  $ near
$\beta=\beta_{0}$ (see top inset). And there is even an interval
$\mathcal{I}\subseteq(\beta_{0},\beta_{1})$, near $\beta=\beta_{1}$, where all
the eigenmode of the system have their frequencies identically zero, i.e.,
$\operatorname{Re}\zeta_{j}\left(  \beta\right)  \equiv0$, $j=1,2,3,4$ (see
bottom insets in this figure and in Fig. \ref{Fig_odmf}).}%
\label{Fig_udmf}%
\end{center}
\end{figure}
\begin{figure}
[pth]
\begin{center}
\includegraphics[
height=4.9476in,
width=6.5947in
]%
{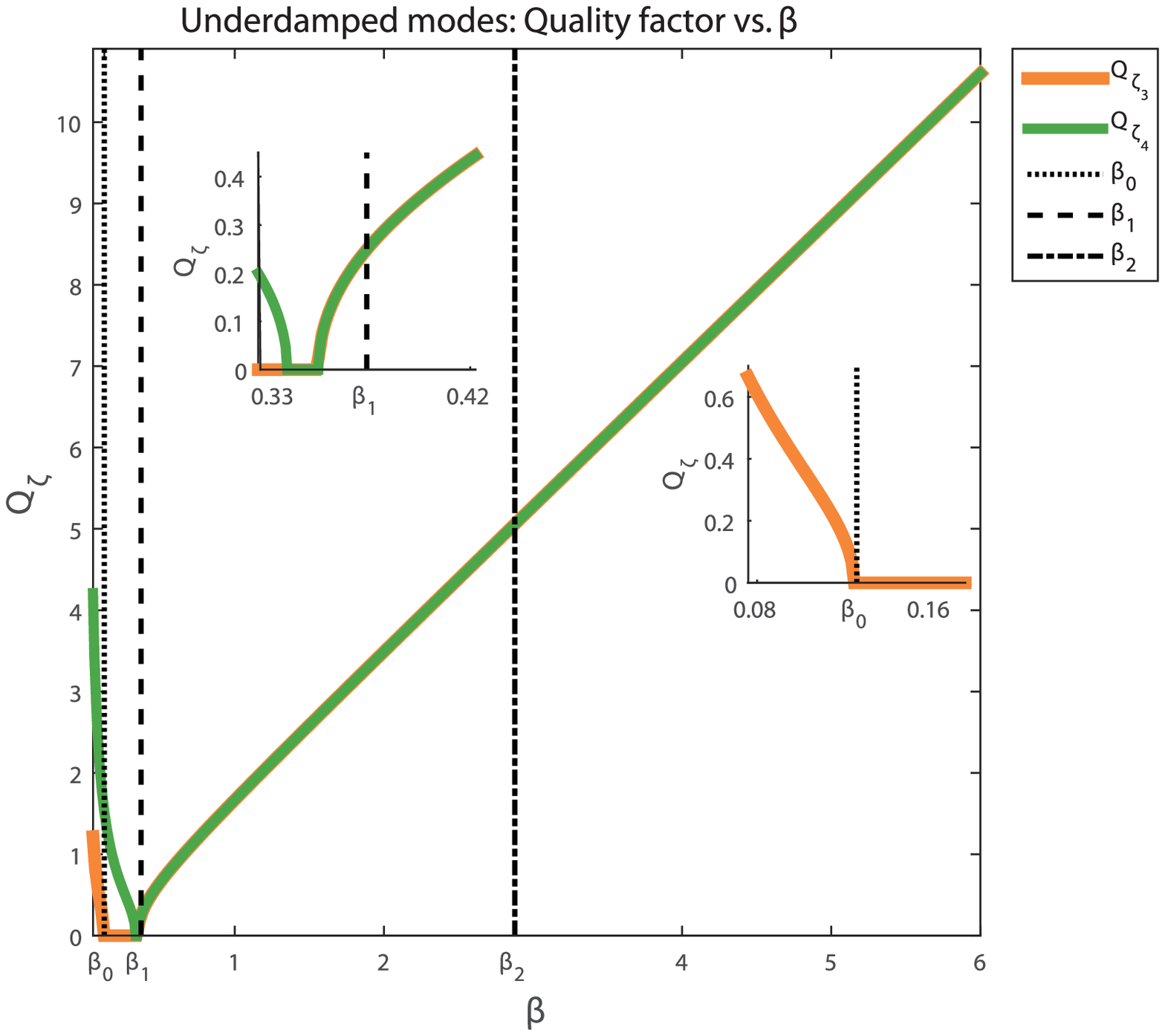}%
\caption{For the electric circuit with gyrator in Fig. \ref{Figc1} with the
numerical values (\ref{nae1}), the graph of quality factors (Q-factor)
$Q_{\zeta_{j}\left(  \beta\right)  }=\frac{1}{2}\frac{\left\vert
\operatorname{Re}\zeta_{j}\left(  \beta\right)  \right\vert }%
{-\operatorname{Im}\zeta_{j}\left(  \beta\right)  }$, $j=3,4$ of the low-loss,
high-Q eigenmodes as a function of the loss parameter $\beta$. These
eigenmodes are guaranteed to be underdamped for $\beta>\beta_{2}$ (by
Corollary \ref{cllspmdd}) and estimates on their dissipative properties in
terms of the modal dichotomy is described in Theorem \ref{tllspmd} and
Corollary \ref{cllspmd}. In particular, according to our theory $Q_{\zeta
_{j}\left(  \beta\right)  }\nearrow+\infty$ (i.e., increases without bound) as
$\beta\rightarrow\infty$ (see Corollary \ref{cllspmd}), which this figure
indicates. Due to the spectral symmetry described in Proposition \ref{pssym}
and Theorem \ref{tmddI} [see (\ref{tmddIsm})], we must have $\left\{
\zeta_{3}\left(  \beta\right)  ,\zeta_{4}\left(  \beta\right)  \right\}
=\left\{  -\overline{\zeta_{3}\left(  \beta\right)  },-\overline{\zeta
_{4}\left(  \beta\right)  }\right\}  $ for $\beta>\beta_{1}$. This offers an
explaination as to why we have intervals in which $Q_{\zeta_{3}\left(
\beta\right)  }=Q_{\zeta_{4}\left(  \beta\right)  }$ in this figure since in
those intervals $\zeta_{4}\left(  \beta\right)  =-\overline{\zeta_{3}\left(
\beta\right)  }$. In fact, this is guaranteed to be true for $\beta>\beta_{2}$
by Theorem \ref{tmddI} [see (\ref{tmddIsm})] and Corollary \ref{cllspmdd}.}%
\label{Fig_udmq}%
\end{center}
\end{figure}

\section{The Lagrangian system and its dual\label{smodel}}

In this paper, a linear \emph{Lagrangian system} will be a system whose state
is described by a time-dependent $Q=Q(t)$ taking values in the Hilbert space $%
\mathbb{C}
^{N}$ with the standard inner product $(\cdot,\cdot)$ (i.e., $(a,b)=a^{\ast}%
b$, where $\ast$ denotes the conjugate transpose, i.e., $a^{\ast}=\overline
{a}^{\mathrm{T}}$) whose dynamics are governed by the ODEs (\ref{sintro1}).
And associated with this system is its Lagrangian $\mathcal{L}$ in
(\ref{sintro2}) and its Rayleigh dissipation function $\mathcal{R}$ in
(\ref{sintro2a}).

\subsection{Dual system \label{subsDualSys}}

To the Lagrangian $\mathcal{L}$ (\ref{sintro2}) there is a corresponding
"dual" Lagrangian $\mathcal{L}^{\flat}$ defined by%
\begin{gather}
\mathcal{L}^{\flat}=\mathcal{L}^{\flat}\left(  Q,\dot{Q}\right)
=-\mathcal{L}\left(  \dot{Q},Q\right)  =\label{ddlag1}\\
=\frac{1}{2}\left[
\begin{array}
[c]{l}%
\dot{Q}\\
Q
\end{array}
\right]  ^{\mathrm{T}}\left[
\begin{array}
[c]{ll}%
\eta & \theta\\
\theta^{\mathrm{T}} & -\alpha
\end{array}
\right]  \left[
\begin{array}
[c]{l}%
\dot{Q}\\
Q
\end{array}
\right]  \qquad\text{(dual Lagrangian).}\nonumber
\end{gather}
From the definition of the dual, a fundamental property is $\left(
\mathcal{L}^{\flat}\right)  ^{\flat}=$ $\mathcal{L}$.

The general Euler-Lagrange equations of motion for the Lagrangian system with
Lagrangian $\mathcal{L}^{\flat}$ and Rayleigh dissipation function
$\mathcal{R}$ in (\ref{sintro2a}) is%
\begin{equation}
\eta\ddot{Q}+\left(  2\theta+\beta R\right)  \dot{Q}+\alpha Q=0\qquad
\text{(evolution equations for the dual system).} \label{dradis2a}%
\end{equation}

This linear Lagrangian system, whose states have dynamics governed by the ODEs
(\ref{dradis2a}), will be called the \emph{dual Lagrangian system}. This
system is obtained from our original Lagrangian system (\ref{sintro1}) by just
making the substitution $\left(  \alpha,\eta\right)  \longmapsto\left(
\eta,\alpha\right)  $. In particular, this means whenever $\eta^{-1}$ exists,
the dual Lagrangian $\mathcal{L}^{\flat}$ has all the same assumptions
satisfied as our Lagrangian $\mathcal{L}$ as well including the duality
condition (\ref{cnddl}) being true. Therefore, whenever the duality condition
(\ref{cnddl}) is true, i.e., $\eta>0$, then all results in this paper that
apply to the Lagrangian system (\ref{sintro1}) will also apply to it's dual
Lagrangian system (\ref{dradis2a}). The importance of this duality will become
clear when we study the modal dichotomy and overdamping phenomena.

We will also need to introduce the following notation convention.

\begin{notation}
In the rest of this paper, whenever we use the superscript notation $X^{\flat
}$ it will be implicitly understood that $X$ is an associated with the
Lagrangian system and $X^{\flat}$ is the same object but associated with dual
Lagrangian system, e.g., the dual Lagrangian $\mathcal{L}^{\flat}$, dual
Hamiltonian $\mathcal{H}^{\flat}$, dual quadratic matrix pencil $C^{\flat
}\left(  \zeta,\beta\right)  $, dual system operator $A^{\flat}\left(
\beta\right)  =\Omega^{\flat}-\mathrm{i}\beta B^{\flat}$, etc.
\end{notation}

Using this notation, the system energy for the dual Lagrangian system, which
has the Lagrangian $\mathcal{L}^{\flat}$ in (\ref{ddlag1}), can be calculated
in terms of the definition in (\ref{sintro3b}) of $\mathcal{V}$, $\mathcal{T}$
of the Lagrangian system (\ref{sintro1}) by:
\begin{gather}
\mathcal{L}^{\flat}=\mathcal{T}^{\flat}-\mathcal{V}^{\flat},\qquad
0\leq\mathcal{H}^{\flat}=\mathcal{T}^{\flat}\mathcal{+V}^{\flat}\text{
\ \ \ (dual Hamiltonian),}\label{dradis3}\\
\mathcal{T}^{\flat}=\mathcal{V}(Q,\dot{Q}),\qquad\mathcal{V}^{\flat
}=\mathcal{T}(Q,\dot{Q})\text{.}\nonumber
\end{gather}
Also from our definition and notation we have the Rayleigh dissipation
functions are the same, i.e., \
\begin{equation}
\mathcal{R}^{\flat}=\mathcal{R}\qquad\text{(dual Rayleigh dissipation
function).} \label{dradis3_a}%
\end{equation}

\textbf{On the eigenmodes and the quality factor of the dual system.} Given an
eigenmode $Q(t)=qe^{-\mathrm{i}\zeta t}$ (with $\zeta\not =0$) of the
Lagrangian system (\ref{sintro1}) it follows that $Q^{\flat}%
(t)=qe^{-\mathrm{i}\left(  -\zeta^{-1}\right)  t}$ is an eigenmode of the dual
Lagrangian system (\ref{dradis2a}) since%
\begin{equation}
0=\eta\ddot{Q}+\left(  2\theta+\beta R\right)  \dot{Q}+\alpha Q=\left(
-\mathrm{i}\zeta\right)  ^{2}\left(  \eta\ddot{Q}^{\flat}+\left(
2\theta+\beta R\right)  \dot{Q}^{\flat}+\alpha Q^{\flat}\right)  .
\label{dradis4_2}%
\end{equation}

For each eigenmode $Q$ with nonzero eigenfrequency $\zeta$, we will refer to
$Q^{\flat}$ as it's \emph{dual eigenmode}. With this definition, it follows
that \emph{the quality factor }$Q_{\zeta}$\emph{ of the eigenmode of the
Lagrangian system (\ref{sintro1}) is exactly the quality factor }%
$Q_{-\zeta^{-1}}$\emph{ of this dual eigenmode} \emph{of the dual Lagrangian
system }(\ref{dradis2a}) since by the definition of the quality factor in
(\ref{sintro6}) we have%
\begin{equation}
Q_{-\zeta^{-1}}=-\frac{1}{2}\frac{\left\vert \operatorname{Re}\left(
-\zeta^{-1}\right)  \right\vert }{\operatorname{Im}\left(  -\zeta^{-1}\right)
}=-\frac{1}{2}\frac{\left\vert \operatorname{Re}\zeta\right\vert
}{\operatorname{Im}\zeta}=Q_{\zeta}\qquad\text{(equivalence of the
Q-factors).} \label{dradis4_3}%
\end{equation}

\subsection{Standard versus pencil formulations of the spectral problems
\label{stvspen}}

In this section we elaborate on the relationship between the two main spectral
problems in this paper, namely, between the standard eigenvalue problem
(\ref{sevp}) for the system operator $A\left(  \beta\right)  =\Omega
-\mathrm{i}\beta B$ and the quadratic eigenvalue problem (\ref{qevp}) for the
quadratic matrix pencil $C\left(  \zeta,\beta\right)  $. We will also describe
some spectral properties of the system operator $A\left(  \beta\right)  $ and
that of the dual system operator $A^{\flat}\left(  \beta\right)  $. The
results in this section are need for our main results in Sec. \ref{smainr}\ on
the modal dichotomy and overdamping phenomena.

First, we begin with some notation. The Hilbert space $H=%
\mathbb{C}
^{2N}$ with standard inner product $\left(  \cdot,\cdot\right)  $ can be
decomposed as $H=H_{\mathrm{p}}\oplus H_{\mathrm{q}}$ into the orthogonal
subspaces $H_{\mathrm{p}}=%
\mathbb{C}
^{N}$, $H_{\mathrm{q}}=%
\mathbb{C}
^{N}$ with orthogonal matrix projections%
\begin{equation}
P_{\mathrm{p}}=\left[
\begin{array}
[c]{cc}%
\mathbf{1} & 0\\
0 & 0
\end{array}
\right]  ,\qquad P_{\mathrm{q}}=\left[
\begin{array}
[c]{cc}%
0 & 0\\
0 & \mathbf{1}%
\end{array}
\right]  . \label{pfsp8a}%
\end{equation}
In particular, the matrix $A\left(  \beta\right)  $ defined in (\ref{ceveqs}),
(\ref{ceveqs2}) is a block matrix already partitioned with respect to the
decomposition $H=H_{\mathrm{p}}\oplus H_{\mathrm{q}}$ and any vector $w\in H$
can be represented uniquely in the block form%
\begin{equation}
w=\left[
\begin{array}
[c]{c}%
\varphi\\
\psi
\end{array}
\right]  ,\qquad\text{where }\varphi=P_{\mathrm{p}}w,\qquad\psi=P_{\mathrm{q}%
}w. \label{pfsp8b}%
\end{equation}
With respect to this decomposition, we have the following results from
\cite{FigWel2}. First, the following proposition tells us the characteristic
matrix of the system operator $\zeta\mathbf{1}-A\left(  \beta\right)  $ can be
factored in terms of the quadratic matrix pencil $C\left(  \zeta,\beta\right)
$. Second, the corollary that follows gives the description of the spectral
equivalence between the two main spectral problems (\ref{qevp}) and
(\ref{sevp}).

\begin{proposition}
\label{ppfsp}If $\zeta\not =0$ then
\begin{gather}
\zeta\mathbf{1}-A\left(  \beta\right)  =\label{pfsp8}\\
=\left[
\begin{array}
[c]{cc}%
K_{\mathrm{p}} & \zeta^{-1}\mathrm{i}\Phi^{\mathrm{T}}\\
0 & \mathbf{1}%
\end{array}
\right]  \left[
\begin{array}
[c]{cc}%
\zeta^{-1}\mathbf{1} & 0\\
0 & \zeta\mathbf{1}%
\end{array}
\right]  \left[
\begin{array}
[c]{cc}%
C(\zeta,\beta) & 0\\
0 & \mathbf{1}%
\end{array}
\right]  \left[
\begin{array}
[c]{cc}%
K_{\mathrm{p}}^{\mathrm{T}} & 0\\
-\zeta^{-1}\mathrm{i}\Phi & \mathbf{1}%
\end{array}
\right]  .\nonumber
\end{gather}

\end{proposition}

\begin{corollary}
[spectral equivalence]\label{cpfsp}For any $\zeta\in%
\mathbb{C}
$,
\begin{equation}
\det\left(  \zeta\mathbf{1}-A\left(  \beta\right)  \right)  =\frac{\det
C(\zeta,\beta)}{\det\alpha}. \label{cpfsp1}%
\end{equation}
In particular, the system operator $A\left(  \beta\right)  $ and quadratic
matrix pencil $C(\zeta,\beta)$ have the same spectrum, i.e.,
\begin{equation}
\sigma\left(  A\left(  \beta\right)  \right)  =\sigma\left(  C\left(
\cdot,\beta\right)  \right)  . \label{cpfsp2}%
\end{equation}
Moreover, if $\zeta\not =0$ then the following statements are true:

\begin{enumerate}
\item If $A\left(  \beta\right)  w=\zeta w$ and $w\not =0$ then
\begin{equation}
w=\left[
\begin{array}
[c]{c}%
-\mathrm{i}\zeta\sqrt{\alpha}q\\
\sqrt{\eta}q
\end{array}
\right]  ,\qquad\text{where }C(\zeta,\beta)q=0,\qquad q\not =0. \label{cpfsp3}%
\end{equation}

\item If $C(\zeta,\beta)q=0$ and $q\not =0$ then
\begin{equation}
A\left(  \beta\right)  w=\zeta w\text{,}\qquad\text{where }w=\left[
\begin{array}
[c]{c}%
-\mathrm{i}\zeta\sqrt{\alpha}q\\
\sqrt{\eta}q
\end{array}
\right]  \not =0. \label{cpfsp4}%
\end{equation}

\end{enumerate}
\end{corollary}

The following lemma tells us that the eigenvalues of the system operator
$A\left(  \beta\right)  $ are nonzero whenever the duality condition
(\ref{cnddl}) is true and so the eigenvectors will have the unique block
representation in Corollary \ref{cpfsp}.

\begin{lemma}
\label{lpfsp}If (\ref{cnddl}) is true then $\Omega$ is invertible and
$A\left(  \beta\right)  $ is invertible. In particular, $0\not \in
\sigma\left(  A\left(  \beta\right)  \right)  $.
\end{lemma}

\begin{proof}
Suppose (\ref{cnddl}) is true, i.e., $\eta$ is invertible. Then $K$, as
defined in (\ref{chvar}) is invertible and hence $\Omega=\mathrm{i}%
KJK^{\mathrm{T}}$ is invertible, where $J\,$\ is the invertible $2N\times2N$
symplectic matrix $J=%
\begin{bmatrix}
0 & -\mathbf{1}\\
\mathbf{1} & 0
\end{bmatrix}
$. Next, if $A\left(  \beta\right)  w=0$ then $\Omega w=-\mathrm{i}\beta Bw$
which implies $\left(  w,\Omega w\right)  =-\mathrm{i}\beta\left(
w,Bw\right)  $. But since $\Omega^{\ast}=\Omega$, $B^{\ast}=B$, and $\beta$ is
real this implies $\left(  w,Bw\right)  =0$ and since $B\geq0$ this implies
$Bw=0$ which implies $\Omega w=0$. And since $\Omega$ was shown to be
invertible then $w=0$ and so $0\not \in \sigma\left(  A\left(  \beta\right)
\right)  $. This completes the proof.
\end{proof}

Whenever the duality condition (\ref{cnddl}) is true, the dual Lagrangian
system (\ref{dradis2a}) with (dual) system operator $A^{\flat}\left(
\beta\right)  $ has\ the corresponding (dual) quadratic matrix pencil
\begin{equation}
C^{\flat}(\zeta,\beta)=\zeta^{2}\eta+\left(  2\theta+\beta R\right)
\mathrm{i}\zeta-\alpha\qquad\text{(dual quadratic pencil).} \label{dpfsp2}%
\end{equation}
The following proposition describes the correspondence between the standard
eigenvalue problem (\ref{sevp}) for the system operator $A\left(
\beta\right)  $ and that of the dual system operator $A^{\flat}\left(
\beta\right)  $.

\begin{proposition}
[spectral equivalence-duality]\label{ppfspdl}Suppose (\ref{cnddl}) is true.
Then the following statements are true:

\begin{enumerate}
\item For any $\zeta\not =0$,
\begin{align}
C(\zeta,\beta)  &  =-\zeta^{2}C^{\flat}(-\zeta^{-1},\beta),\nonumber\\
\det\left(  \zeta\mathbf{1}-A\left(  \beta\right)  \right)   &  =\frac{\left(
-\zeta^{2}\right)  ^{N}\det\eta}{\det\alpha}\det\left(  \left(  -\zeta
^{-1}\right)  \mathbf{1}-A^{\flat}\left(  \beta\right)  \right) \\
&  =\left(  -1\right)  ^{N}\det\left(  \left(  -\zeta\right)  \mathbf{1}%
-\left[  A^{\flat}\left(  \beta\right)  \right]  ^{-1}\right)  .\nonumber
\end{align}

\item The system operator $A\left(  \beta\right)  $ is related to the spectrum
of dual system operator $A^{\flat}\left(  \beta\right)  $ by
\begin{equation}
\sigma\left(  A\left(  \beta\right)  \right)  =-\sigma\left(  A^{\flat}\left(
\beta\right)  \right)  ^{-1}.
\end{equation}

\item If $\zeta$ is an eigenvalue of $A\left(  \beta\right)  $ then
$-\zeta^{-1}$ is an eigenvalue of $A^{\flat}\left(  \beta\right)  $ and they
have the same geometric multiplicity, algebraic multiplicity, and partial
multiplicities (i.e., for the corresponding eigenvalue they have the same
Jordan normal form).

\item If $\zeta$, $w$ is an eigenpair of the system operator $A\left(
\beta\right)  $ then $-\zeta^{-1}$, $w^{\flat}$ is an eigenpair of the dual
system operator $A^{\flat}\left(  \beta\right)  $, where
\begin{gather}
w=\left[
\begin{array}
[c]{c}%
-\mathrm{i}\zeta\sqrt{\alpha}q\\
\sqrt{\eta}q
\end{array}
\right]  ,\qquad w^{\flat}=\left[
\begin{array}
[c]{c}%
-\mathrm{i}\left(  -\zeta^{-1}\right)  \sqrt{\eta}q\\
\sqrt{\alpha}q
\end{array}
\right]  ,\\
C(\zeta,\beta)q=0,\text{ }C^{\flat}(-\zeta^{-1},\beta)q=0,\text{\ \ }%
q\not =0.\nonumber
\end{gather}

\end{enumerate}
\end{proposition}

\begin{proof}
First, if $\zeta\not =0$ then by Corollary \ref{cpfsp} and duality we have
\begin{gather*}
C(\zeta,\beta)=\zeta^{2}\alpha+\left(  2\theta+\beta R\right)  \mathrm{i}%
\zeta-\eta\\
=-\zeta^{2}\left(  -\alpha+\left(  2\theta+\beta R\right)  \mathrm{i}\left(
-\zeta^{-1}\right)  +\left(  -\zeta^{-1}\right)  ^{2}\eta\right)  =-\zeta
^{2}C^{\flat}(-\zeta^{-1},\beta),\\
\det\left(  \zeta\mathbf{1}-A\left(  \beta\right)  \right)  =\frac{\det
C(\zeta,\beta)}{\det\alpha}=\frac{\left(  -\zeta^{2}\right)  ^{N}\det\eta
}{\det\alpha}\frac{\det C^{\flat}(-\zeta^{-1},\beta)}{\det\eta}\\
=\frac{\left(  -\zeta^{2}\right)  ^{N}\det\eta}{\det\alpha}\det\left(  \left(
-\zeta^{-1}\right)  \mathbf{1}-A^{\flat}\left(  \beta\right)  \right) \\
=\frac{\left(  -1\right)  ^{N}\det\eta}{\det\alpha}\det\left[  A^{\flat
}\left(  \beta\right)  \right]  \det\left(  \left(  -\zeta\right)
\mathbf{1}-\left[  A^{\flat}\left(  \beta\right)  \right]  ^{-1}\right) \\
=\left(  -1\right)  ^{N}\det\left(  \left(  -\zeta\right)  \mathbf{1}-\left[
A^{\flat}\left(  \beta\right)  \right]  ^{-1}\right)  .
\end{gather*}
It then follows immediately from this and Lemma \ref{lpfsp} that%
\[
\sigma\left(  A\left(  \beta\right)  \right)  =-\sigma\left(  A^{\flat}\left(
\beta\right)  \right)  ^{-1}.
\]
Next, it follows from Corollary \ref{cpfsp} and Lemma \ref{lpfsp} that if
$\zeta$, $w$ is an eigenpair of the system operator $A\left(  \beta\right)  $
then $\zeta\not =0$ and
\[
w=\left[
\begin{array}
[c]{c}%
-\mathrm{i}\zeta\sqrt{\alpha}q\\
\sqrt{\eta}q
\end{array}
\right]  ,\qquad C(\zeta,\beta)q=0,\text{\ \ }q\not =0.
\]
But this implies that
\[
C^{\flat}(-\zeta^{-1},\beta)q=\left(  -\zeta^{2}\right)  ^{-1}C(\zeta
,\beta)q=0,\qquad q\not =0
\]
implying by Corollary \ref{cpfsp} and duality that $-\zeta^{-1}$, $w^{\flat}$
is an eigenpair of the dual system operator $A^{\flat}\left(  \beta\right)  $
with
\[
w^{\flat}=\left[
\begin{array}
[c]{c}%
-\mathrm{i}\left(  -\zeta^{-1}\right)  \sqrt{\eta}q\\
\sqrt{\alpha}q
\end{array}
\right]  .
\]
This proves statements 1, 2, and 4. Finally, statement 3 follows immediately
from statement 1 and Lemma \ref{lpfsp}. This completes the proof.
\end{proof}

The next proposition from \cite{FigWel2} describes the spectral symmetries of
the system operator $A\left(  \beta\right)  $ which follows from its
fundamental property (\ref{ceveqs1}).

\begin{proposition}
[spectral symmetry]\label{pssym}The following statements are true:

\begin{enumerate}
\item The characteristic polynomial of $A\left(  \beta\right)  $ satisfies
\begin{equation}
\overline{\det\left(  -\overline{\zeta}\mathbf{1}-A\left(  \beta\right)
\right)  }=\det\left(  \zeta I-A\left(  \beta\right)  \right)  ,
\label{pssym1}%
\end{equation}
for every $\zeta\in%
\mathbb{C}
$. \ In particular, the spectrum $\sigma\left(  A\left(  \beta\right)
\right)  $ of the system operator $A\left(  \beta\right)  $ is symmetric with
respect to the imaginary axis of the complex plane, i.e.,%
\begin{equation}
\sigma\left(  A\left(  \beta\right)  \right)  =-\overline{\sigma\left(
A\left(  \beta\right)  \right)  }\text{.} \label{pssym2}%
\end{equation}

\item If $w$ is an eigenvector of the system operator $A\left(  \beta\right)
$ with corresponding eigenvalue $\zeta$ then $\overline{w}$ is an eigenvector
of $A$ with corresponding eigenvalue $-\overline{\zeta}$.
\end{enumerate}
\end{proposition}

The next proposition in this section relates the spectrum of matrices $\Omega
$, $B$ in the definition of the system operator $A\left(  \beta\right)
=\Omega-\mathrm{i}\beta B$ and the spectrum of the matrices $\Omega^{\flat}$,
$B^{\flat}$ for the dual system operator $A^{\flat}\left(  \beta\right)
=\Omega^{\flat}-\mathrm{i}\beta B^{\flat}$ to the matrices $\alpha$, $\eta$,
$\theta$, $R$.

\begin{proposition}
[spectra relations I]\label{pspredl}For the $2N\times2N$ matrices $B$,
$\Omega$ and the $N\times N$ matrices $\alpha$, $\eta$, $\theta$, $R$ we have
\begin{align}
\det\left(  \zeta\mathbf{1}-\Omega\right)   &  =\left(  \det\alpha\right)
^{-1}\det C\left(  \zeta,0\right)  =\left(  \det\alpha\right)  ^{-1}%
\det\left(  \zeta^{2}\alpha+2\theta\mathrm{i}\zeta-\eta\right)  ,\\
\det\left(  \zeta\mathbf{1}-B\right)   &  =\zeta^{N}\det\left(  \zeta
\alpha-R\right)  ,\nonumber
\end{align}
for every $\zeta\in%
\mathbb{C}
$ and%
\begin{equation}
\sigma\left(  \Omega\right)  =\sigma\left(  C\left(  \cdot,0\right)  \right)
\subseteq%
\mathbb{R}
,\text{\quad}\sigma\left(  \Omega\right)  =-\sigma\left(  \Omega\right)
,\text{\quad}\sigma\left(  B\right)  \setminus\left\{  0\right\}
=\sigma\left(  \alpha^{-1}R\right)  \setminus\left\{  0\right\}  .
\end{equation}
In particular, if $b_{\min}$ and $\omega_{\max}$ denote the smallest nonzero
eigenvalue and the largest eigenvalue of $B$ and $\Omega$, respectively, then
\begin{equation}
b_{\min}=\min\left[  \sigma\left(  \alpha^{-1}R\right)  \setminus\left\{
0\right\}  \right]  ,\qquad\omega_{\max}=\left\Vert \Omega\right\Vert
=\max\sigma\left(  C\left(  \cdot,0\right)  \right)  .
\end{equation}
Moreover, if (\ref{cnddl})\ is true then%
\begin{equation}
\sigma\left(  \Omega^{\flat}\right)  =-\sigma\left(  \Omega\right)
^{-1}=\sigma\left(  \Omega^{-1}\right)  ,\qquad\sigma\left(  B^{\flat}\right)
\setminus\left\{  0\right\}  =\sigma\left(  \eta^{-1}R\right)  \setminus
\left\{  0\right\}  .
\end{equation}
In particular, if $b_{\text{min}}^{\flat}$ and $\omega_{\text{max}}^{\flat}$
denote the smallest nonzero eigenvalue and the largest eigenvalue of
$B^{\flat}$ and $\Omega^{\flat}$, respectively, then%
\begin{equation}
b_{\min}^{\flat}=\min\left[  \sigma\left(  \eta^{-1}R\right)  \setminus
\left\{  0\right\}  \right]  ,\qquad\omega_{\max}^{\flat}=\left\Vert
\Omega^{\flat}\right\Vert =\left\Vert \Omega^{-1}\right\Vert =\omega_{\min
}^{-1},
\end{equation}
where%
\begin{equation}
\omega_{\min}=\min\left\vert \sigma\left(  \Omega\right)  \right\vert
\end{equation}
is the smallest positive eigenvalue of $\Omega$.
\end{proposition}

\begin{proof}
It follows by Corollary \ref{cpfsp} that%
\begin{gather*}
\det\left(  \zeta\mathbf{1}-\Omega\right)  =\det\left(  \zeta\mathbf{1}%
-A\left(  0\right)  \right)  =\\
=\left(  \det\alpha\right)  ^{-1}\det C\left(  \zeta,0\right)  =\left(
\det\alpha\right)  ^{-1}\det\left(  \zeta^{2}\alpha+2\theta\mathrm{i}%
\zeta-\eta\right)
\end{gather*}
and it follows immediately from the definition of $B$ that%
\[
\det\left(  \zeta\mathbf{1}-B\right)  =\zeta^{N}\det\left(  \zeta
\alpha-R\right)  =\zeta^{N}\left(  \det\alpha\right)  \det\left(
\zeta\mathbf{1}-\alpha^{-1}R\right)
\]
for every $\zeta\in%
\mathbb{C}
$. Now the $2N\times2N$ matrices $\Omega$, $B$ have the fundamental properties%
\[
\Omega=\Omega^{\ast}=-\Omega^{\mathrm{T}},\qquad\text{ }B\geq0,
\]
from which it follows that%
\begin{gather*}
\sigma\left(  \Omega\right)  =\sigma\left(  C\left(  \cdot,0\right)  \right)
\subseteq%
\mathbb{R}
,\qquad\sigma\left(  \Omega\right)  =-\sigma\left(  \Omega\right)  ,\\
\sigma\left(  B\right)  \setminus\left\{  0\right\}  =\sigma\left(
\alpha^{-1}R\right)  \setminus\left\{  0\right\}  ,
\end{gather*}
and%
\begin{gather*}
b_{\min}=\min\left[  \sigma\left(  \alpha^{-1}R\right)  \setminus\left\{
0\right\}  \right]  ,\\
\omega_{\max}=\left\Vert \Omega\right\Vert =\max\left\vert \sigma\left(
\Omega\right)  \right\vert =\max\left\vert \sigma\left(  \Omega\right)
\right\vert =\max\sigma\left(  C\left(  \cdot,0\right)  \right)  .
\end{gather*}
Suppose now that (\ref{cnddl})\ is true. Then from what we just proved we have%
\[
\sigma\left(  B^{\flat}\right)  \setminus\left\{  0\right\}  =\sigma\left(
\eta^{-1}R\right)  \setminus\left\{  0\right\}  ,
\]
and%
\[
\omega_{\max}^{\flat}=\left\Vert \Omega^{\flat}\right\Vert .
\]
In particular,
\[
b_{\min}^{\flat}=\min\left[  \sigma\left(  \eta^{-1}R\right)  \setminus
\left\{  0\right\}  \right]  .
\]
By Proposition \ref{ppfspdl} we have
\[
\sigma\left(  \Omega^{\flat}\right)  =\sigma\left(  A^{\flat}\left(  0\right)
\right)  =-\sigma\left(  A\left(  0\right)  \right)  ^{-1}=-\sigma\left(
\Omega\right)  ^{-1}.
\]
Then since $\sigma\left(  \Omega\right)  =-\sigma\left(  \Omega\right)  $ this
together with Lemma \ref{lpfsp}\ implies%
\[
\sigma\left(  \Omega^{\flat}\right)  =-\sigma\left(  \Omega\right)
^{-1}=\sigma\left(  \Omega^{-1}\right)  .
\]
From which it follows that%
\[
\omega_{\text{max}}^{\flat}=\left\Vert \Omega^{\flat}\right\Vert =\left\Vert
\Omega^{-1}\right\Vert =\left(  \min\left\vert \sigma\left(  \Omega\right)
\right\vert \right)  ^{-1}.
\]
But since $\sigma\left(  \Omega\right)  =-\sigma\left(  \Omega\right)  $ then%
\[
\min\left\vert \sigma\left(  \Omega\right)  \right\vert =\min\sigma\left(
\Omega\right)  =\omega_{\min}%
\]
is the smallest positive eigenvalue of $\Omega$. This completes the proof.
\end{proof}

\begin{remark}
[limiting frequencies]\label{rspredl2_0}This final proposition and the remark
that follows, describes the spectrum of an important self-adjoint operator
$\Omega_{1}$ on $\operatorname{Ker}B$ that plays a key role in the high-loss
regime $\beta\gg1$ in describing the modal dichotomy in Sec. \ref{sevmd} and
in describing the spectral asymptotics of $A\left(  \beta\right)  $ as
$\beta\rightarrow\infty$ in Sec. \ref{smdhlr}. In particular, if $\rho_{j}$,
$j=N_{R}+1,\ldots,2N$ are all the eigenvalues of $\Omega_{1}$ (repeated
according to their multiplicities as eigenvalues of $\Omega_{1}$) then in the
perturbation analysis of the frequencies $\operatorname{Re}\zeta_{j}\left(
\beta\right)  $, $j=N_{R}+1,\ldots,2N$ of the low-loss eigenmodes, as
described in Sec. \ref{smdhlr} [see (\ref{pod16})], we have the limiting
frequencies:
\begin{equation}
\lim_{\beta\rightarrow\infty}\operatorname{Re}\zeta_{j}\left(  \beta\right)
=\rho_{j}\text{, for }j=N_{R}+1,\ldots,2N;\text{ \ \ }\sigma\left(  \Omega
_{1}\right)  =\left\{  \rho_{j}:N_{R}+1\leq j\leq2N\right\}  .
\label{pspredl2lim}%
\end{equation}

\end{remark}

\begin{proposition}
[spectral relations II]\label{pspredl2}Let $P_{B}^{\bot}$ denote the
orthogonal projection of the Hilbert space $%
\mathbb{C}
^{2N}$, with standard inner product $\left(  \cdot,\cdot\right)  $, onto
$\operatorname{Ker}B$ (i.e., the nullspace of $B$). Denote by $\Omega
_{1}=\left.  P_{B}^{\bot}\Omega P_{B}^{\bot}\right\vert _{\operatorname{Ker}%
B}:\operatorname{Ker}B\rightarrow\operatorname{Ker}B$ the restriction of
$P_{B}^{\bot}\Omega P_{B}^{\bot}$ to $\operatorname{Ker}B$. For any operator
$M$ on a finite dimensional vector space over $%
\mathbb{C}
$, we will denote the product of its eigenvalues (counting multiplicities) by
$\det M$. Then the matrix $P_{B}^{\bot}\Omega P_{B}^{\bot}\not =0$ if and only
if $N_{R}<N$ (i.e., $R$ is rank deficient), and in this case the following
statements are true:

\begin{enumerate}
\item The nonzero eigenvalues of $P_{B}^{\bot}\Omega P_{B}^{\bot}$ are real
and come in $\pm$ pairs with equal multiplicity. In particular,
\begin{equation}
0\in\sigma\left(  P_{B}^{\bot}\Omega P_{B}^{\bot}\right)  =-\sigma\left(
P_{B}^{\bot}\Omega P_{B}^{\bot}\right)  \subseteq%
\mathbb{R}
.
\end{equation}

\item The operator $\Omega_{1}$ is self-adjoint with respect to the inner
product $\left(  \cdot,\cdot\right)  $ and for every $\rho\in%
\mathbb{C}
$,
\begin{equation}
\det\left(  \rho\mathbf{1}-\Omega_{1}\right)  =\rho^{N_{R}}\frac{\det\left[
\left.  P_{R}^{\bot}C\left(  \rho,0\right)  P_{R}^{\bot}\right\vert
_{\operatorname{Ker}R}\right]  }{\det\left(  \left.  P_{R}^{\bot}\alpha
P_{R}^{\bot}\right\vert _{\operatorname{Ker}R}\right)  }%
\end{equation}
where $P_{R}^{\bot}$ denotes the orthogonal projection onto
$\operatorname{Ker}R$,%
\begin{align}
\left.  P_{R}^{\bot}C\left(  \rho,0\right)  P_{R}^{\bot}\right\vert
_{\operatorname{Ker}R}  &  :\operatorname{Ker}R\rightarrow\operatorname{Ker}%
R,\\
\left.  P_{R}^{\bot}\alpha P_{R}^{\bot}\right\vert _{\operatorname{Ker}R}  &
:\operatorname{Ker}R\rightarrow\operatorname{Ker}R,
\end{align}
are the restriction of $P_{R}^{\bot}C\left(  \rho,0\right)  P_{R}^{\bot}$ and
$P_{R}^{\bot}\alpha P_{R}^{\bot}$ to $\operatorname{Ker}R$, and $\left.
P_{R}^{\bot}\alpha P_{R}^{\bot}\right\vert _{\operatorname{Ker}R}$ is
invertible. In particular,
\begin{equation}
\sigma\left(  P_{B}^{\bot}\Omega P_{B}^{\bot}\right)  =\sigma\left(
\Omega_{1}\right)  =\left\{  \rho\in%
\mathbb{C}
:\det\left[  \left.  P_{R}^{\bot}C\left(  \rho,0\right)  P_{R}^{\bot
}\right\vert _{\operatorname{Ker}R}\right]  =0\right\}  \cup\left\{
0\right\}  .
\end{equation}

\item If $\rho_{0}\not =0$ is an eigenvalue of $P_{B}^{\bot}\Omega P_{B}%
^{\bot}$ then its multiplicity is equal to the multiplicity as an eigenvalue
of $\Omega_{1}$ which is equal to the order of the zero of the polynomial (of
degree $2N-2N_{R}$) $\det\left[  \left.  P_{R}^{\bot}C\left(  \rho,0\right)
P_{R}^{\bot}\right\vert _{\operatorname{Ker}R}\right]  $ at $\rho=\rho_{0}$.

\item The eigenvalue $0$ of $\Omega_{1}$ has multiplicity $N_{R}$ if and only
if $\operatorname{Ker}R\cap\operatorname{Ker}\eta=\left\{  0\right\}  $, in
which case%
\begin{equation}
\dim\operatorname{Ker}\left(  \Omega_{1}\right)  =N_{R}\text{, }%
\dim\operatorname{Ran}\left(  P_{B}^{\bot}\Omega P_{B}^{\bot}\right)
=\dim\operatorname{Ran}\left(  \Omega_{1}\right)  =2\left(  N-N_{R}\right)  .
\end{equation}

\item If (\ref{cnddl}) is true then, denoting the corresponding dual operator
of $\Omega_{1}$ for the dual Lagrangian system (\ref{dradis2a}) by $\Omega
_{1}^{\flat}$, the following are true: i) $0$ is an eigenvalue of both
$\Omega_{1}^{\flat}$ and $\Omega_{1}$ of equal multiplicity $N_{R}$; (iii)
$\left.  P_{R}^{\bot}C\left(  \rho,0\right)  P_{R}^{\bot}\right\vert
_{\operatorname{Ker}R}=-\rho^{2}\left.  P_{R}^{\bot}C^{\flat}\left(
-\rho^{-1},0\right)  P_{R}^{\bot}\right\vert _{\operatorname{Ker}R}$ for
$\rho\not =0$; (iv) $\rho\not =0$ is an eigenvalue of $\Omega_{1}$ of
multiplicity $m$ if and only $-\rho^{-1}$ is an eigenvalue of $\Omega
_{1}^{\flat}$ of multiplicity $m$. In particular, if we define $\rho_{\min}$
and $\rho_{\max}$ by
\[
\rho_{\min}=\min\left[  \sigma\left(  \Omega_{1}\right)  \cap\left(
0,\infty\right)  \right]  ,\text{ }\rho_{\max}=\max\sigma\left(  \Omega
_{1}\right)  ,
\]
and letting $\rho_{\min}^{\flat}$, $\rho_{\max}^{\flat}$ denote these
corresponding values for the dual Lagrangian system (\ref{dradis2a}) then%
\[
\rho_{\min}^{\flat}=\rho_{\max}^{-1}\text{, }\rho_{\max}^{\flat}=\rho_{\min
}^{-1}.
\]

\end{enumerate}
\end{proposition}

\begin{proof}
Suppose that $P_{B}^{\bot}\Omega P_{B}^{\bot}\not =0$. From block matrix
representation of $B$ and $\Omega$ in (\ref{ceveqs2}) with respect to the
decomposition $%
\mathbb{C}
^{2N}=%
\mathbb{C}
^{N}\oplus%
\mathbb{C}
^{N}$, it follows that $P_{B}^{\bot}$, $P_{B}^{\bot}\Omega P_{B}^{\bot}$ have
the block matrix representation%
\begin{equation}
P_{B}^{\bot}=\left[
\begin{array}
[c]{ll}%
P_{\mathsf{\tilde{R}}}^{\bot} & 0\\
0 & \mathbf{1}%
\end{array}
\right]  ,\text{ }P_{B}^{\bot}\Omega P_{B}^{\bot}=\mathrm{i}%
\begin{bmatrix}
-2P_{\mathsf{\tilde{R}}}^{\bot}\sqrt{\alpha}^{-1}\theta\sqrt{\alpha}%
^{-1}P_{\mathsf{\tilde{R}}}^{\bot} & -P_{\mathsf{\tilde{R}}}^{\bot}%
\sqrt{\alpha}^{-1}\sqrt{\eta}\\
\sqrt{\eta}\sqrt{\alpha}^{-1}P_{\mathsf{\tilde{R}}}^{\bot} & 0
\end{bmatrix}
,\label{brepproj_1}%
\end{equation}
where $\mathbf{1}$ denotes the identity operator on\ $%
\mathbb{C}
^{N}$, $P_{\mathsf{\tilde{R}}}^{\bot}$ denotes the orthogonal projection onto
$\operatorname{Ker}\mathsf{\tilde{R}}$, and $\mathsf{\tilde{R}}=K_{\mathrm{p}%
}RK_{\mathrm{p}}^{\mathrm{T}}$ with $K_{\mathrm{p}}=\sqrt{\alpha}^{-1}$. Also,
as $B=B^{\ast}=B^{\mathrm{T}}$ this implies $P_{B}^{\bot}=\left(  P_{B}^{\bot
}\right)  ^{\ast}=\left(  P_{B}^{\bot}\right)  ^{\mathrm{T}}$ so that since
$\Omega=\Omega^{\ast}=-\Omega^{\mathrm{T}}$ then%
\[
P_{B}^{\bot}\Omega P_{B}^{\bot}=\left(  P_{B}^{\bot}\Omega P_{B}^{\bot
}\right)  ^{\ast}=-\left(  P_{B}^{\bot}\Omega P_{B}^{\bot}\right)
^{\mathrm{T}}.
\]
From these facts we have, $0$ is an eigenvalue of $P_{B}^{\bot}\Omega
P_{B}^{\bot}$ and it's nonzero eigenvalues are real and come in $\pm$ pairs
with equal multiplicity. In particular,
\[
0\in\sigma\left(  P_{B}^{\bot}\Omega P_{B}^{\bot}\right)  =-\sigma\left(
P_{B}^{\bot}\Omega P_{B}^{\bot}\right)  \subseteq%
\mathbb{R}
\text{.}%
\]
It also follows immediately that $\Omega_{1}$ is self-adjoint with respect to
the inner product $\left(  \cdot,\cdot\right)  $ since $P_{B}^{\bot}\Omega
P_{B}^{\bot}$ is.

Next, we will prove the operator identity $\sqrt{\alpha}^{-1}P_{\mathsf{\tilde
{R}}}^{\bot}\sqrt{\alpha}^{-1}=P_{R}^{\bot}\sqrt{\alpha}^{-1}P_{\mathsf{\tilde
{R}}}^{\bot}\sqrt{\alpha}^{-1}P_{R}^{\bot}$, where $P_{R}^{\bot}$ denotes the
orthogonal projection onto $\operatorname{Ker}R$. First, since
$P_{\mathsf{\tilde{R}}}^{\bot}$ and $P_{R}^{\bot}$ are orthogonal projections
onto $\operatorname{Ker}\mathsf{\tilde{R}}$ and $\operatorname{Ker}R$ for the
real symmetric matrices $\mathsf{\tilde{R}}$ and $R$, respectively, then
$\left(  P_{\mathsf{\tilde{R}}}^{\bot}\right)  ^{2}=P_{\mathsf{\tilde{R}}%
}^{\bot}=\left(  P_{\mathsf{\tilde{R}}}^{\bot}\right)  ^{\ast}$, $\left(
P_{R}^{\bot}\right)  ^{2}=P_{R}^{\bot}=\left(  P_{R}^{\bot}\right)  $,
$\operatorname{Ran}P_{\mathsf{\tilde{R}}}^{\bot}=\operatorname{Ker}%
\mathsf{\tilde{R}}$ $=\sqrt{\alpha}\operatorname{Ker}R$, $\operatorname{Ker}%
P_{\mathsf{\tilde{R}}}^{\bot}=\operatorname{Ran}\mathsf{\tilde{R}}$ $=\left(
\operatorname{Ker}\mathsf{\tilde{R}}\right)  ^{\bot}=\sqrt{\alpha}%
^{-1}\operatorname{Ran}R$, $\operatorname{Ran}P_{R}^{\bot}=\operatorname{Ker}%
R$, and $\operatorname{Ker}P_{R}^{\bot}=\operatorname{Ran}R=\left(
\operatorname{Ker}R\right)  ^{\bot}$. These facts imply that $P_{R}^{\bot
}\sqrt{\alpha}^{-1}P_{\mathsf{\tilde{R}}}^{\bot}=\sqrt{\alpha}^{-1}%
P_{\mathsf{\tilde{R}}}^{\bot}$ and so by taking complex conjugate transpose we
have $P_{\mathsf{\tilde{R}}}^{\bot}\sqrt{\alpha}^{-1}P_{R}^{\bot
}=P_{\mathsf{\tilde{R}}}^{\bot}\sqrt{\alpha}^{-1}$. This implies $\sqrt
{\alpha}^{-1}P_{\mathsf{\tilde{R}}}^{\bot}\sqrt{\alpha}^{-1}=P_{R}^{\bot}%
\sqrt{\alpha}^{-1}P_{\mathsf{\tilde{R}}}^{\bot}\sqrt{\alpha}^{-1}=\sqrt
{\alpha}^{-1}P_{\mathsf{\tilde{R}}}^{\bot}\sqrt{\alpha}^{-1}P_{R}^{\bot}$.
Multiply this identity by $P_{R}^{\bot}$ on the left implies that
$\sqrt{\alpha}^{-1}P_{\mathsf{\tilde{R}}}^{\bot}\sqrt{\alpha}^{-1}=P_{R}%
^{\bot}\sqrt{\alpha}^{-1}P_{\mathsf{\tilde{R}}}^{\bot}\sqrt{\alpha}^{-1}%
P_{R}^{\bot}$, which is the desired identity.

Next, it follows from this and the fact that $\sqrt{\alpha}^{-1}%
:\operatorname{Ker}\mathsf{\tilde{R}}\rightarrow\operatorname{Ker}R$ is an
invertible map with inverse $\sqrt{\alpha}:\operatorname{Ker}R\rightarrow
\operatorname{Ker}\mathsf{\tilde{R}}$ with that%
\begin{gather*}
P_{\mathsf{\tilde{R}}}^{\bot}K_{\mathrm{p}}C\left(  \rho,0\right)
K_{\mathrm{p}}^{\mathrm{T}}P_{\mathsf{\tilde{R}}}^{\bot}=\left(
P_{\mathsf{\tilde{R}}}^{\bot}\sqrt{\alpha}^{-1}C\left(  \rho,0\right)
\sqrt{\alpha}^{-1}P_{\mathsf{\tilde{R}}}^{\bot}\right) \\
=P_{\mathsf{\tilde{R}}}^{\bot}\sqrt{\alpha}^{-1}P_{R}^{\bot}C\left(
\rho,0\right)  P_{R}^{\bot}\sqrt{\alpha}^{-1}P_{\mathsf{\tilde{R}}}^{\bot
}=P_{\mathsf{\tilde{R}}}^{\bot}\sqrt{\alpha}^{-1}P_{R}^{\bot}\left(
P_{R}^{\bot}C\left(  \rho,0\right)  P_{R}^{\bot}\right)  |_{\operatorname{Ker}%
R}P_{R}^{\bot}\sqrt{\alpha}^{-1}P_{\mathsf{\tilde{R}}}^{\bot}%
\end{gather*}
where $\left(  P_{R}^{\bot}C\left(  \rho,0\right)  P_{R}^{\bot}\right)
|_{\operatorname{Ker}R}$ is the restrictions of $P_{R}^{\bot}C\left(
\rho,0\right)  P_{R}^{\bot}$ to $\operatorname{Ker}R$.

We will now prove that $\sqrt{\alpha}^{-1}P_{\mathsf{\tilde{R}}}^{\bot}%
\sqrt{\alpha}^{-1}:\operatorname{Ker}R\rightarrow\operatorname{Ker}R$ is
invertible with inverse $P_{R}^{\bot}\alpha P_{R}^{\bot}:\operatorname{Ker}%
R\rightarrow\operatorname{Ker}R$. First, $\sqrt{\alpha}^{-1}%
:\operatorname{Ker}\mathsf{\tilde{R}}\rightarrow\operatorname{Ker}R$ is an
invertible map with inverse $\sqrt{\alpha}:\operatorname{Ker}R\rightarrow
\operatorname{Ker}\mathsf{\tilde{R}}$. Second, we have
\begin{align*}
\operatorname{rank}\left(  \sqrt{\alpha}^{-1}P_{\mathsf{\tilde{R}}}^{\bot
}\sqrt{\alpha}^{-1}\right)   &  =\operatorname{rank}P_{\mathsf{\tilde{R}}%
}^{\bot}=\dim\operatorname{Ran}P_{\mathsf{\tilde{R}}}^{\bot}=\dim\left(
\operatorname{Ker}\mathsf{\tilde{R}}\right) \\
&  =\dim\left(  \sqrt{\alpha}\operatorname{Ker}R\right)  =\dim
\operatorname{Ker}R.
\end{align*}
It follows immediately from this that $\sqrt{\alpha}^{-1}P_{\mathsf{\tilde{R}%
}}^{\bot}\sqrt{\alpha}^{-1}:\operatorname{Ker}R\rightarrow:\operatorname{Ker}%
R$ is invertible. Next, it follows that $P_{\mathsf{\tilde{R}}}^{\bot}%
\sqrt{\alpha}P_{R}^{\bot}=\sqrt{\alpha}P_{R}^{\bot}$, since
$\operatorname{Ran}P_{\mathsf{\tilde{R}}}^{\bot}=\operatorname{Ker}%
\mathsf{\tilde{R}}$ $=\sqrt{\alpha}\operatorname{Ker}R=\operatorname{Ran}%
\sqrt{\alpha}P_{R}^{\bot}$ and $%
\mathbb{C}
^{N}=\operatorname{Ran}P_{R}^{\bot}\oplus\operatorname{Ker}P_{R}^{\bot}$, and
hence taking the complex conjugate transpose yields $P_{R}^{\bot}\sqrt{\alpha
}P_{\mathsf{\tilde{R}}}^{\bot}=$ $P_{R}^{\bot}\sqrt{\alpha}$ so that together
with the fact that $P_{R}^{\bot}\sqrt{\alpha}^{-1}P_{\mathsf{\tilde{R}}}%
^{\bot}=\sqrt{\alpha}^{-1}P_{\mathsf{\tilde{R}}}^{\bot}$ we have
\begin{gather*}
P_{R}^{\bot}\alpha P_{R}^{\bot}\sqrt{\alpha}^{-1}P_{\mathsf{\tilde{R}}}^{\bot
}\sqrt{\alpha}^{-1}=P_{R}^{\bot}\alpha\left(  \sqrt{\alpha}^{-1}%
P_{\mathsf{\tilde{R}}}^{\bot}\right)  \sqrt{\alpha}^{-1}\\
=P_{R}^{\bot}\sqrt{\alpha}P_{\mathsf{\tilde{R}}}^{\bot}\sqrt{\alpha}%
^{-1}=\left(  P_{R}^{\bot}\sqrt{\alpha}\right)  \sqrt{\alpha}^{-1}=P_{R}%
^{\bot},
\end{gather*}
which implies the inverse of $\sqrt{\alpha}^{-1}P_{\mathsf{\tilde{R}}}^{\bot
}\sqrt{\alpha}^{-1}:\operatorname{Ker}R\rightarrow:\operatorname{Ker}R$ is
$P_{R}^{\bot}\alpha P_{R}^{\bot}:\operatorname{Ker}R\rightarrow
\operatorname{Ker}R$.

Now suppose $\rho\not =0$. Then by the block representations of $\rho
\mathbf{1}-A\left(  0\right)  $ and $P_{B}^{\bot}$ in Proposition \ref{ppfsp}
and (\ref{brepproj_1}), we have%
\begin{gather}
\left[
\begin{array}
[c]{cc}%
P_{R}^{\bot}\alpha P_{R}^{\bot}\sqrt{\alpha}^{-1}P_{\mathsf{\tilde{R}}}^{\bot}
& 0\\
-\rho^{-1}\mathrm{i}\Phi P_{\mathsf{\tilde{R}}}^{\bot} & \mathbf{1}%
\end{array}
\right]  =\left[
\begin{array}
[c]{cc}%
P_{R}^{\bot}\alpha\sqrt{\alpha}^{-1}P_{\mathsf{\tilde{R}}}^{\bot} & 0\\
-\rho^{-1}\mathrm{i}\Phi P_{\mathsf{\tilde{R}}}^{\bot} & \mathbf{1}%
\end{array}
\right] \nonumber\\
=\left[
\begin{array}
[c]{cc}%
P_{R}^{\bot}\sqrt{\alpha}P_{\mathsf{\tilde{R}}}^{\bot} & 0\\
0 & \mathbf{1}%
\end{array}
\right]  \left[
\begin{array}
[c]{cc}%
\mathbf{1} & 0\\
-\rho^{-1}\mathrm{i}\Phi P_{\mathsf{\tilde{R}}}^{\bot} & \mathbf{1}%
\end{array}
\right] \nonumber
\end{gather}
and so defining%
\begin{align}
X  &  =\left[
\begin{array}
[c]{cc}%
\mathbf{1} & \rho^{-1}\mathrm{i}P_{\mathsf{\tilde{R}}}^{\bot}\Phi^{\mathrm{T}%
}\\
0 & \mathbf{1}%
\end{array}
\right]  \left[
\begin{array}
[c]{cc}%
P_{\mathsf{\tilde{R}}}^{\bot}\sqrt{\alpha}^{-1}P_{R}^{\bot} & 0\\
0 & \mathbf{1}%
\end{array}
\right]  ,\text{ }\label{FacIdOmega1a}\\
Y  &  =\left[
\begin{array}
[c]{cc}%
P_{R}^{\bot}\sqrt{\alpha}P_{\mathsf{\tilde{R}}}^{\bot} & 0\\
0 & \mathbf{1}%
\end{array}
\right]  \left[
\begin{array}
[c]{cc}%
\mathbf{1} & 0\\
-\rho^{-1}\mathrm{i}\Phi P_{\mathsf{\tilde{R}}}^{\bot} & \mathbf{1}%
\end{array}
\right]  ,\nonumber
\end{align}
we have%
\begin{gather}
\rho P_{B}^{\bot}-P_{B}^{\bot}\Omega P_{B}^{\bot}=P_{B}^{\bot}\left[
\rho\mathbf{1}-A\left(  0\right)  \right]  P_{B}^{\bot}=\label{FacIdOmega1}\\
=\left[
\begin{array}
[c]{cc}%
P_{\mathsf{\tilde{R}}}^{\bot} & \rho^{-1}\mathrm{i}P_{\mathsf{\tilde{R}}%
}^{\bot}\Phi^{\mathrm{T}}\\
0 & \mathbf{1}%
\end{array}
\right]  \left[
\begin{array}
[c]{cc}%
\rho^{-1}P_{\mathsf{\tilde{R}}}^{\bot} & 0\\
0 & \rho\mathbf{1}%
\end{array}
\right] \nonumber\\
\times\left[
\begin{array}
[c]{cc}%
P_{\mathsf{\tilde{R}}}^{\bot}K_{\mathrm{p}}C\left(  \rho,0\right)
K_{\mathrm{p}}^{\mathrm{T}}P_{\mathsf{\tilde{R}}}^{\bot} & 0\\
0 & \mathbf{1}%
\end{array}
\right]  \left[
\begin{array}
[c]{cc}%
P_{\mathsf{\tilde{R}}}^{\bot} & 0\\
-\rho^{-1}\mathrm{i}\Phi P_{\mathsf{\tilde{R}}}^{\bot} & \mathbf{1}%
\end{array}
\right] \nonumber\\
=X\left[
\begin{array}
[c]{cc}%
\rho^{-1}P_{R}^{\bot} & 0\\
0 & \rho\mathbf{1}%
\end{array}
\right]  \left[
\begin{array}
[c]{cc}%
\left(  P_{R}^{\bot}C\left(  \rho,0\right)  P_{R}^{\bot}\right)
|_{\operatorname{Ker}R}\left.  P_{R}^{\bot}\alpha P_{R}^{\bot}\right\vert
_{\operatorname{Ker}R}^{-1} & 0\\
0 & \mathbf{1}%
\end{array}
\right]  Y,\nonumber
\end{gather}
where the operator $P_{R}^{\bot}\sqrt{\alpha}^{-1}P_{\mathsf{\tilde{R}}}%
^{\bot}:\operatorname{Ker}\mathsf{\tilde{R}}\rightarrow\operatorname{Ker}R$ is
invertible whose inverse is the operator $P_{\mathsf{\tilde{R}}}^{\bot}%
\sqrt{\alpha}P_{R}^{\bot}:\operatorname{Ker}R\rightarrow\operatorname{Ker}%
\mathsf{\tilde{R}}$. The latter facts follow from facts that%
\begin{align*}
P_{R}^{\bot}\sqrt{\alpha}^{-1}P_{\mathsf{\tilde{R}}}^{\bot}\left(
P_{\mathsf{\tilde{R}}}^{\bot}\sqrt{\alpha}P_{R}^{\bot}\right)   &
=P_{R}^{\bot}\sqrt{\alpha}^{-1}\left(  P_{\mathsf{\tilde{R}}}^{\bot}%
\sqrt{\alpha}P_{R}^{\bot}\right)  =P_{R}^{\bot}\sqrt{\alpha}^{-1}\sqrt{\alpha
}P_{R}^{\bot}=P_{R}^{\bot},\\
\left(  P_{\mathsf{\tilde{R}}}^{\bot}\sqrt{\alpha}P_{R}^{\bot}\right)
P_{R}^{\bot}\sqrt{\alpha}^{-1}P_{\mathsf{\tilde{R}}}^{\bot}  &
=P_{\mathsf{\tilde{R}}}^{\bot}\sqrt{\alpha}\left(  P_{R}^{\bot}\sqrt{\alpha
}^{-1}P_{\mathsf{\tilde{R}}}^{\bot}\right)  =P_{\mathsf{\tilde{R}}}^{\bot
}\sqrt{\alpha}\sqrt{\alpha}^{-1}P_{\mathsf{\tilde{R}}}^{\bot}%
=P_{\mathsf{\tilde{R}}}^{\bot},
\end{align*}
which follow from the facts $P_{R}^{\bot}$ and $P_{\mathsf{\tilde{R}}}^{\bot}$
are projections onto $\operatorname{Ker}R$ and $\operatorname{Ker}%
\mathsf{\tilde{R}}$, respectively, together with identities $P_{\mathsf{\tilde
{R}}}^{\bot}\sqrt{\alpha}P_{R}^{\bot}=\sqrt{\alpha}P_{R}^{\bot}$ and
$P_{R}^{\bot}\sqrt{\alpha}^{-1}P_{\mathsf{\tilde{R}}}^{\bot}=\sqrt{\alpha
}^{-1}P_{\mathsf{\tilde{R}}}^{\bot}$. It follows immediately, from the
identities (\ref{FacIdOmega1}), (\ref{FacIdOmega1a}) and the fact that
$P_{R}^{\bot}\sqrt{\alpha}^{-1}P_{\mathsf{\tilde{R}}}^{\bot}%
:\operatorname{Ker}\mathsf{\tilde{R}}\rightarrow\operatorname{Ker}R$ and
$P_{\mathsf{\tilde{R}}}^{\bot}\sqrt{\alpha}P_{R}^{\bot}:\operatorname{Ker}%
R\rightarrow\operatorname{Ker}\mathsf{\tilde{R}}$ are inverses of each other,
that for all $\rho\in%
\mathbb{C}
,$
\begin{align*}
\det\left(  \rho\mathbf{1}-\Omega_{1}\right)   &  =\rho^{N_{R}}\det\left[
\left.  P_{R}^{\bot}C\left(  \rho,0\right)  P_{R}^{\bot}\right\vert
_{\operatorname{Ker}R}\left.  P_{R}^{\bot}\alpha P_{R}^{\bot}\right\vert
_{\operatorname{Ker}R}^{-1}\right] \\
&  =\rho^{N_{R}}\frac{\det\left[  \left.  P_{R}^{\bot}C\left(  \rho,0\right)
P_{R}^{\bot}\right\vert _{\operatorname{Ker}R}\right]  }{\det\left(  \left.
P_{R}^{\bot}\alpha P_{R}^{\bot}\right\vert _{\operatorname{Ker}R}\right)  }.
\end{align*}
In particular, this implies that if $\rho_{0}\not =0$ then $\rho=\rho_{0}$ is
an eigenvalue of $\Omega_{1}$ of multiplicity $m$ if and only if $\rho
=\rho_{0}$ is a zero of $\det\left[  \left(  P_{R}^{\bot}C\left(
\rho,0\right)  P_{R}^{\bot}\right)  |_{\operatorname{Ker}R}\right]  $ of
multiplicity $m$. And $\rho=0$ is an eigenvalue of $\Omega_{1}$ of
multiplicity $m$ if and only if $\rho=0$ is a zero of $\det\left[  \left(
P_{R}^{\bot}C\left(  \rho,0\right)  P_{R}^{\bot}\right)  |_{\operatorname{Ker}%
R}\right]  $ of multiplicity $m-N_{R}\geq0$. This proves the first three
statements of this theorem.

Now it follows that $\dim\operatorname{Ker}\left(  \Omega_{1}\right)  =N_{R}$
if and only if $\rho=0$ is not a zero of $\det\left[  \left(  P_{R}^{\bot
}C\left(  \rho,0\right)  P_{R}^{\bot}\right)  |_{\operatorname{Ker}R}\right]
$. And if the former is true then
\begin{gather*}
2N-N_{R}=\dim\operatorname{Ker}B=\dim\operatorname{Ker}\left(  \Omega
_{1}\right)  +\dim\operatorname{Ran}\Omega_{1}=N_{R}+\dim\operatorname{Ran}%
\left(  P_{B}^{\bot}\Omega P_{B}^{\bot}\right)  ,\\
\dim\operatorname{Ran}\left(  P_{B}^{\bot}\Omega P_{B}^{\bot}\right)
=2\left(  N-N_{R}\right)  .
\end{gather*}
Thus to complete the proof of the fourth statement of this theorem we need
only prove that\ $\rho=0$ is a zero of $\det\left[  \left(  P_{R}^{\bot
}C\left(  \rho,0\right)  P_{R}^{\bot}\right)  |_{\operatorname{Ker}R}\right]
$ if and only if $\operatorname{Ker}R\cap\operatorname{Ker}\eta\not =\left\{
0\right\}  $. If $x\in\operatorname{Ker}R\cap\operatorname{Ker}\eta$ with
$x\not =0$ then this implies $0=\det\left[  \left(  P_{R}^{\bot}\eta
P_{R}^{\bot}\right)  |_{\operatorname{Ker}R}\right]  =\det\left[  \left(
P_{R}^{\bot}C\left(  0,0\right)  P_{R}^{\bot}\right)  |_{\operatorname{Ker}%
R}\right]  $ and hence $\rho=0$ is a zero of $\det\left[  \left(  P_{R}^{\bot
}C\left(  \rho,0\right)  P_{R}^{\bot}\right)  |_{\operatorname{Ker}R}\right]
$. Conversely, if $\rho=0$ is a zero of $\det\left[  \left(  P_{R}^{\bot
}C\left(  \rho,0\right)  P_{R}^{\bot}\right)  |_{\operatorname{Ker}R}\right]
$ then
\[
0=\det\left[  \left(  P_{R}^{\bot}C\left(  0,0\right)  P_{R}^{\bot}\right)
|_{\operatorname{Ker}R}\right]  =\det\left[  \left(  P_{R}^{\bot}\eta
P_{R}^{\bot}\right)  |_{\operatorname{Ker}R}\right]  .
\]
This implies there exists $x\not =0$ with $x\in\operatorname{Ker}R$ such that
$0=P_{R}^{\bot}\eta P_{R}^{\bot}x=P_{R}^{\bot}\eta x$ implying that $\eta
x\in\operatorname{Ran}R$. But $\operatorname{Ran}R$ is orthogonal to
$\operatorname{Ker}R$ which implies $0=\left(  \eta x,x\right)  =\left\Vert
\sqrt{\eta}x\right\Vert ^{2}$ implying $\eta x=0$ so that $\operatorname{Ker}%
R\cap\operatorname{Ker}\eta$ with $x\not =0$. This proves the fourth
statement. The fifth statement now follows immediately from the first four
statements by duality.

Finally, it follows that $P_{B}^{\bot}\Omega P_{B}^{\bot}=0$ if and only if
$P_{R}^{\bot}C\left(  \rho,0\right)  P_{R}^{\bot}\equiv0$ if and only if
$P_{R}^{\bot}\alpha P_{R}^{\bot}=0$, where the latter equivalence follows from
the fact that since $\alpha>0$ then $P_{R}^{\bot}\alpha P_{R}^{\bot}=0$ if and
only $P_{R}^{\bot}=0$ if and only if $N_{R}=N$. This completes the proof of
the theorem.
\end{proof}

The following remark provides an interpret of Proposition (\ref{pspredl2}) in
terms of Lagrangian systems within our Lagrangian framework introduced in
\cite{FigWel2}, albeit slightly more abstractly as it is defined not on the
Euclidean space $%
\mathbb{R}
^{N}$, but on the finite-dimensional vector space $\operatorname{Ker}R\cap%
\mathbb{R}
^{N}$ over $%
\mathbb{R}
$ equipped with the dot product whose complexification is the vector space
$\operatorname{Ker}R$ over $%
\mathbb{C}
$ with standard inner product. Recall from Sec. \ref{sinmodel}, in our model
$\operatorname{Ker}R$ was associated with the lossless component of the
two-component composite system with a high-loss and a lossless component.

\begin{remark}
[Interpretation of the limiting frequencies]\label{rspredl2}Suppose that
$N_{R}<N$, i.e., $\operatorname{Ker}R\not =\left\{  0\right\}  $. Let
$P_{R}^{\bot}$ denote $N\times N$ matrix representing on $%
\mathbb{C}
^{N}$ the orthogonal projection onto $\operatorname{Ker}R$, in particular,
because $R$ is a real symmetric matrix this implies that $\left(  P_{R}^{\bot
}\right)  ^{\ast}=\left(  P_{R}^{\bot}\right)  ^{\mathrm{T}}=P_{R}^{\bot}$.
Now define on the vector space $\operatorname{Ker}R\cap%
\mathbb{R}
^{N}$ over $%
\mathbb{R}
$ equipped with the dot product, the Lagrangian%
\[
\mathcal{L}_{\operatorname{Ker}R}=\mathcal{L}_{\operatorname{Ker}R}\left(
Q,\dot{Q}\right)  =\frac{1}{2}\left[
\begin{array}
[c]{l}%
\dot{Q}\\
Q
\end{array}
\right]  ^{\mathrm{T}}\left[
\begin{array}
[c]{ll}%
\left.  P_{R}^{\bot}\alpha P_{R}^{\bot}\right\vert _{\operatorname{Ker}R} &
\left.  P_{R}^{\bot}\theta P_{R}^{\bot}\right\vert _{\operatorname{Ker}R}\\
\left.  P_{R}^{\bot}\theta^{\mathrm{T}}P_{R}^{\bot}\right\vert
_{\operatorname{Ker}R} & -\left.  P_{R}^{\bot}\eta P_{R}^{\bot}\right\vert
_{\operatorname{Ker}R}%
\end{array}
\right]  \left[
\begin{array}
[c]{l}%
\dot{Q}\\
Q
\end{array}
\right]  .
\]
Then all the assumptions of our Lagrangian framework in \cite{FigWel2} are
satisfied, namely, the operators satisfy
\begin{gather*}
\left.  P_{R}^{\bot}\alpha P_{R}^{\bot}\right\vert _{\operatorname{Ker}%
R}^{\mathrm{T}}=\left.  P_{R}^{\bot}\alpha P_{R}^{\bot}\right\vert
_{\operatorname{Ker}R}>0,\text{ }\left.  P_{R}^{\bot}\eta P_{R}^{\bot
}\right\vert _{\operatorname{Ker}R}^{\mathrm{T}}=\left.  P_{R}^{\bot}\eta
P_{R}^{\bot}\right\vert _{\operatorname{Ker}R}\geq0,\text{ }\\
\left.  P_{R}^{\bot}\theta P_{R}^{\bot}\right\vert _{\operatorname{Ker}%
R}^{\mathrm{T}}=-\left.  P_{R}^{\bot}\theta P_{R}^{\bot}\right\vert
_{\operatorname{Ker}R}.
\end{gather*}
This conservative Lagrangian system has the equations of motion given by the
Euler-Lagrange equations%
\[
\left.  P_{R}^{\bot}\alpha P_{R}^{\bot}\right\vert _{\operatorname{Ker}R}%
\ddot{Q}+2\left.  P_{R}^{\bot}\theta P_{R}^{\bot}\right\vert
_{\operatorname{Ker}R}\dot{Q}+\left.  P_{R}^{\bot}\eta P_{R}^{\bot}\right\vert
_{\operatorname{Ker}R}Q=0.
\]
The eigenmodes $Q\left(  t\right)  =qe^{-\mathrm{i}\rho t}$, $0\not =%
q\in\operatorname{Ker}R$ (with frequency $\rho$) of this Lagrangian system
correspond to the eigenpairs $\rho$, $q$ of the quadratic matrix pencil
\[
\left.  P_{R}^{\bot}C\left(  \rho,0\right)  P_{R}^{\bot}\right\vert
_{\operatorname{Ker}R}=\rho^{2}\left.  P_{R}^{\bot}\alpha P_{R}^{\bot
}\right\vert _{\operatorname{Ker}R}+2\left.  P_{R}^{\bot}\theta P_{R}^{\bot
}\right\vert _{\operatorname{Ker}R}\mathrm{i}\rho-\left.  P_{R}^{\bot}\eta
P_{R}^{\bot}\right\vert _{\operatorname{Ker}R}.
\]
Therefore, the set of eigenvalues of the pencil $\left.  P_{R}^{\bot}C\left(
\rho,0\right)  P_{R}^{\bot}\right\vert _{\operatorname{Ker}R}$ is the set%
\[
\sigma\left(  \left.  P_{R}^{\bot}C\left(  \cdot,0\right)  P_{R}^{\bot
}\right\vert _{\operatorname{Ker}R}\right)  =\left\{  \rho\in%
\mathbb{C}
:\det\left[  \left.  P_{R}^{\bot}C\left(  \rho,0\right)  P_{R}^{\bot
}\right\vert _{\operatorname{Ker}R}\right]  =0\right\}  ,
\]
which are exactly the frequencies $\rho$ of the conservative Lagrangian system
with Lagrangian $\mathcal{L}_{\operatorname{Ker}R}$.
\end{remark}

\section{Detailed statements of main results and their proofs\label{smainr}}

We provide here detailed statements of main results and their proofs.

\subsection{Modal dichotomy-duality\label{sevmd}}

In this section we will recall some results in \cite{FigWel2} on the modal
dichotomy on the spectrum $\sigma\left(  A\left(  \beta\right)  \right)
=\sigma\left(  C\left(  \cdot,\beta\right)  \right)  $ of the system operator
$A\left(  \beta\right)  =\Omega-\mathrm{i}\beta B$. We will then apply the
duality to achieve deeper results on this dichotomy which we describe below by
considering the spectrum $\sigma\left(  A^{\flat}\left(  \beta\right)
\right)  =\sigma\left(  C^{\flat}\left(  \cdot,\beta\right)  \right)  $ of the
dual system operator $A^{\flat}\left(  \beta\right)  =\Omega^{\flat
}-\mathrm{i}\beta B^{\flat}$ (whenever the duality condition \ref{cnddl} holds).

Denote the eigenvalues of $B$, listed in increasing order and indexed
according to their respective multiplicities, by $0=b_{0}<b_{1}\leq\ldots\leq
b_{N_{R}}$. In particular,
\begin{equation}
\sigma\left(  B\right)  =\left\{  b_{0},b_{1},\ldots,b_{N_{R}}\right\}  ,
\end{equation}
and by Proposition \ref{pspredl} we have
\begin{equation}
\sigma\left(  \alpha^{-1}R\right)  \setminus\{0\}=\left\{  b_{1}%
,\ldots,b_{N_{R}}\right\}  ,\text{ \ \ }b_{\min}:=b_{1}=\min\left[
\sigma\left(  \alpha^{-1}R\right)  \setminus\left\{  0\right\}  \right]  .
\end{equation}
Denote the largest eigenvalue of $\Omega$ by $\omega_{\max}$. By Proposition
\ref{pspredl} we have
\begin{equation}
\omega_{\max}=\left\Vert \Omega\right\Vert =\max\sigma\left(  C\left(
\cdot,0\right)  \right)  .
\end{equation}
Also, denote the discs centered at the eigenvalues of $-\mathrm{i}\beta B$
with radius $\omega_{\max}$ by%
\begin{equation}
D_{j}\left(  \beta\right)  =\left\{  \zeta\in%
\mathbb{C}
:\left\vert \zeta-(-\mathrm{i}\beta b_{j})\right\vert \leq\omega_{\max
}\right\}  ,\qquad0\leq j\leq N_{R}.
\end{equation}
Two\ subsets of the spectrum $\sigma\left(  A\left(  \beta\right)  \right)  $
which play a key role below are%
\begin{align}
\sigma_{0}\left(  A\left(  \beta\right)  \right)   &  =\sigma\left(  A\left(
\beta\right)  \right)  \cap D_{0}\left(  \beta\right)  ,\\
\sigma_{1}\left(  A\left(  \beta\right)  \right)   &  =\sigma\left(  A\left(
\beta\right)  \right)  \cap\bigcup\nolimits_{j=1}^{N_{R}}D_{j}\left(
\beta\right)  .\nonumber
\end{align}

The following result on eigenvalue bounds and clustering (from \cite{FigWel2})
is key to proving the modal dichotomy.

\begin{proposition}
[eigenvalue bounds \& clustering]\label{pevbd}For all $\beta\geq0$, the
following statements are true:

\begin{enumerate}
\item The eigenvalues of the system operator $A\left(  \beta\right)  $ lie in
the union of the closed discs whose centers are the eigenvalues of
$-\mathrm{i}\beta B$ with radius $\omega_{\max}$, that is,%
\begin{equation}
\sigma\left(  A\left(  \beta\right)  \right)  =\sigma_{0}\left(  A\left(
\beta\right)  \right)  \cup\sigma_{1}\left(  A\left(  \beta\right)  \right)  .
\end{equation}
In addition, these sets are symmetric with respect to the imaginary axis,
i.e.,%
\[
\sigma_{0}\left(  A\left(  \beta\right)  \right)  =-\overline{\sigma
_{0}\left(  A\left(  \beta\right)  \right)  },\qquad\sigma_{1}\left(  A\left(
\beta\right)  \right)  =-\overline{\sigma_{1}\left(  A\left(  \beta\right)
\right)  }.
\]

\item If $w\not =0$ and $A\left(  \beta\right)  w=\zeta w$ then%
\begin{equation}
\operatorname{Re}\zeta=\frac{\left(  w,\Omega w\right)  }{\left(  w,w\right)
},\qquad-\operatorname{Im}\zeta=\beta\frac{\left(  w,Bw\right)  }{\left(
w,w\right)  }\text{.}%
\end{equation}
In particular,
\begin{equation}
\left\vert \operatorname{Re}\zeta\right\vert \leq\omega_{\max},\qquad
0\leq-\operatorname{Im}\zeta\leq\beta\left\Vert B\right\Vert ,
\end{equation}
where%
\begin{equation}
\left\Vert B\right\Vert =\max\sigma\left(  B\right)  =\max\sigma\left(
\alpha^{-1}R\right)  .
\end{equation}

\end{enumerate}
\end{proposition}

\begin{proof}
This proposition except for the last two parts of both these statements were
proved in \cite{FigWel2}. To prove the symmetry in part 1 we have by
Proposition \ref{pssym} that
\[
\sigma\left(  A\left(  \beta\right)  \right)  =-\overline{\sigma\left(
A\left(  \beta\right)  \right)  },
\]
and by definition we have%
\[
D_{j}\left(  \beta\right)  =-\overline{D_{j}\left(  \beta\right)
},\text{\qquad}0\leq j\leq N_{R},
\]
so that%
\begin{align*}
-\overline{\sigma\left(  A\left(  \beta\right)  \right)  \cap D_{j}\left(
\beta\right)  }  &  =\left[  -\overline{\sigma\left(  A\left(  \beta\right)
\right)  }\right]  \cap\left[  -\overline{D_{j}\left(  \beta\right)  }\right]
=\\
&  =\sigma\left(  A\left(  \beta\right)  \right)  \cap D_{j}\left(
\beta\right)
\end{align*}
for $0\leq j\leq N_{R}$. In part 2, if $w\not =0$ and $A\left(  \beta\right)
w=\zeta w$ then%
\[
\left\vert \operatorname{Re}\zeta\right\vert =\left\vert \frac{\left(
w,\Omega w\right)  }{\left(  w,w\right)  }\right\vert \leq\left\Vert
\Omega\right\Vert =\omega_{\max}%
\]
by Proposition \ref{pspredl} and
\[
0\leq-\operatorname{Im}\zeta=\beta\frac{\left(  w,Bw\right)  }{\left(
w,w\right)  }\leq\beta\left\Vert B\right\Vert =\beta\max\sigma\left(
B\right)  =\beta\max\sigma\left(  \alpha^{-1}R\right)
\]
since $\beta\geq0$, $B\geq0$ and by Proposition \ref{pspredl}. This completes
the proof.
\end{proof}

The following main results on modal dichotomy were proved in \cite{FigWel2}.

\begin{theorem}
[modal dichotomy I]\label{tmdic}If $\beta>2\frac{\omega_{\max}}{b_{\min}}$
then
\begin{equation}
\sigma\left(  A\left(  \beta\right)  \right)  =\sigma_{0}\left(  A\left(
\beta\right)  \right)  \cup\sigma_{1}\left(  A\left(  \beta\right)  \right)
,\qquad\sigma_{0}\left(  A\left(  \beta\right)  \right)  \cap\sigma_{1}\left(
A\left(  \beta\right)  \right)  =\emptyset. \label{tmdic1}%
\end{equation}
Furthermore, there exists unique\ invariant subspaces $H_{\ell\ell}\left(
\beta\right)  $, $H_{h\ell}\left(  \beta\right)  $ of the system operator
$A\left(  \beta\right)  =$ $\Omega-\mathrm{i}\beta B$ with the properties%
\begin{align}
\mathrm{(i)}\quad H  &  =H_{h\ell}\left(  \beta\right)  \oplus H_{\ell\ell
}\left(  \beta\right)  ;\label{tmdic2}\\
\mathrm{(ii)}\quad\sigma\left(  A\left(  \beta\right)  |_{H_{\ell\ell}\left(
\beta\right)  }\right)   &  =\sigma_{0}\left(  A\left(  \beta\right)  \right)
,\qquad\sigma\left(  A\left(  \beta\right)  |_{H_{h\ell}\left(  \beta\right)
}\right)  =\sigma_{1}\left(  A\left(  \beta\right)  \right)  ,\nonumber
\end{align}
where $H=%
\mathbb{C}
^{2N}$. Moreover, the dimensions of these subspaces satisfy%
\begin{equation}
\dim H_{h\ell}\left(  \beta\right)  =N_{R},\ \ \dim H_{\ell\ell}\left(
\beta\right)  =2N-N_{R}. \label{tmdic3}%
\end{equation}

\end{theorem}

\begin{corollary}
[high-loss subspace: dissipative properties]\label{cmdic}If $\beta
>2\frac{\omega_{\max}}{b_{\min}}$ then spectrum of $A\left(  \beta\right)  $
can be partitioned in terms of the damping factors by%
\begin{align}
\sigma\left(  A\left(  \beta\right)  |_{H_{\ell\ell}\left(  \beta\right)
}\right)   &  =\left\{  \zeta\in\sigma\left(  A\left(  \beta\right)  \right)
:0\leq-\operatorname{Im}\zeta\leq\omega_{\max}\right\}  ,\\
\text{\ }\sigma\left(  A\left(  \beta\right)  |_{H_{h\ell}\left(
\beta\right)  }\right)   &  =\left\{  \zeta\in\sigma\left(  A\left(
\beta\right)  \right)  :-\operatorname{Im}\zeta\geq\beta b_{\min}-\omega
_{\max}>\omega_{\max}\right\}  .\nonumber
\end{align}
Moreover, maximum of the quality factors $Q_{\zeta}=\frac{1}{2}\frac
{\left\vert \operatorname{Re}\zeta\right\vert }{-\operatorname{Im}\zeta}$ for
$\zeta\in$ $\sigma\left(  A\left(  \beta\right)  |_{H_{h\ell}\left(
\beta\right)  }\right)  $ satisfy%
\begin{equation}
\max_{\zeta\in\sigma\left(  A\left(  \beta\right)  |_{H_{h\ell}\left(
\beta\right)  }\right)  }Q_{\zeta}\leq\frac{1}{2}\frac{\omega_{\max}}{\beta
b_{\min}-\omega_{\max}}<\frac{1}{2}.
\end{equation}

\end{corollary}

We will now use duality to further refine these results. Denote by $b_{\min
}^{\flat}$ the smallest nonzero eigenvalue of $B^{\flat}$ and by $\omega
_{\max}^{\flat}$ largest eigenvalue of $\Omega^{\flat}$. By Proposition
\ref{pspredl}
\begin{align}
b_{\min}^{\flat}  &  =\min\sigma\left(  B^{\flat}\right)  \setminus\left\{
0\right\}  =\min\sigma\left(  \eta^{-1}R\right)  \setminus\left\{  0\right\}
,\\
\omega_{\max}^{\flat}  &  =\left\Vert \Omega^{\flat}\right\Vert =\left\Vert
\Omega^{-1}\right\Vert =\omega_{\min}^{-1},
\end{align}
where $\omega_{\min}$ is the smallest positive eigenvalue of $\Omega$.

\begin{theorem}
[modal dichotomy I-duality]\label{tmddI}Suppose the condition (\ref{cnddl}) is
true. If $\beta>\max\left\{  2\frac{\omega_{\max}}{b_{\min}},2\frac
{\omega_{\max}^{\flat}}{b_{\min}^{\flat}}\right\}  $ then there exists
unique\ invariant subspaces $H_{\ell\ell,0}\left(  \beta\right)  $,
$H_{\ell\ell,1}\left(  \beta\right)  $ of the system operator $A\left(
\beta\right)  =$ $\Omega-\mathrm{i}\beta B$ with the properties%
\begin{align}
&  \mathrm{(i)}\text{ }H_{\ell\ell}\left(  \beta\right)  =H_{\ell\ell
,0}\left(  \beta\right)  \oplus H_{\ell\ell,1}\left(  \beta\right)  ;\\
&  \mathrm{(ii)}\text{ }\sigma\left(  A\left(  \beta\right)  |_{H_{\ell\ell
,0}\left(  \beta\right)  }\right)  =\sigma\left(  A\left(  \beta\right)
\right)  \cap\left\{  \zeta:\left\vert \zeta\right\vert <\omega_{\min
}\right\}  ;\\
&  \mathrm{(iii)}\text{ }\sigma\left(  A\left(  \beta\right)  |_{H_{\ell
\ell,1}\left(  \beta\right)  }\right)  =\sigma\left(  A\left(  \beta\right)
\right)  \cap\left\{  \zeta:\omega_{\min}\leq\left\vert \zeta\right\vert
\leq\omega_{\max}\right\}  .
\end{align}
Futhermore, the dimensions of these subspaces satisfy%
\begin{align}
\dim H_{\ell\ell,0}\left(  \beta\right)   &  =N_{R},\text{ \ }\\
\dim H_{\ell\ell,1}\left(  \beta\right)   &  =2N-2N_{R}.
\end{align}
Moreover, these invariant subspaces of the system operator $A\left(
\beta\right)  $ have the following properties:%
\begin{equation}
\sigma\left(  A\left(  \beta\right)  |_{H_{\ell\ell,0}\left(  \beta\right)
}\right)  =-\overline{\sigma\left(  A\left(  \beta\right)  |_{H_{\ell\ell
,0}\left(  \beta\right)  }\right)  }=-\sigma\left(  A^{\flat}\left(
\beta\right)  |_{H_{h\ell}^{\flat}\left(  \beta\right)  }\right)  ^{-1},
\end{equation}%
\begin{equation}
\sigma\left(  A\left(  \beta\right)  |_{H_{h\ell}\left(  \beta\right)
}\right)  =-\overline{\sigma\left(  A\left(  \beta\right)  |_{H_{h\ell}\left(
\beta\right)  }\right)  }=-\sigma\left(  A^{\flat}\left(  \beta\right)
|_{H_{\ell\ell,0}^{\flat}\left(  \beta\right)  }\right)  ^{-1},
\end{equation}%
\begin{equation}
\sigma\left(  A\left(  \beta\right)  |_{H_{\ell\ell,1}\left(  \beta\right)
}\right)  =-\overline{\sigma\left(  A\left(  \beta\right)  |_{H_{\ell\ell
,1}\left(  \beta\right)  }\right)  }=-\sigma\left(  A^{\flat}\left(
\beta\right)  |_{H_{\ell\ell,1}^{\flat}\left(  \beta\right)  }\right)  ^{-1}.
\label{tmddIsm}%
\end{equation}

\end{theorem}

\begin{proof}
Suppose the duality condition \ref{cnddl} holds and
\[
\beta>\max\left\{  2\frac{\omega_{\max}}{b_{\min}},2\frac{\omega_{\max}%
^{\flat}}{b_{\min}^{\flat}}\right\}  .
\]
Then by Theorem \ref{tmdic} and duality we have%
\begin{align*}
&  (1)\text{ \ \ }H=H_{h\ell}\left(  \beta\right)  \oplus H_{\ell\ell}\left(
\beta\right)  ;\\
&  (2)\text{ \ \ }\sigma\left(  A\left(  \beta\right)  |_{H_{\ell\ell}\left(
\beta\right)  }\right)  =\sigma_{0}\left(  A\left(  \beta\right)  \right)
,\text{\quad}\sigma\left(  A\left(  \beta\right)  |_{H_{h\ell}\left(
\beta\right)  }\right)  =\sigma_{1}\left(  A\left(  \beta\right)  \right)  ,\\
&  (1^{\flat})\text{ \ \ \ }H=H_{h\ell}^{\flat}\left(  \beta\right)  \oplus
H_{\ell\ell}^{\flat}\left(  \beta\right)  ;\\
&  (2^{\flat})\text{ }\sigma\left(  A^{\flat}\left(  \beta\right)
|_{H_{\ell\ell}^{\flat}\left(  \beta\right)  }\right)  =\sigma_{0}^{\flat
}\left(  A^{\flat}\left(  \beta\right)  \right)  ,\text{\quad}\sigma\left(
A^{\flat}\left(  \beta\right)  |_{H_{h\ell}^{\flat}\left(  \beta\right)
}\right)  =\sigma_{1}^{\flat}\left(  A^{\flat}\left(  \beta\right)  \right)
\end{align*}
where $H=%
\mathbb{C}
^{2N}$. Moreover, the dimensions of these subspaces satisfy%
\[
\dim H_{h\ell}\left(  \beta\right)  =\dim H_{h\ell}^{\flat}\left(
\beta\right)  =N_{R},\ \ \dim H_{\ell\ell}\left(  \beta\right)  =\dim
H_{\ell\ell}^{\flat}\left(  \beta\right)  =2N-N_{R}.
\]
It follows from the fact that%
\begin{gather*}
\sigma\left(  A\left(  \beta\right)  \right)  =\sigma_{0}\left(  A\left(
\beta\right)  \right)  \cup\sigma_{1}\left(  A\left(  \beta\right)  \right)
,\text{ }\sigma_{0}\left(  A\left(  \beta\right)  \right)  \cap\sigma
_{1}\left(  A\left(  \beta\right)  \right)  =\emptyset,\\
\sigma_{0}\left(  A\left(  \beta\right)  \right)  =\sigma\left(  A\left(
\beta\right)  \right)  \cap D_{0}\left(  \beta\right)  =\sigma\left(  A\left(
\beta\right)  \right)  \cap\left\{  \zeta:\left\vert \zeta\right\vert
\leq\omega_{\max}\right\}
\end{gather*}
that%
\begin{align*}
\sigma\left(  A\left(  \beta\right)  |_{H_{h\ell}\left(  \beta\right)
}\right)   &  =\sigma\left(  A\left(  \beta\right)  \right)  \cap\left\{
\zeta:\left\vert \zeta\right\vert >\omega_{\max}\right\}  ,\\
\sigma\left(  A\left(  \beta\right)  |_{H_{\ell\ell}\left(  \beta\right)
}\right)   &  =\sigma\left(  A\left(  \beta\right)  \right)  \cap\left\{
\zeta:\left\vert \zeta\right\vert \leq\omega_{\max}\right\}  .
\end{align*}
Thus we can partition the spectrum $\sigma\left(  A\left(  \beta\right)
|_{H_{\ell\ell}\left(  \beta\right)  }\right)  $ into the two sets%
\begin{align*}
&  \sigma\left(  A\left(  \beta\right)  \right)  \cap\left\{  \zeta:\left\vert
\zeta\right\vert <\omega_{\min}\right\}  ;\\
&  \sigma\left(  A\left(  \beta\right)  \right)  \cap\left\{  \zeta
:\omega_{\min}\leq\left\vert \zeta\right\vert \leq\omega_{\max}\right\}  .
\end{align*}
In follows from elementary linear algebra that to these two sets there exists
two unique invariant subspaces $H_{\ell\ell,0}\left(  \beta\right)  $,
$H_{\ell\ell,1}\left(  \beta\right)  $ of $A\left(  \beta\right)  $ with the
properties
\begin{align*}
&  \mathrm{(i)}\text{ }H_{\ell\ell}\left(  \beta\right)  =H_{\ell\ell
,0}\left(  \beta\right)  \oplus H_{\ell\ell,1}\left(  \beta\right)  ;\\
&  \mathrm{(ii)}\text{ }\sigma\left(  A\left(  \beta\right)  |_{H_{\ell\ell
,0}\left(  \beta\right)  }\right)  =\sigma\left(  A\left(  \beta\right)
\right)  \cap\left\{  \zeta:\left\vert \zeta\right\vert <\omega_{\min
}\right\}  ;\\
&  \mathrm{(iii)}\text{ }\sigma\left(  A\left(  \beta\right)  |_{H_{\ell
\ell,1}\left(  \beta\right)  }\right)  =\sigma\left(  A\left(  \beta\right)
\right)  \cap\left\{  \zeta:\omega_{\min}\leq\left\vert \zeta\right\vert
\leq\omega_{\max}\right\}  .
\end{align*}
In particular, $H_{\ell\ell,0}\left(  \beta\right)  $ and $H_{\ell\ell
,1}\left(  \beta\right)  $ are the union of the algebraic eigenspaces of
$A\left(  \beta\right)  $ corresponding to the eigenvalues in the set
$\sigma\left(  A\left(  \beta\right)  \right)  \cap\left\{  \zeta:\left\vert
\zeta\right\vert <\omega_{\min}\right\}  $ and $\sigma\left(  A\left(
\beta\right)  \right)  \cap\left\{  \zeta:\omega_{\min}\leq\left\vert
\zeta\right\vert \leq\omega_{\max}\right\}  $, respectively$.$

Now by Proposition \ref{pssym} we know that
\[
\sigma\left(  A\left(  \beta\right)  \right)  =-\overline{\sigma\left(
A\left(  \beta\right)  \right)  }%
\]
and since $\left\vert \zeta\right\vert =\left\vert -\overline{\zeta
}\right\vert $ then these fact imply that%
\[
\sigma\left(  A\left(  \beta\right)  |_{H_{\ell\ell,0}\left(  \beta\right)
}\right)  =-\overline{\sigma\left(  A\left(  \beta\right)  |_{H_{\ell\ell
,0}\left(  \beta\right)  }\right)  },
\]%
\[
\sigma\left(  A\left(  \beta\right)  |_{H_{h\ell}\left(  \beta\right)
}\right)  =-\overline{\sigma\left(  A\left(  \beta\right)  |_{H_{h\ell}\left(
\beta\right)  }\right)  },
\]%
\[
\sigma\left(  A\left(  \beta\right)  |_{H_{\ell\ell,1}\left(  \beta\right)
}\right)  =-\overline{\sigma\left(  A\left(  \beta\right)  |_{H_{\ell\ell
,1}\left(  \beta\right)  }\right)  }.
\]
And by Proposition \ref{ppfspdl} we know that%
\[
\sigma\left(  A^{\flat}\left(  \beta\right)  \right)  =-\sigma\left(  A\left(
\beta\right)  \right)  ^{-1}.
\]
Now we begin by proving that%
\[
\sigma\left(  A\left(  \beta\right)  |_{H_{\ell\ell,0}\left(  \beta\right)
}\right)  =-\sigma\left(  A^{\flat}\left(  \beta\right)  |_{H_{h\ell}^{\flat
}\left(  \beta\right)  }\right)  ^{-1}.
\]
Let $\zeta\in\sigma\left(  A\left(  \beta\right)  |_{H_{\ell\ell,0}\left(
\beta\right)  }\right)  $. Then $\left\vert \zeta\right\vert <\omega_{\min}$
so that $-\zeta^{-1}\in\sigma\left(  A^{\flat}\left(  \beta\right)  \right)  $
and $\left\vert -\zeta^{-1}\right\vert >\omega_{\min}^{-1}=\omega_{\max
}^{\flat}$. Thus, by what we have proven for $A\left(  \beta\right)  $ in this
statement already which by duality is true for $A^{\flat}\left(  \beta\right)
$ as well, we cannot have $-\zeta^{-1}$ in $\sigma\left(  A^{\flat}\left(
\beta\right)  |_{H_{\ell\ell}^{\flat}\left(  \beta\right)  }\right)  $
implying we must have $-\zeta^{-1}\in$ $\sigma\left(  A^{\flat}\left(
\beta\right)  |_{H_{h\ell}^{\flat}\left(  \beta\right)  }\right)  $. This
proves that%
\[
-\sigma\left(  A\left(  \beta\right)  |_{H_{\ell\ell,0}\left(  \beta\right)
}\right)  ^{-1}\subseteq\sigma\left(  A^{\flat}\left(  \beta\right)
|_{H_{h\ell}^{\flat}\left(  \beta\right)  }\right)  .
\]
We will now prove the reverse inclusion. Let $\zeta\in\sigma\left(  A^{\flat
}\left(  \beta\right)  |_{H_{h\ell}^{\flat}\left(  \beta\right)  }\right)  $.
Then $\left\vert \zeta\right\vert >\omega_{\max}^{\flat}$ and $-\zeta^{-1}%
\in\sigma\left(  A\left(  \beta\right)  \right)  $. Hence, $\left\vert
-\zeta^{-1}\right\vert <\left(  \omega_{\max}^{\flat}\right)  ^{-1}%
=\omega_{\min}$ implying $-\zeta^{-1}\in\sigma\left(  A\left(  \beta\right)
|_{H_{\ell\ell,0}\left(  \beta\right)  }\right)  $. That is, $\zeta\in
-\sigma\left(  A\left(  \beta\right)  |_{H_{\ell\ell,0}\left(  \beta\right)
}\right)  ^{-1}$. Thus,
\[
\sigma\left(  A^{\flat}\left(  \beta\right)  |_{H_{h\ell}^{\flat}\left(
\beta\right)  }\right)  \subseteq-\sigma\left(  A\left(  \beta\right)
|_{H_{\ell\ell,0}\left(  \beta\right)  }\right)  ^{-1}.
\]
This proves that
\[
\sigma\left(  A^{\flat}\left(  \beta\right)  |_{H_{h\ell}^{\flat}\left(
\beta\right)  }\right)  =-\sigma\left(  A\left(  \beta\right)  |_{H_{\ell
\ell,0}\left(  \beta\right)  }\right)  ^{-1}%
\]
which implies that%
\[
\sigma\left(  A\left(  \beta\right)  |_{H_{\ell\ell,0}\left(  \beta\right)
}\right)  =-\sigma\left(  A^{\flat}\left(  \beta\right)  |_{H_{h\ell}^{\flat
}\left(  \beta\right)  }\right)  ^{-1}.
\]
Now it immediately follows from this and duality that%
\[
\sigma\left(  A^{\flat}\left(  \beta\right)  |_{H_{\ell\ell,0}^{\flat}\left(
\beta\right)  }\right)  =-\sigma\left(  A\left(  \beta\right)  |_{H_{h\ell
}\left(  \beta\right)  }\right)  ^{-1}%
\]
which implies%
\[
\sigma\left(  A\left(  \beta\right)  |_{H_{h\ell}\left(  \beta\right)
}\right)  =-\sigma\left(  A^{\flat}\left(  \beta\right)  |_{H_{\ell\ell
,0}^{\flat}\left(  \beta\right)  }\right)  ^{-1}.
\]
Hence we have that%
\begin{align*}
\sigma\left(  A\left(  \beta\right)  \right)   &  =\sigma\left(  A\left(
\beta\right)  |_{H_{h\ell}\left(  \beta\right)  }\right)  \cup\sigma\left(
A\left(  \beta\right)  |_{H_{\ell\ell,0}\left(  \beta\right)  }\right)
\cup\sigma\left(  A\left(  \beta\right)  |_{H_{\ell\ell,1}\left(
\beta\right)  }\right) \\
\sigma\left(  A^{\flat}\left(  \beta\right)  \right)   &  =\sigma\left(
A^{\flat}\left(  \beta\right)  |_{H_{h\ell}^{\flat}\left(  \beta\right)
}\right)  \cup\sigma\left(  A^{\flat}\left(  \beta\right)  |_{H_{\ell\ell
,0}^{\flat}\left(  \beta\right)  }\right)  \cup\sigma\left(  A^{\flat}\left(
\beta\right)  |_{H_{\ell\ell,1}^{\flat}\left(  \beta\right)  }\right)
\end{align*}
which are the union of disjoint sets and%
\begin{align*}
\sigma\left(  A\left(  \beta\right)  |_{H_{\ell\ell,0}\left(  \beta\right)
}\right)   &  =-\sigma\left(  A^{\flat}\left(  \beta\right)  |_{H_{h\ell
}^{\flat}\left(  \beta\right)  }\right)  ^{-1},\\
\sigma\left(  A\left(  \beta\right)  |_{H_{h\ell}\left(  \beta\right)
}\right)   &  =-\sigma\left(  A^{\flat}\left(  \beta\right)  |_{H_{\ell\ell
,0}^{\flat}\left(  \beta\right)  }\right)  ^{-1},\\
\sigma\left(  A\left(  \beta\right)  \right)   &  =-\sigma\left(  A^{\flat
}\left(  \beta\right)  \right)  ^{-1}.
\end{align*}
These facts imply that%
\[
\sigma\left(  A\left(  \beta\right)  |_{H_{\ell\ell,1}\left(  \beta\right)
}\right)  =-\sigma\left(  A^{\flat}\left(  \beta\right)  |_{H_{\ell\ell
,1}^{\flat}\left(  \beta\right)  }\right)  ^{-1}.
\]

Now by Theorem \ref{tmdic}\ and duality we have that%
\begin{align*}
\dim H_{h\ell}\left(  \beta\right)   &  =N_{R},\quad\dim H_{\ell\ell}\left(
\beta\right)  =2N-N_{R},\\
\dim H_{h\ell}^{\flat}\left(  \beta\right)   &  =N_{R},\quad\dim H_{\ell\ell
}^{\flat}\left(  \beta\right)  =2N-N_{R}.
\end{align*}
And $H_{h\ell}^{\flat}\left(  \beta\right)  $ is the union of the algebraic
eigenspaces of $A^{\flat}\left(  \beta\right)  $ corresponding to the
eigenvalues in the set $\sigma\left(  A^{\flat}\left(  \beta\right)  \right)
\cap\left\{  \zeta:\left\vert \zeta\right\vert >\omega_{\max}^{\flat}\right\}
$ and $H_{\ell\ell,0}\left(  \beta\right)  $ is the union of the algebraic
eigenspaces of $A\left(  \beta\right)  $ corresponding to the eigenvalues in
the set $\sigma\left(  A\left(  \beta\right)  \right)  \cap\left\{
\zeta:\left\vert \zeta\right\vert <\omega_{\min}\right\}  $. Hence since
\[
\sigma\left(  A\left(  \beta\right)  |_{H_{\ell\ell,0}\left(  \beta\right)
}\right)  =-\sigma\left(  A^{\flat}\left(  \beta\right)  |_{H_{h\ell}^{\flat
}\left(  \beta\right)  }\right)  ^{-1},
\]
then it follows from Proposition \ref{ppfspdl}.3 that%
\[
\dim H_{\ell\ell,0}\left(  \beta\right)  =\dim H_{h\ell}^{\flat}\left(
\beta\right)  =N_{R},
\]
where by definition $N_{R}=\dim\operatorname{Ran}R$. Hence implying that%
\[
\dim H_{\ell\ell,1}\left(  \beta\right)  =\dim H_{\ell\ell}\left(
\beta\right)  -\dim H_{\ell\ell,0}\left(  \beta\right)  =2\left(
N-N_{R}\right)  .
\]
This completes the proof.
\end{proof}

\begin{corollary}
[low-loss/low-Q subspace: dissipative properties]\label{cmddI}If \ref{cnddl}
is true and $\beta>\max\left\{  2\frac{\omega_{\max}}{b_{\min}},2\frac
{\omega_{\max}^{\flat}}{b_{\min}^{\flat}}\right\}  $ then, in addition to
Corollary \ref{cmdic}\ being true, the spectrum of $A\left(  \beta\right)  $
rectricted to $H_{\ell\ell}\left(  \beta\right)  $ can be further partitioned
in terms of the damping factors by
\begin{equation}
\sigma\left(  A\left(  \beta\right)  |_{H_{\ell\ell,0}\left(  \beta\right)
}\right)  =\left\{  \zeta\in\sigma\left(  A\left(  \beta\right)  \right)
:\frac{-\operatorname{Im}\zeta}{\left\vert \zeta\right\vert ^{2}}\geq\beta
b_{\min}^{\flat}-\omega_{\max}^{\flat}>\omega_{\max}^{\flat}\right\}  .
\end{equation}
Furthermore, the quality factor $Q_{\zeta}=\frac{1}{2}\frac{\left\vert
\operatorname{Re}\zeta\right\vert }{-\operatorname{Im}\zeta}$ for $\zeta\in$
$\sigma\left(  A\left(  \beta\right)  |_{H_{\ell\ell,0}\left(  \beta\right)
}\right)  $ satisfies%
\begin{equation}
\max_{\zeta\in\sigma\left(  A\left(  \beta\right)  |_{H_{\ell\ell,0}\left(
\beta\right)  }\right)  }Q_{\zeta}\leq\frac{1}{2}\frac{\omega_{\max}^{\flat}%
}{\beta b_{\min}^{\flat}-\omega_{\max}^{\flat}}<\frac{1}{2}.
\end{equation}
Moreover,
\begin{equation}
\max_{\zeta\in\sigma\left(  A\left(  \beta\right)  |_{H_{\ell\ell,0}\left(
\beta\right)  }\right)  }\left\vert \zeta\right\vert \leq\frac{1}{\beta
b_{\min}^{\flat}-\omega_{\max}^{\flat}}<\omega_{\min}.
\end{equation}

\end{corollary}

\begin{proof}
Suppose \ref{cnddl} is true and $\beta>\max\left\{  2\frac{\omega_{\max}%
}{b_{\min}},2\frac{\omega_{\max}^{\flat}}{b_{\min}^{\flat}}\right\}  $. Then
by Corollary \ref{cmdic} we know that
\begin{align*}
\sigma\left(  A\left(  \beta\right)  |_{H_{\ell\ell}\left(  \beta\right)
}\right)   &  =\left\{  \zeta\in\sigma\left(  A\left(  \beta\right)  \right)
:0\leq-\operatorname{Im}\zeta\leq\omega_{\max}\right\}  ,\\
\text{\ }\sigma\left(  A\left(  \beta\right)  |_{H_{h\ell}\left(
\beta\right)  }\right)   &  =\left\{  \zeta\in\sigma\left(  A\left(
\beta\right)  \right)  :-\operatorname{Im}\zeta\geq\beta b_{\min}-\omega
_{\max}>\omega_{\max}\right\}  .
\end{align*}
Thus by Theorem \ref{tmddI} and duality we have%
\begin{gather*}
-\sigma\left(  A\left(  \beta\right)  |_{H_{\ell\ell,0}\left(  \beta\right)
}\right)  ^{-1}=\sigma\left(  A^{\flat}\left(  \beta\right)  |_{H_{h\ell
}^{\flat}\left(  \beta\right)  }\right) \\
=\left\{  \zeta\in\sigma\left(  A^{\flat}\left(  \beta\right)  \right)
:-\operatorname{Im}\zeta\geq\beta b_{\min}^{\flat}-\omega_{\max}^{\flat
}>\omega_{\text{\thinspace}\max}^{\flat}\right\}  .
\end{gather*}
Then since $\sigma\left(  A^{\flat}\left(  \beta\right)  \right)
=-\sigma\left(  A\left(  \beta\right)  \right)  ^{-1}$ this implies%
\begin{align*}
\sigma\left(  A\left(  \beta\right)  |_{H_{\ell\ell,0}\left(  \beta\right)
}\right)   &  =\left\{  \zeta\in\sigma\left(  A\left(  \beta\right)  \right)
:-\operatorname{Im}\left(  -\frac{1}{\zeta}\right)  \geq\beta b_{\min}^{\flat
}-\omega_{\max}^{\flat}>\omega_{\max}^{\flat}\right\} \\
&  =\left\{  \zeta\in\sigma\left(  A\left(  \beta\right)  \right)
:\frac{-\operatorname{Im}\zeta}{\left\vert \zeta\right\vert }\geq\beta
b_{\min}^{\flat}-\omega_{\max}^{\flat}>\omega_{\max}^{\flat}\right\}  .
\end{align*}
Next, by Corollary \ref{cmdic} and duality we have%
\[
\max_{\zeta\in\sigma\left(  A^{\flat}\left(  \beta\right)  |_{H_{h\ell}%
^{\flat}\left(  \beta\right)  }\right)  }Q_{\zeta}\leq\frac{1}{2}\frac
{\omega_{\max}^{\flat}}{\beta b_{\min}^{\flat}-\omega_{\max}^{\flat}}<\frac
{1}{2}.
\]
\ But since $Q_{-\zeta^{-1}}=Q_{\zeta}$ and $\sigma\left(  A\left(
\beta\right)  |_{H_{\ell\ell,0}\left(  \beta\right)  }\right)  =-\sigma\left(
A^{\flat}\left(  \beta\right)  |_{H_{h\ell}^{\flat}\left(  \beta\right)
}\right)  ^{-1}$ then
\begin{gather*}
\max_{\zeta\in\sigma\left(  A^{\flat}\left(  \beta\right)  |_{H_{h\ell}%
^{\flat}\left(  \beta\right)  }\right)  }Q_{\zeta}=\max_{-\zeta^{-1}\in
\sigma\left(  A^{\flat}\left(  \beta\right)  |_{H_{h\ell}^{\flat}\left(
\beta\right)  }\right)  }Q_{-\zeta^{-1}}\\
=\max_{\zeta\in\sigma\left(  A\left(  \beta\right)  |_{H_{\ell\ell,0}\left(
\beta\right)  }\right)  }Q_{-\zeta^{-1}}=\max_{\zeta\in\sigma\left(  A\left(
\beta\right)  |_{H_{\ell\ell,0}\left(  \beta\right)  }\right)  }Q_{\zeta}%
\end{gather*}
implying%
\[
\max_{\zeta\in\sigma\left(  A\left(  \beta\right)  |_{H_{\ell\ell,0}\left(
\beta\right)  }\right)  }Q_{\zeta}\leq\frac{1}{2}\frac{\omega_{\max}^{\flat}%
}{\beta b_{\min}^{\flat}-\omega_{\max}^{\flat}}<\frac{1}{2}.
\]
Finally, it follows from the fact
\[
\sigma\left(  A\left(  \beta\right)  |_{H_{\ell\ell,0}\left(  \beta\right)
}\right)  =\left\{  \zeta\in\sigma\left(  A\left(  \beta\right)  \right)
:\frac{-\operatorname{Im}\zeta}{\left\vert \zeta\right\vert ^{2}}\geq\beta
b_{\min}^{\flat}-\omega_{\max}^{\flat}>\omega_{\max}^{\flat}\right\}
\]
that%
\[
\zeta\in\sigma\left(  A\left(  \beta\right)  |_{H_{\ell\ell,0}\left(
\beta\right)  }\right)  \Rightarrow\beta b_{\min}^{\flat}-\omega_{\max}%
^{\flat}\leq\frac{-\operatorname{Im}\zeta}{\left\vert \zeta\right\vert ^{2}%
}\leq\frac{\left\vert \zeta\right\vert }{\left\vert \zeta\right\vert ^{2}%
}=\frac{1}{\left\vert \zeta\right\vert }%
\]
implying%
\[
\left\vert \zeta\right\vert \leq\frac{1}{\beta b_{\min}^{\flat}-\omega_{\max
}^{\flat}}<\left(  \omega_{\max}^{\flat}\right)  ^{-1}=\omega_{\min}.
\]
where the latter equality follows from Proposition (\ref{pspredl}). This
completes the proof.
\end{proof}

\begin{remark}
The results above show that we may consider the $N_{R}$-dimensional invariant
subspaces $H_{h\ell}\left(  \beta\right)  $ and $H_{\ell\ell,1}\left(
\beta\right)  $ of $A\left(  \beta\right)  $ to be the high-loss and
low-loss/low-$Q$ susceptible subspaces, respectively. Our results below will
show that we may consider the $2\left(  N-N_{R}\right)  $-dimensional
invariant subspace $H_{\ell\ell,1}\left(  \beta\right)  $ to be the
low-loss/high-$Q$ susceptible subspace.
\end{remark}

The following lemma from \cite[Appendix B, Lemma 28]{FigWel2}, \cite[Sec. V.4,
Prob. 4.8]{Kato} will be used to prove the next results.

\begin{lemma}
\label{lllspeb}If $M=M_{0}+E$ is a square matrix and $M_{0}$ is normal (i.e.,
$M_{0}M_{0}^{\ast}=M_{0}^{\ast}M_{0}$) then
\begin{equation}
\sigma\left(  M\right)  \subseteq\left\{  \lambda\in%
\mathbb{C}
:\operatorname{dist}\left(  \lambda,\sigma\left(  M_{0}\right)  \right)
\leq\left\Vert E\right\Vert \right\}  ,
\end{equation}
where $\operatorname{dist}\left(  \lambda,X\right)  :=\inf_{x\in X}\left\vert
\lambda-x\right\vert $ for any nonempty set $X\subseteq%
\mathbb{C}
$.
\end{lemma}

Let us now introduce some notation. Denote the disc in the complex plane of
radius $r>0$ centered at $\rho\in%
\mathbb{C}
$ by%
\begin{equation}
D\left(  \rho;r\right)  :=\left\{  \lambda\in%
\mathbb{C}
:\left\vert \lambda-\rho\right\vert \leq r\right\}  .
\end{equation}
Denote the largest eigenvalue of $B$ by $b_{\max}$, in particular, by
Proposition \ref{pspredl}%
\begin{equation}
b_{\max}=\max\sigma\left(  B\right)  =\max\left[  \sigma\left(  \alpha
^{-1}R\right)  \right]  .
\end{equation}
Let
\begin{equation}
y=c\left(  \beta\right)  ,\text{ }\beta>2\frac{\omega_{\max}}{b_{\min}}%
\end{equation}
be the function defined in (\ref{sintmr1}). If $\Omega\not =0$ (note that
$\Omega=0$ if and only if $\theta=0$ and $\eta=0$) then it is a strictly
decreasing function whose inverse
\begin{equation}
\beta=c^{-1}\left(  y\right)  ,\text{ }y>0,
\end{equation}
which is given by (\ref{sintmr1a}), is also strictly decreasing.

\begin{theorem}
[low-loss subspace: eigenvalue bounds]\label{tllspeb}If $\beta>2\frac
{\omega_{\max}}{b_{\min}}$ then Theorem \ref{tmdic} and Corollary \ref{cmdic}
are true and we have the additional bounds%
\[
\max_{\zeta\in\sigma\left(  A\left(  \beta\right)  |_{H_{\ell\ell}\left(
\beta\right)  }\right)  }\min_{\rho\in\sigma\left(  P_{B}^{\bot}\Omega
P_{B}^{\bot}\right)  }\left\vert \zeta-\rho\right\vert \leq c\left(
\beta\right)  .
\]
In particular, for the frequencies and damping factors, we have the bounds%
\[
0\leq\min_{\rho\in\sigma\left(  P_{B}^{\bot}\Omega P_{B}^{\bot}\right)
}\left\vert \operatorname{Re}\zeta-\rho\right\vert ,\text{\quad}%
-\operatorname{Im}\zeta\leq c\left(  \beta\right)
\]
for all $\beta>2\frac{\omega_{\max}}{b_{\min}}$ whenever $\zeta\in
\sigma\left(  A\left(  \beta\right)  |_{H_{\ell\ell}\left(  \beta\right)
}\right)  $.
\end{theorem}

We will prove this theorem and the next two results below. To do so we
introduce some new notation.

Denote the smallest positive eigenvalue of the Hermitian matrix $P_{B}^{\bot
}\Omega P_{B}^{\bot}$ by $\rho_{\min}$ whenever $P_{B}^{\bot}\Omega
P_{B}^{\bot}\not =0$ (by Proposition \ref{pspredl2} this is equivalent to
$N_{R}<N$), in particular, it follows from Proposition \ref{pspredl2} that
\begin{equation}
\rho_{\min}=\min_{0\not =\rho\in\sigma\left(  P_{B}^{\bot}\Omega P_{B}^{\bot
}\right)  }\left\vert \rho\right\vert =\min\left[  \sigma\left(  P_{B}^{\bot
}\Omega P_{B}^{\bot}\right)  \cap\left(  0,\infty\right)  \right]  .
\end{equation}

\begin{theorem}
[modal dichotomy II]\label{tllspmd}Assume that $N_{R}<N$. If $\beta
>c^{-1}\left(  \rho_{\min}/2\right)  $ then $\beta>2\frac{\omega_{\max}%
}{b_{\min}}$, $c\left(  \beta\right)  <\rho_{\min}/2\leq\omega_{\max}/2$ and,
in addition to Theorems \ref{tmdic}, \ref{tllspeb} being true, there exists
unique invariant subspaces of the system operator $A\left(  \beta\right)
=\Omega-\mathrm{i}\beta B$ such that
\begin{align*}
&  H_{\ell\ell}\left(  \beta\right)  =H_{\ell\ell,0}\left(  \beta\right)
\oplus H_{\ell\ell,1}\left(  \beta\right)  ;\\
&  \sigma\left(  A\left(  \beta\right)  |_{H_{\ell\ell,0}\left(  \beta\right)
}\right)  =\sigma\left(  A\left(  \beta\right)  \right)  \cap D\left(
0;c\left(  \beta\right)  \right)  ;\\
&  \sigma\left(  A\left(  \beta\right)  |_{H_{\ell\ell,1}\left(  \beta\right)
}\right)  =\sigma\left(  A\left(  \beta\right)  \right)  \cap%
{\textstyle\bigcup_{0\not =\rho\in\sigma\left(  P_{B}^{\bot}\Omega P_{B}%
^{\bot}\right)  }}
D\left(  \rho;c\left(  \beta\right)  \right)  .
\end{align*}
Furthermore,%
\begin{align*}
\sigma\left(  A\left(  \beta\right)  |_{H_{\ell\ell,1}\left(  \beta\right)
}\right)   &  =\sigma\left(  A\left(  \beta\right)  \right)  \cap\left\{
\zeta\in%
\mathbb{C}
:\rho_{\min}/2<\left\vert \zeta\right\vert \leq\omega_{\max}\right\}  ,\\
\sigma\left(  A\left(  \beta\right)  |_{H_{\ell\ell,0}\left(  \beta\right)
}\right)   &  =\sigma\left(  A\left(  \beta\right)  \right)  \cap\left\{
\zeta\in%
\mathbb{C}
:\left\vert \zeta\right\vert <\rho_{\min}/2\right\}  .
\end{align*}
Moreover, the dimensions of these invariant subspaces are%
\[
\dim H_{\ell\ell,1}\left(  \beta\right)  =\dim\operatorname{Ran}\left(
P_{B}^{\bot}\Omega P_{B}^{\bot}\right)  ,\text{\quad}\dim H_{\ell\ell
,0}\left(  \beta\right)  =2N-N_{R}-\dim H_{\ell\ell,1}\left(  \beta\right)  .
\]
In particular, if $\operatorname{Ker}R\cap\operatorname{Ker}\eta=\left\{
0\right\}  $ then
\[
\dim H_{\ell\ell,1}\left(  \beta\right)  =2\left(  N-N_{R}\right)  .
\]

\end{theorem}

\begin{corollary}
[low-loss/high-Q subspace: dissipative properties]\label{cllspmd}If $N_{R}<N$
and $\beta>c^{-1}\left(  \rho_{\min}/2\right)  $ then the spectrum of
$A\left(  \beta\right)  $ can be partitioned in terms of the frequencies and
damping factors by%
\begin{align*}
\sigma\left(  A\left(  \beta\right)  |_{H_{\ell\ell,0}\left(  \beta\right)
}\right)   &  =\left\{  \zeta\in\sigma\left(  A\left(  \beta\right)  \right)
:0\leq-\operatorname{Im}\zeta\leq c\left(  \beta\right)  \text{ and
}\left\vert \operatorname{Re}\zeta\right\vert \leq c\left(  \beta\right)
\right\}  \\
\sigma\left(  A\left(  \beta\right)  |_{H_{\ell\ell,1}\left(  \beta\right)
}\right)   &  =\left\{  \zeta\in\sigma\left(  A\left(  \beta\right)  \right)
:0\leq-\operatorname{Im}\zeta\leq c\left(  \beta\right)  \text{ and
}\left\vert \operatorname{Re}\zeta\right\vert \geq\rho_{\min}-c\left(
\beta\right)  \right\}
\end{align*}
where $\rho_{\min}-c\left(  \beta\right)  >\max\left\{  c\left(  \beta\right)
,\rho_{\min}/2\right\}  $, $\rho_{\min}/2\leq\omega_{\max}/2$, and $c\left(
\beta\right)  \searrow0$ as $\beta\rightarrow\infty$. Moreover, minimum of the
quality factors $Q_{\zeta}=\frac{1}{2}\frac{\left\vert \operatorname{Re}%
\zeta\right\vert }{-\operatorname{Im}\zeta}$ for $\zeta\in$ $\sigma\left(
A\left(  \beta\right)  |_{H_{\ell\ell,1}\left(  \beta\right)  }\right)  $
satisfy%
\begin{equation}
\min_{\zeta\in\sigma\left(  A\left(  \beta\right)  |_{H_{\ell\ell,1}\left(
\beta\right)  }\right)  }Q_{\zeta}\geq\frac{1}{2}\frac{\rho_{\min}-c\left(
\beta\right)  }{c\left(  \beta\right)  }>\frac{1}{4}\frac{\rho_{\min}%
}{c\left(  \beta\right)  }\geq\frac{1}{2}.
\end{equation}
In particular, $\operatorname{Re}\zeta\not =0$ for every $\zeta\in$
$\sigma\left(  A\left(  \beta\right)  |_{H_{\ell\ell,1}\left(  \beta\right)
}\right)  $.
\end{corollary}

\begin{proof}
We begin by proving Theorem \ref{tllspeb}. If $\Omega=0$ then the statement is
true trivially. \ Thus, suppose $\Omega\not =0$. Consider the perturbation and
it's resolvent%
\begin{gather}
C\left(  \varepsilon\right)  =\varepsilon\Omega+B,\text{ }\varepsilon\in%
\mathbb{C}
;\text{ }R\left(  \lambda,\varepsilon\right)  =\left(  \lambda\mathbf{1}%
-C\left(  \varepsilon\right)  \right)  ^{-1},\label{PertC}\\
\lambda\not \in \sigma\left(  C\left(  \varepsilon\right)  \right)  ;\text{
}R_{0}\left(  \lambda\right)  =R\left(  \lambda,0\right)  .\nonumber
\end{gather}
In particular, since $B\geq0$, $\Omega\not =0$ are Hermitian matrices then
from the spectral theory of self-adjoint operators it follows that%
\begin{gather*}
\left\Vert R_{0}\left(  \lambda\right)  \right\Vert =\operatorname{dist}%
\left(  \lambda,\sigma\left(  B\right)  \right)  ^{-1},\text{ }\omega_{\max
}=\left\Vert \Omega\right\Vert >0,\\
b_{\min}=\operatorname{dist}\left(  0,\sigma\left(  B\right)  \setminus
\left\{  0\right\}  \right)  >0\text{.}%
\end{gather*}
We will denote the circle centered at $0$ with radius $\frac{b_{\min}}{2}$ by
$\Gamma$, i.e.,
\begin{equation}
\Gamma=\left\{  \lambda\in%
\mathbb{C}
:\left\vert \lambda\right\vert =\frac{b_{\min}}{2}\right\}  .
\label{Stg1SpltGamma}%
\end{equation}
Then it follows from the results of \cite[Theorems 1 \& 2, Sec. 8.1.]{Bau85},
\cite[Lemma 4, Sec. 3.3.3.]{Bau85}, \cite[Formula (3.5), Sec. 3.3.1.]%
{Bau85}\ that the group projection of the $\left(  0,0\right)  $-group of
perturbed eigenvalues of $C\left(  \varepsilon\right)  $, which we denote by
$P\left(  \varepsilon\right)  $, is analytic in $\varepsilon$ for $\left\vert
\varepsilon\right\vert <\frac{b_{\min}}{2\omega_{\max}}$ and can be
represented by the contour integral over the circle $\Gamma$ (positively
oriented) with
\begin{gather}
P\left(  \varepsilon\right)  =\frac{1}{2\pi i}%
{\textstyle\int_{\Gamma}}
R\left(  \lambda,\varepsilon\right)  d\lambda=%
{\textstyle\sum_{n=0}^{\infty}}
P_{n}\varepsilon^{n},\text{\quad}\left\vert \varepsilon\right\vert
<\frac{b_{\min}}{2\omega_{\max}},\label{PertP}\\
P_{0}=P_{B}^{\bot},\text{ }P_{n}=\frac{1}{2\pi i}%
{\textstyle\int_{\Gamma}}
R_{0}\left(  \lambda\right)  \left[  \Omega R_{0}\left(  \lambda\right)
\right]  ^{n}d\lambda,\text{\quad}n=0,1,2,\ldots.,\nonumber\\
\left\Vert P_{n}\right\Vert \leq\frac{1}{2\pi}%
{\textstyle\int_{\Gamma}}
\left\Vert \Omega\right\Vert ^{n}\left\Vert R_{0}\left(  \lambda\right)
\right\Vert ^{n+1}d\left\vert \lambda\right\vert =\omega_{\max}^{n}\left(
\frac{b_{\min}}{2}\right)  ^{-\left(  n+1\right)  }\frac{1}{2\pi}%
{\textstyle\int_{\Gamma}}
d\left\vert \lambda\right\vert =\left(  \frac{2\omega_{\max}}{b_{\min}%
}\right)  ^{n}.\nonumber
\end{gather}

Now we can define the analytic matrix-valued function%
\begin{equation}
M\left(  \varepsilon\right)  =\frac{1}{\varepsilon}P\left(  \varepsilon
\right)  C\left(  \varepsilon\right)  P\left(  \varepsilon\right)
,\text{\quad}\left\vert \varepsilon\right\vert <\frac{b_{\min}}{2\omega_{\max
}} \label{PertM}%
\end{equation}
which by by \cite[Lemma 4, Sec. 3.3.3.]{Bau85} has the properties%
\begin{gather}
M\left(  \varepsilon\right)  =\frac{1}{\varepsilon}C\left(  \varepsilon
\right)  P\left(  \varepsilon\right)  =M\left(  0\right)  +%
{\textstyle\sum_{n=1}^{\infty}}
\left(  \Omega P_{n}+BP_{n+1}\right)  \varepsilon^{n},\text{\quad}M\left(
0\right)  =P_{B}^{\bot}\Omega P_{B}^{\bot},\label{PertM_1}\\
\Omega P_{n}+BP_{n+1}=\frac{1}{2\pi i}%
{\textstyle\int_{\Gamma}}
\lambda R_{0}\left(  \lambda\right)  \left[  \Omega R_{0}\left(
\lambda\right)  \right]  ^{n+1}d\lambda,\text{\quad}n=1,2,\ldots.,\nonumber\\
\left\Vert \Omega P_{n}+BP_{n+1}\right\Vert \leq\frac{1}{2\pi}%
{\textstyle\int_{\Gamma}}
\left\vert \lambda\right\vert \left\Vert \Omega\right\Vert ^{n+1}\left\Vert
R_{0}\left(  \lambda\right)  \right\Vert ^{n+2}d\left\vert \lambda\right\vert
=\frac{b_{\min}}{2}\left(  \frac{2\omega_{\max}}{b_{\min}}\right)
^{n+1}.\nonumber
\end{gather}
As $M\left(  0\right)  =P_{B}^{\bot}\Omega P_{B}^{\bot}$ is a Hermitian matrix
(since $\Omega$ and $P_{B}^{\bot}$ are) and $0\leq\frac{2\omega_{\max}%
}{b_{\min}}\left\vert \varepsilon\right\vert <1$ then it follows immediately
from Lemma \ref{lllspeb} that%
\begin{gather*}
\max_{\lambda\in\sigma\left(  M\left(  \varepsilon\right)  \right)  }%
\min_{\rho\in\sigma\left(  M\left(  0\right)  \right)  }\left\vert
\lambda-\rho\right\vert \leq\left\Vert M\left(  \varepsilon\right)  -M\left(
0\right)  \right\Vert \leq%
{\textstyle\sum_{n=1}^{\infty}}
\left\Vert \Omega P_{n}+BP_{n+1}\right\Vert \left\vert \varepsilon\right\vert
^{n}\\
\leq%
{\textstyle\sum_{n=1}^{\infty}}
\frac{b_{\min}}{2}\left(  \frac{2\omega_{\max}}{b_{\min}}\right)
^{n+1}\left\vert \varepsilon\right\vert ^{n}=\frac{b_{\min}}{2}\left(
\frac{2\omega_{\max}}{b_{\min}}\right)
{\textstyle\sum_{n=1}^{\infty}}
\left(  \frac{2\omega_{\max}}{b_{\min}}\left\vert \varepsilon\right\vert
\right)  ^{n}\\
=\left(  \frac{2\omega_{\max}^{2}}{b_{\min}}\right)  \left(  \frac{\left\vert
\varepsilon\right\vert }{1-\frac{2\omega_{\max}}{b_{\min}}\left\vert
\varepsilon\right\vert }\right)  .
\end{gather*}
Thus we have proven%
\begin{align}
\left\vert \varepsilon\right\vert  &  <\frac{b_{\min}}{2\omega_{\max}%
}\Rightarrow\max_{\lambda\in\sigma\left(  \varepsilon^{-1}C\left(
\varepsilon\right)  P\left(  \varepsilon\right)  \right)  }\min_{\rho\in
\sigma\left(  P_{B}^{\bot}\Omega P_{B}^{\bot}\right)  }\left\vert \lambda
-\rho\right\vert \leq b\left(  \left\vert \varepsilon\right\vert \right)
,\text{ where}\nonumber\\
b\left(  u\right)   &  =\left(  \frac{2\omega_{\max}^{2}}{b_{\min}}\right)
\left(  \frac{u}{1-\frac{2\omega_{\max}}{b_{\min}}u}\right)  \text{,\quad for
}0\leq u<\frac{b_{\min}}{2\omega_{\max}} \label{bfunc}%
\end{align}

We will now prove that
\[
\sigma\left(  C\left(  \varepsilon\right)  |_{\operatorname{Ran}P\left(
\varepsilon\right)  }\right)  \subseteq D\left(  0;\left\vert \varepsilon
\right\vert \omega_{\text{max}}\right)  ,\text{\quad}\sigma\left(  C\left(
\varepsilon\right)  |_{\operatorname{Ran}\left[  \mathbf{1}-P\left(
\varepsilon\right)  \right]  }\right)  \subseteq%
{\textstyle\bigcup_{j=1}^{N_{R}}}
D\left(  b_{j};\left\vert \varepsilon\right\vert \omega_{\max}\right)  .
\]
First, it follows from Lemma \ref{lllspeb} and the fact that $\omega
_{\text{max}}=\left\Vert \Omega\right\Vert $ that%
\begin{gather}
\sigma\left(  C\left(  \varepsilon\right)  \right)  \subseteq\left\{
\lambda\in%
\mathbb{C}
:\operatorname{dist}\left(  \lambda,\sigma\left(  B\right)  \right)
\leq\left\vert \varepsilon\right\vert \omega_{\max}\right\} \nonumber\\
=D\left(  0;\left\vert \varepsilon\right\vert \omega_{\max}\right)  \cup%
{\textstyle\bigcup_{j=1}^{N_{R}}}
D\left(  b_{j};\left\vert \varepsilon\right\vert \omega_{\max}\right)
,\nonumber\\
\left\vert \varepsilon\right\vert <\frac{b_{\min}}{2\omega_{\max}}\Rightarrow
D\left(  0;\left\vert \varepsilon\right\vert \omega_{\max}\right)  \cap%
{\textstyle\bigcup_{j=1}^{N_{R}}}
D\left(  b_{j};\left\vert \varepsilon\right\vert \omega_{\max}\right)
=\emptyset\text{ and}\label{Stg1Splt_1}\\
D\left(  b_{j};\left\vert \varepsilon\right\vert \omega_{\max}\right)
\subseteq\left\{  \lambda\in%
\mathbb{C}
:\left\vert \lambda-b_{j}\right\vert <\frac{b_{\min}}{2}\right\}  \text{ for
}0\leq j\leq N_{R}.\nonumber
\end{gather}
Now by analytic continuation of the eigenvalues of the perturbation $C\left(
\varepsilon\right)  $ of $C\left(  0\right)  =B$ from $\varepsilon=0$ in the
neighborhood $\left\vert \varepsilon\right\vert <\frac{b_{\min}}{2\omega
_{\max}}$, it follows that $\sigma\left(  C\left(  \varepsilon\right)
|_{\operatorname{Ran}P\left(  \varepsilon\right)  }\right)  \subseteq
\operatorname{int}\Gamma=\left\{  \lambda\in%
\mathbb{C}
:\left\vert \lambda\right\vert <\frac{b_{\min}}{2}\right\}  $ implying that if
$\left\vert \varepsilon\right\vert <\frac{b_{\min}}{2\omega_{\max}}$ then
\begin{equation}
\sigma\left(  C\left(  \varepsilon\right)  |_{\operatorname{Ran}P\left(
\varepsilon\right)  }\right)  \subseteq D\left(  0;\left\vert \varepsilon
\right\vert \omega_{\max}\right)  \text{, }\sigma\left(  C\left(
\varepsilon\right)  |_{\operatorname{Ran}\left[  \mathbf{1}-P\left(
\varepsilon\right)  \right]  }\right)  \subseteq%
{\textstyle\bigcup_{j=1}^{N_{R}}}
D\left(  b_{j};\left\vert \varepsilon\right\vert \omega_{\max}\right)  .
\label{Stg1Splt_2}%
\end{equation}

Now we make the substitute $\varepsilon=\left(  -\mathrm{i}\beta\right)
^{-1}$ with $\left\vert \varepsilon\right\vert <\frac{b_{\min}}{2\omega_{\max
}}$ so that%
\[
A\left(  \beta\right)  =\varepsilon^{-1}C\left(  \varepsilon\right)
,\text{\quad}\sigma\left(  \varepsilon^{-1}C\left(  \varepsilon\right)
|_{\operatorname{Ran}P\left(  \varepsilon\right)  }\right)  \subseteq D\left(
0;\omega_{\max}\right)
\]
which implies by the uniqueness portion of Theorem \ref{tmdic} that%
\begin{align*}
\sigma_{0}\left(  A\left(  \beta\right)  \right)   &  =\sigma\left(
\varepsilon^{-1}C\left(  \varepsilon\right)  |_{\operatorname{Ran}P\left(
\varepsilon\right)  }\right)  ,\text{ }H_{\ell\ell}\left(  \beta\right)
=\operatorname{Ran}P\left(  \varepsilon\right)  ,\text{ }H_{h\ell}\left(
\beta\right)  =\operatorname{Ran}\left[  \mathbf{1}-P\left(  \varepsilon
\right)  \right]  ,\\
\sigma_{1}\left(  A\left(  \beta\right)  \right)   &  =\sigma\left(
\varepsilon^{-1}C\left(  \varepsilon\right)  |_{\operatorname{Ran}\left[
\mathbf{1}-P\left(  \varepsilon\right)  \right]  }\right)  \subseteq%
{\textstyle\bigcup_{j=1}^{N_{R}}}
D\left(  \varepsilon^{-1}b_{j};\omega_{\max}\right)
\end{align*}
and hence
\begin{equation}
\left\vert \beta\right\vert ^{-1}=\left\vert \varepsilon\right\vert
<\frac{b_{\min}}{2\omega_{\max}}\Rightarrow\max_{\lambda\in\sigma\left(
\varepsilon^{-1}C\left(  \varepsilon\right)  |_{\operatorname{Ran}P\left(
\varepsilon\right)  }\right)  }\min_{\rho\in\sigma\left(  P_{B}^{\bot}\Omega
P_{B}^{\bot}\right)  }\left\vert \lambda-\rho\right\vert \leq b\left(
\left\vert \varepsilon\right\vert \right)  =c\left(  \left\vert \beta
\right\vert \right)  . \label{Stg1Splt_3}%
\end{equation}
This completes the proof of Theorem \ref{tllspeb}.

We will now prove Theorem \ref{tllspmd}. Assume that $N_{R}<N$ (i.e.,
$P_{B}^{\bot}\Omega P_{B}^{\bot}\not =0$, by Proposition \ref{pspredl2}). We
will work with the analytic perturbation (\ref{PertM}) and at the end
interpret the results for the substitute $\varepsilon=\left(  -\mathrm{i}%
\beta\right)  ^{-1}$. We begin by assuming that $\left\vert \varepsilon
\right\vert <\frac{b_{\min}}{2\omega_{\max}}$. Then by definition of $M\left(
\varepsilon\right)  $ in (\ref{PertM}), the fact that $\dim\operatorname{Ran}%
P\left(  \varepsilon\right)  =2N-N_{R}>0$ for the projection $P\left(
\varepsilon\right)  $ in (\ref{PertP}), and by (\ref{Stg1Splt_2}),
(\ref{Stg1Splt_3}) we have must have%
\begin{gather}
M\left(  \varepsilon\right)  =\varepsilon^{-1}C\left(  \varepsilon\right)
\text{ on }\operatorname{Ran}P\left(  \varepsilon\right)  ,\label{Stg2Splt_1}%
\\
\sigma\left(  M\left(  \varepsilon\right)  \right)  =\sigma\left(
\varepsilon^{-1}C\left(  \varepsilon\right)  |_{\operatorname{Ran}P\left(
\varepsilon\right)  }\right)  \cup\left\{  0\right\}  \subseteq D\left(
0;\omega_{\max}\right)  ,\label{Stg2Splt_1a}\\
\sigma\left(  M\left(  \varepsilon\right)  \right)  \subseteq%
{\textstyle\bigcup_{\rho\in\sigma\left(  P_{B}^{\bot}\Omega P_{B}^{\bot
}\right)  }}
D\left(  \rho;b\left(  \left\vert \varepsilon\right\vert \right)  \right)
.\label{Stg2Splt_1b}%
\end{gather}
We now define $b^{-1}\left(  v\right)  $, $v>0$ to be the inverse function of
$b\left(  u\right)  $ in (\ref{bfunc}), i.e.,%
\begin{equation}
b^{-1}\left(  v\right)  =\frac{v}{\left(  \frac{2\omega_{\max}^{2}}{b_{\min}%
}\right)  +\left(  2\frac{\omega_{\max}}{b_{\min}}\right)  v},\text{\quad for
}v\geq0.\label{bfuncinv}%
\end{equation}
It follows that the functions $b$ and $b^{-1}$ are strictly increasing
functions and%
\begin{equation}
b^{-1}\left(  v\right)  <\lim_{v\rightarrow\infty}b^{-1}\left(  v\right)
=\frac{b_{\min}}{2\omega_{\max}}\text{ for any }v\geq0\text{.}%
\label{bfuninv_1}%
\end{equation}
From now on we will assume $\varepsilon\in%
\mathbb{C}
$ is such that $\left\vert \varepsilon\right\vert <b^{-1}\left(  \rho_{\min
}/2\right)  $. This implies that
\begin{equation}
\left\vert \varepsilon\right\vert <b^{-1}\left(  \rho_{\min}/2\right)
<\frac{b_{\min}}{2\omega_{\max}},\text{\qquad}b\left(  \left\vert
\varepsilon\right\vert \right)  <\rho_{\min}/2\leq\omega_{\max}%
/2.\label{bfunc_1}%
\end{equation}
Thus since $b\left(  \left\vert \varepsilon\right\vert \right)  <\rho_{\min
}/2$ then%
\begin{equation}
D\left(  0;b\left(  \left\vert \varepsilon\right\vert \right)  \right)  \cap%
{\textstyle\bigcup_{0\not =\rho\in\sigma\left(  P_{B}^{\bot}\Omega P_{B}%
^{\bot}\right)  }}
D\left(  \rho;b\left(  \left\vert \varepsilon\right\vert \right)  \right)
=\emptyset\label{Stg2Splt_2}%
\end{equation}
implying the spectrum of $M\left(  \varepsilon\right)  $ splits as
\begin{align}
\sigma\left(  M\left(  \varepsilon\right)  \right)   &  =\sigma_{0}\left(
M\left(  \varepsilon\right)  \right)  \cup\sigma_{1}\left(  M\left(
\varepsilon\right)  \right)  \subseteq D\left(  0;\omega_{\text{max}}\right)
,\text{ }\sigma_{0}\left(  M\left(  \varepsilon\right)  \right)  \cap
\sigma_{1}\left(  M\left(  \varepsilon\right)  \right)  =\emptyset
,\label{Stg2Splt_3}\\
\sigma_{0}\left(  M\left(  \varepsilon\right)  \right)   &  :=\sigma\left(
M\left(  \varepsilon\right)  \right)  \cap D\left(  0;b\left(  \left\vert
\varepsilon\right\vert \right)  \right)  \subseteq\left\{  \lambda\in%
\mathbb{C}
:\left\vert \lambda\right\vert <\rho_{\min}/2\right\}  ,\label{Stg2Splt_3a}\\
\sigma_{1}\left(  M\left(  \varepsilon\right)  \right)   &  :=\sigma\left(
M\left(  \varepsilon\right)  \right)  \cap%
{\textstyle\bigcup_{0\not =\rho\in\sigma\left(  P_{B}^{\bot}\Omega P_{B}%
^{\bot}\right)  }}
D\left(  \rho;b\left(  \left\vert \varepsilon\right\vert \right)  \right)
\subseteq%
\mathbb{C}
\setminus D\left(  0;\rho_{\min}/2\right)  ,\label{Stg2Splt_3b}%
\end{align}
which follows from (\ref{Stg2Splt_1a}), (\ref{Stg2Splt_1b}), (\ref{bfunc_1}),
and the fact that if $b\left(  \left\vert \varepsilon\right\vert \right)
\geq\left\vert \lambda-\rho\right\vert $ and $0\not =\rho\in\sigma\left(
P_{B}^{\bot}\Omega P_{B}^{\bot}\right)  $ then%
\[
\rho_{\min}/2>b\left(  \left\vert \varepsilon\right\vert \right)
\geq\left\vert \lambda-\rho\right\vert \geq\left\vert \rho\right\vert
-\left\vert \lambda\right\vert \geq\rho_{\min}-\left\vert \lambda\right\vert
.\text{ }%
\]
Denote the circle centered at $0$ with radius $\frac{\rho_{\min}}{2}$ by
$\Gamma_{0}$, i.e.,%
\begin{equation}
\Gamma_{0}=\left\{  \lambda\in%
\mathbb{C}
:\left\vert \lambda\right\vert =\rho_{\min}/2\right\}
.\label{Stg2SpltGamma_0}%
\end{equation}
Then it follows from the results of \cite[Theorems 1 \& 2, Sec. 8.1.]{Bau85},
\cite[Lemma 4, Sec. 3.3.3.]{Bau85}, \cite[Formula (3.5), Sec. 3.3.1.]%
{Bau85}\ that the group projection of the $\left(  0,0\right)  $-group of
perturbed eigenvalues of $M\left(  \varepsilon\right)  $, which we denote by
$P_{0}\left(  \varepsilon\right)  $, is analytic in $\varepsilon$ for
$\left\vert \varepsilon\right\vert <b^{-1}\left(  \rho_{\min}/2\right)  $ and
can be represented by the contour integral over the circle $\Gamma_{0}$
(positively oriented) with
\begin{equation}
P_{0}\left(  \varepsilon\right)  =\frac{1}{2\pi\mathrm{i}}%
{\textstyle\int_{\Gamma_{0}}}
\left(  \lambda\mathbf{1}-M\left(  \varepsilon\right)  \right)  \,\mathrm{d}%
\lambda,\text{\qquad}\left\vert \varepsilon\right\vert <b^{-1}\left(
\rho_{\min}/2\right)  ,\label{PertP_0}%
\end{equation}
where $P_{0}\left(  0\right)  $ is the projection onto the algebraic
eigenspace of $M\left(  0\right)  =P_{B}^{\bot}\Omega P_{B}^{\bot}$
corresponding to the zero eigenvalue. As $M\left(  0\right)  $ is a Hermitian
matrix this implies that $P_{0}\left(  0\right)  $ is the orthogonal
projection onto the $\operatorname{Ker}\left(  P_{B}^{\bot}\Omega P_{B}^{\bot
}\right)  $ and $\mathbf{1}-P_{0}\left(  0\right)  $ is the orthogonal
projection onto $\operatorname{Ran}\left(  P_{B}^{\bot}\Omega P_{B}^{\bot
}\right)  $. It follows from analytic continuation of the eigenvalues of
$M\left(  \varepsilon\right)  $ from $\varepsilon=0$ that $\mathbf{1}%
-P_{0}\left(  \varepsilon\right)  $ is the sum over all $\rho\in\sigma\left(
P_{B}^{\bot}\Omega P_{B}^{\bot}\right)  $, $\rho\not =0$ of the group
projection of the $\left(  0,\rho\right)  $-group of perturbed eigenvalues of
$M\left(  \varepsilon\right)  $ and hence%
\begin{align}%
\mathbb{C}
^{2N} &  =\operatorname{Ran}P_{0}\left(  \varepsilon\right)  \oplus
\operatorname{Ran}\left[  \mathbf{1}-P_{0}\left(  \varepsilon\right)  \right]
\label{Stg2Splt4}\\
\sigma\left(  M\left(  \varepsilon\right)  |_{\operatorname{Ran}P_{0}\left(
\varepsilon\right)  }\right)   &  =\sigma_{0}\left(  M\left(  \varepsilon
\right)  \right)  ,\text{ }\sigma\left(  M\left(  \varepsilon\right)
|_{\operatorname{Ran}\left[  \mathbf{1}-P_{0}\left(  \varepsilon\right)
\right]  }\right)  =\sigma_{1}\left(  M\left(  \varepsilon\right)  \right)
,\nonumber\\
\dim\operatorname{Ran}\left[  \mathbf{1}-P_{0}\left(  \varepsilon\right)
\right]   &  =\dim\operatorname{Ran}\left[  \mathbf{1}-P_{0}\left(  0\right)
\right]  =\dim\operatorname{Ran}\left(  P_{B}^{\bot}\Omega P_{B}^{\bot
}\right)  .\nonumber
\end{align}
Now it follows that since $C\left(  \varepsilon\right)  $ and $M\left(
\varepsilon\right)  $ commute then from their integral representations it
follows that $P\left(  \varepsilon\right)  $ and $P_{0}\left(  \varepsilon
\right)  $ are commuting projections which also commute with $C\left(
\varepsilon\right)  $. Thus in particular, $P\left(  \varepsilon\right)
\left[  \mathbf{1}-P_{0}\left(  \varepsilon\right)  \right]  $ and $P\left(
\varepsilon\right)  P_{0}\left(  \varepsilon\right)  $ are also analytic
projections that commute with $C\left(  \varepsilon\right)  $ such that
\begin{gather}
\operatorname{Ran}P\left(  \varepsilon\right)  =\operatorname{Ran}\left[
P\left(  \varepsilon\right)  P_{0}\left(  \varepsilon\right)  \right]
\oplus\operatorname{Ran}\left\{  P\left(  \varepsilon\right)  \left[
\mathbf{1}-P_{0}\left(  \varepsilon\right)  \right]  \right\}
,\label{Stg2Splt5}\\
P\left(  \varepsilon\right)  \left[  \mathbf{1}-P_{0}\left(  \varepsilon
\right)  \right]  =\left[  \mathbf{1}-P_{0}\left(  \varepsilon\right)
\right]  P\left(  \varepsilon\right)  =\left[  \mathbf{1}-P_{0}\left(
\varepsilon\right)  \right]  P\left(  \varepsilon\right)  \left[
\mathbf{1}-P_{0}\left(  \varepsilon\right)  \right]  .\label{Stg2Splt5a}%
\end{gather}
We will now prove that
\begin{equation}
P\left(  \varepsilon\right)  \left[  \mathbf{1}-P_{0}\left(  \varepsilon
\right)  \right]  =\mathbf{1}-P_{0}\left(  \varepsilon\right)
.\label{Stg2Splt5b}%
\end{equation}
To do we need only prove that $\dim\operatorname{Ran}\left\{  P\left(
\varepsilon\right)  \left[  \mathbf{1}-P_{0}\left(  \varepsilon\right)
\right]  \right\}  =\dim\operatorname{Ran}\left[  \mathbf{1}-P_{0}\left(
\varepsilon\right)  \right]  $, but by the fact that these are the ranges of
analytic projections then these dimensions are constant and hence it sufficies
to prove that $P\left(  0\right)  \left[  \mathbf{1}-P_{0}\left(  0\right)
\right]  =\mathbf{1}-P_{0}\left(  0\right)  $. But this follows immediately
from the facts%
\begin{equation}
P\left(  0\right)  =P_{B}^{\bot},\text{ }\operatorname{Ker}\left[
\mathbf{1}-P_{0}\left(  0\right)  \right]  =\operatorname{Ker}\left(
P_{B}^{\bot}\Omega P_{B}^{\bot}\right)  ,\text{ }\operatorname{Ran}\left[
\mathbf{1}-P_{0}\left(  0\right)  \right]  =\operatorname{Ran}\left(
P_{B}^{\bot}\Omega P_{B}^{\bot}\right)  .\label{Stg2Splt5c}%
\end{equation}
Thus we conclude from (\ref{Stg2Splt5})-(\ref{Stg2Splt5c}) that%
\begin{align}
\operatorname{Ran}P\left(  \varepsilon\right)   &  =\operatorname{Ran}\left[
P\left(  \varepsilon\right)  P_{0}\left(  \varepsilon\right)  \right]
\oplus\operatorname{Ran}\left[  \mathbf{1}-P_{0}\left(  \varepsilon\right)
\right]  ,\label{Stg2Splt6}\\
\dim\operatorname{Ran}\left[  P\left(  \varepsilon\right)  P_{0}\left(
\varepsilon\right)  \right]   &  =\dim\operatorname{Ran}P\left(
\varepsilon\right)  -\dim\operatorname{Ran}\left[  \mathbf{1}-P_{0}\left(
\varepsilon\right)  \right]  \nonumber\\
&  =2N-N_{R}-\dim\operatorname{Ran}\left(  P_{B}^{\bot}\Omega P_{B}^{\bot
}\right)  .\nonumber
\end{align}
By making the substitute $\varepsilon=\left(  -\mathrm{i}\beta\right)  ^{-1}$
so that $A\left(  \beta\right)  =\varepsilon^{-1}C\left(  \varepsilon\right)
$, $b\left(  \left\vert \varepsilon\right\vert \right)  =c\left(  \left\vert
\beta\right\vert \right)  $, and by defining%
\begin{equation}
H_{\ell\ell,0}\left(  \beta\right)  =\operatorname{Ran}\left[  P\left(
\varepsilon\right)  P_{0}\left(  \varepsilon\right)  \right]  ,\text{\quad
}H_{\ell\ell,1}\left(  \beta\right)  =\operatorname{Ran}\left[  \mathbf{1}%
-P_{0}\left(  \varepsilon\right)  \right]  ,\label{Stg2Splt6a}%
\end{equation}
the proof of Theorem \ref{tllspmd} follows immediately from this and
Proposition \ref{pspredl2}.

We will now prove Corollary \ref{cllspmd}. Suppose $P_{B}^{\bot}\Omega
P_{B}^{\bot}\not =0$ and $\left\vert \beta\right\vert >c^{-1}\left(
\rho_{\min}/2\right)  $. Then from what we have proved above $\rho_{\min
}-c\left(  \left\vert \beta\right\vert \right)  >\max\left\{  c\left(
\left\vert \beta\right\vert \right)  ,\rho_{\min}/2\right\}  $, $\rho_{\min
}/2\leq\omega_{\max}/2$, and $c\left(  \left\vert \beta\right\vert \right)
\searrow0$ as $\left\vert \beta\right\vert \rightarrow\infty$ and
\begin{align*}
&  \sigma\left(  A\left(  \beta\right)  |_{H_{\ell\ell,0}\left(  \beta\right)
}\right)  =\sigma\left(  A\left(  \beta\right)  \right)  \cap D\left(
0;c\left(  \left\vert \beta\right\vert \right)  \right) \\
&  \sigma\left(  A\left(  \beta\right)  |_{H_{\ell\ell,1}\left(  \beta\right)
}\right)  =\sigma\left(  A\left(  \beta\right)  \right)  \cap%
{\textstyle\bigcup_{0\not =\rho\in\sigma\left(  P_{B}^{\bot}\Omega P_{B}%
^{\bot}\right)  }}
D\left(  \rho;c\left(  \left\vert \beta\right\vert \right)  \right)  .
\end{align*}
Thus if $\zeta\in\sigma\left(  A\left(  \beta\right)  |_{H_{\ell\ell,0}\left(
\beta\right)  }\right)  $ then $\left\vert \zeta\right\vert \leq c\left(
\left\vert \beta\right\vert \right)  $ implies $\left\vert \operatorname{Im}%
\zeta\right\vert \leq c\left(  \beta\right)  $ and $\left\vert
\operatorname{Re}\zeta\right\vert \leq c\left(  \beta\right)  $. Also, if
$\zeta\in\sigma\left(  A\left(  \beta\right)  |_{H_{\ell\ell,0}\left(
\beta\right)  }\right)  $ then there exists $\rho\in\sigma\left(  P_{B}^{\bot
}\Omega P_{B}^{\bot}\right)  $ with $\rho\not =0$ such that $c\left(
\left\vert \beta\right\vert \right)  \geq\left\vert \zeta-\rho\right\vert
\geq\left\vert \rho\right\vert -\left\vert \zeta\right\vert $ implying
$\left\vert \zeta\right\vert \geq\left\vert \rho\right\vert -c\left(
\left\vert \beta\right\vert \right)  >c\left(  \left\vert \beta\right\vert
\right)  >\rho_{\min}/2$ and since $\rho\in$ $%
\mathbb{R}
$ then $\left\vert \operatorname{Re}\zeta-\rho\right\vert ,\left\vert
\operatorname{Im}\zeta\right\vert \leq\left\vert \zeta-\rho\right\vert \leq
c\left(  \left\vert \beta\right\vert \right)  $ and hence $\left\vert
\operatorname{Re}\zeta\right\vert \geq\left\vert \rho\right\vert -c\left(
\left\vert \beta\right\vert \right)  \geq\rho_{\min}-c\left(  \left\vert
\beta\right\vert \right)  \geq\rho_{\min}/2$. Also, if $\beta>c^{-1}\left(
\rho_{\min}/2\right)  $ then $\beta>2\frac{\omega_{\max}}{b_{\min}}$ so that
$\zeta\in\sigma\left(  A\left(  \beta\right)  \right)  $ implies
$0\leq-\operatorname{Im}\zeta=\left\vert \operatorname{Im}\zeta\right\vert $
\ and by Corollary \ref{cmdic} we have $\sigma\left(  A\left(  \beta\right)
|_{H_{\ell\ell}\left(  \beta\right)  }\right)  =\left\{  \zeta\in\sigma\left(
A\left(  \beta\right)  \right)  :0\leq-\operatorname{Im}\zeta\leq\omega_{\max
}\right\}  $. These facts prove that
\begin{align*}
\sigma\left(  A\left(  \beta\right)  |_{H_{\ell\ell,0}\left(  \beta\right)
}\right)   &  =\left\{  \zeta\in\sigma\left(  A\left(  \beta\right)  \right)
:0\leq-\operatorname{Im}\zeta\leq c\left(  \beta\right)  \text{ and
}\left\vert \operatorname{Re}\zeta\right\vert \leq c\left(  \beta\right)
\right\} \\
\sigma\left(  A\left(  \beta\right)  |_{H_{\ell\ell,1}\left(  \beta\right)
}\right)   &  =\left\{  \zeta\in\sigma\left(  A\left(  \beta\right)  \right)
:0\leq-\operatorname{Im}\zeta\leq c\left(  \beta\right)  \text{ and
}\left\vert \operatorname{Re}\zeta\right\vert \geq\rho_{\min}-c\left(
\beta\right)  \right\}
\end{align*}
and%
\[
\max_{\zeta\in\sigma\left(  A\left(  \beta\right)  |_{H_{\ell\ell,1}\left(
\beta\right)  }\right)  }Q_{\zeta}\geq\frac{1}{2}\frac{\rho_{\min}-c\left(
\beta\right)  }{c\left(  \beta\right)  }>\frac{1}{4}\frac{\rho_{\min}%
}{c\left(  \beta\right)  }\geq\frac{1}{2}.
\]
This completes the proof of Corollary \ref{cllspmd}.
\end{proof}

\begin{corollary}
[Underdamped: low-loss/high-Q subspace]\label{cllspmdd}Suppose the condition
(\ref{cnddl}) is true and $N_{R}<N$. Define $\beta_{2}$ by%
\begin{equation}
\beta_{2}=\max\left\{  \min\left\{  c^{-1}\left(  \rho_{\min}/2\right)
,\left(  c^{\flat}\right)  ^{-1}\left(  \rho_{\min}^{\flat}/2\right)
\right\}  ,2\frac{\omega_{\max}}{b_{\min}},2\frac{\omega_{\max}^{\flat}%
}{b_{\min}^{\flat}}\right\}  .\label{DefBeta2}%
\end{equation}
If $\beta>\beta_{2}$ then in addition to Theorems \ref{tmdic}, \ref{tmddI}
being true, the minimum of the quality factors $Q_{\zeta}=\frac{1}{2}%
\frac{\left\vert \operatorname{Re}\zeta\right\vert }{-\operatorname{Im}\zeta}$
for $\zeta\in$ $\sigma\left(  A\left(  \beta\right)  |_{H_{\ell\ell,1}\left(
\beta\right)  }\right)  $ satisfy%
\begin{equation}
\min_{\zeta\in\sigma\left(  A\left(  \beta\right)  |_{H_{\ell\ell,1}\left(
\beta\right)  }\right)  }Q_{\zeta}>\frac{1}{2}.
\end{equation}
In particular, $\operatorname{Re}\zeta\not =0$ for every $\zeta\in$
$\sigma\left(  A\left(  \beta\right)  |_{H_{\ell\ell,1}\left(  \beta\right)
}\right)  $.
\end{corollary}

\begin{proof}
Suppose (\ref{cnddl}) and $N_{R}<N$ \ are true and let $\beta_{2}$ be defined
by (\ref{DefBeta2}). In particular, $\beta_{2}\geq\max\left\{  2\frac
{\omega_{\max}}{b_{\min}},2\frac{\omega_{\max}^{\flat}}{b_{\min}^{\flat}%
}\right\}  $. Hence if $\beta>\beta_{2}$ then Theorems \ref{tmdic} and
\ref{tmddI} are true. If $\min\left\{  c^{-1}\left(  \rho_{\min}/2\right)
,\left(  c^{\flat}\right)  ^{-1}\left(  \rho_{\min}^{\flat}/2\right)
\right\}  =c^{-1}\left(  \rho_{\min}/2\right)  $ then by Proposition
\ref{pspredl2}.6 we must have $P_{B}^{\bot}\Omega P_{B}^{\bot}\not =0$ and so
this corollary follows immediately from Corollary \ref{cllspmd}. Thus, suppose
$\min\left\{  c^{-1}\left(  \rho_{\min}/2\right)  ,\left(  c^{\flat}\right)
^{-1}\left(  \rho_{\min}^{\flat}/2\right)  \right\}  =\left(  c^{\flat
}\right)  ^{-1}\left(  \rho_{\min}^{\flat}/2\right)  $. Then for the dual
Lagrangian system by Proposition \ref{pspredl2}.6 we must have $P_{B^{\flat}%
}^{\bot}\Omega^{\flat}P_{B^{\flat}}^{\bot}\not =0$ and so by Corollary
\ref{cllspmd} we have the minimum of the quality factors $Q_{\zeta}=\frac
{1}{2}\frac{\left\vert \operatorname{Re}\zeta\right\vert }{-\operatorname{Im}%
\zeta}$ for $\zeta\in$ $\sigma\left(  A^{\flat}\left(  \beta\right)
|_{H_{\ell\ell,1}^{\flat}\left(  \beta\right)  }\right)  $ satisfy%
\begin{equation}
\min_{\zeta\in\sigma\left(  A^{\flat}\left(  \beta\right)  |_{H_{\ell\ell
,1}^{\flat}\left(  \beta\right)  }\right)  }Q_{\zeta}\geq\frac{1}{2}\frac
{\rho_{\min}^{\flat}-c^{\flat}\left(  \beta\right)  }{c^{\flat}\left(
\beta\right)  }>\frac{1}{4}\frac{\rho_{\min}^{\flat}}{c^{\flat}\left(
\beta\right)  }\geq\frac{1}{2}.
\end{equation}
But by Theorem \ref{tmddI}%
\begin{equation}
\sigma\left(  A\left(  \beta\right)  |_{H_{\ell\ell,1}\left(  \beta\right)
}\right)  =-\overline{\sigma\left(  A\left(  \beta\right)  |_{H_{\ell\ell
,1}\left(  \beta\right)  }\right)  }=-\sigma\left(  A^{\flat}\left(
\beta\right)  |_{H_{\ell\ell,1}^{\flat}\left(  \beta\right)  }\right)  ^{-1}.
\end{equation}
Therefore, by the equivalence of the Q-factors as described in Section
\ref{subsDualSys} on duality we have%
\[
Q_{-\zeta^{-1}}=Q_{\zeta}\text{, for all }\zeta\in\sigma\left(  A\left(
\beta\right)  |_{H_{\ell\ell,1}\left(  \beta\right)  }\right)
\]
implying that%
\begin{equation}
\min_{\zeta\in\sigma\left(  A\left(  \beta\right)  |_{H_{\ell\ell,1}\left(
\beta\right)  }\right)  }Q_{\zeta}=\min_{\zeta\in\sigma\left(  A^{\flat
}\left(  \beta\right)  |_{H_{\ell\ell,1}^{\flat}\left(  \beta\right)
}\right)  }Q_{\zeta}\geq\frac{1}{2}\frac{\rho_{\min}^{\flat}-c^{\flat}\left(
\beta\right)  }{c^{\flat}\left(  \beta\right)  }>\frac{1}{4}\frac{\rho_{\min
}^{\flat}}{c^{\flat}\left(  \beta\right)  }\geq\frac{1}{2}.
\end{equation}
In particular, $\operatorname{Re}\zeta\not =0$ for every $\zeta\in$
$\sigma\left(  A\left(  \beta\right)  |_{H_{\ell\ell,1}\left(  \beta\right)
}\right)  $. This proves the corollary.
\end{proof}

\subsection{Spectral perturbation theory: high-loss regime\label{smdhlr}}

We are interested in describing the spectrum $\sigma\left(  A\left(
\beta\right)  \right)  $ of the system operator $A\left(  \beta\right)
=\Omega-\mathrm{i}\beta B$, $\beta\geq0$ in the high-loss regime (i.e.,
$\beta\gg1$) and, in particular, giving an asymptotic characterization, as
$\beta\rightarrow\infty$, of the modal dichotomy as described in Sec.
\ref{sevmd}. \ In order to do so we need to give a spectral perturbation
analysis of the matrix $A\left(  \beta\right)  $ as $\beta\rightarrow\infty$.
\ Fortunately,\ most of this analysis has already been carried out in
\cite{FigWel1} and \cite{FigWel2}. Our goal here is to extend these results by
appealing to the duality and using our results on the modal dichotomy. To do
this, we will begin by introducing the necessary notion to describe the
results from \cite{FigWel1}, \cite{FigWel2} and then describe the perturbation
theory in the high-loss regime in terms of the modal dichotomy results in Sec.
\ref{sevmd} based on duality.

The Hilbert space $H=%
\mathbb{C}
^{2N}$ with standard inner product $\left(  \cdot,\cdot\right)  $ is
decomposed into the direct sum of orthogonal invariant subspaces of the
operator $B$, namely,
\begin{equation}
H=H_{B}\oplus H_{B}^{\bot},\qquad\dim H_{B}=N_{R},\label{pod5}%
\end{equation}
where $H_{B}=\operatorname{Ran}B$ (the range of $B$) is the loss subspace of
dimension $N_{R}=\operatorname{rank}B$ with orthogonal projection $P_{B}$ and
its orthogonal complement, $H_{B}^{\bot}=\operatorname{Ker}B$ (the nullspace
of $B$), is the no-loss subspace of dimension $2N-N_{R}$ with orthogonal
projection $P_{B}^{\bot}$.

The operators $\Omega$ and $B$ with respect to the direct sum (\ref{pod5}) are
the $2\times2$ block operator matrices%
\begin{equation}
\Omega=\left[
\begin{array}
[c]{cc}%
\Omega_{2} & \Theta\\
\Theta^{\ast} & \Omega_{1}%
\end{array}
\right]  ,\qquad B=\left[
\begin{array}
[c]{cc}%
B_{2} & 0\\
0 & 0
\end{array}
\right]  , \label{pod7}%
\end{equation}
where $\Omega_{2}=\left.  P_{B}\Omega P_{B}\right\vert _{H_{B}}:H_{B}%
\rightarrow H_{B}$ and $B_{2}=\left.  P_{B}BP_{B}\right\vert _{H_{B}}%
:H_{B}\rightarrow H_{B}$ are restrictions of the operators $\Omega$ and $B$
respectively to loss subspace $H_{B}$ whereas $\Omega_{1}=\left.  P_{B}^{\bot
}\Omega P_{B}^{\bot}\right\vert _{H_{B}^{\bot}}:H_{B}^{\bot}\rightarrow
H_{B}^{\bot}$ is the restriction of $\Omega$ to complementary subspace
$H_{B}^{\bot}$. \ Also, $\Theta:H_{B}^{\bot}\rightarrow H_{B}$ is the operator
$\Theta=\left.  P_{B}\Omega P_{B}^{\bot}\right\vert _{H_{B}^{\bot}}$ whose
adjoint is given by $\Theta^{\ast}=\left.  P_{B}^{\bot}\Omega P_{B}\right\vert
_{H_{B}}:H_{B}\rightarrow H_{B}^{\bot}$.

The following condition will be important in our study of overdamping.

\begin{condition}
\label{cndgc}The generic condition is the case in which the operator
\[
B_{2}=\left.  P_{B}BP_{B}\right\vert _{H_{B}}:H_{B}\rightarrow H_{B}%
\]
has distinct eigenvalues (since $\sigma\left(  B_{2}\right)  =\sigma\left(
B\right)  \setminus\left\{  0\right\}  =\sigma\left(  \alpha^{-1}R\right)
\setminus\{0\}=\left\{  b_{1},\ldots,b_{N_{R}}\right\}  $ this just means
$b_{i}\not =b_{j}$ if $i\not =j$) and then we say we are in the generic case
(and nongeneric otherwise).
\end{condition}

The perturbation analysis in the high-loss regime $\beta\gg1$ for the system
operator $A(\beta)$ described in \cite[\S VI.A, Theorem 5 \& Proposition
11]{FigWel1} introduces an orthonormal basis $\left\{  \mathring{w}%
_{j}\right\}  _{j=1}^{2N}$ diagonalizing the self-adjoint operators
$\Omega_{1}$ and $B_{2}>0$ from (\ref{pod7}) with%
\begin{equation}
B_{2}\mathring{w}_{j}=b_{j}\mathring{w}_{j}\text{ for }1\leq j\leq
N_{R};\qquad\Omega_{1}\mathring{w}_{j}=\rho_{j}\mathring{w}_{j}\text{ for
}N_{R}+1\leq j\leq2N, \label{pod8}%
\end{equation}
Then for $\beta\gg1$ the system operator $A(\beta)$ is diagonalizable with
basis of eigenvectors $\left\{  w_{j}\left(  \beta\right)  \right\}
_{j=1}^{2N}$ satisfying%
\begin{equation}
A\left(  \beta\right)  w_{j}\left(  \beta\right)  =\zeta_{j}\left(
\beta\right)  w_{j}\left(  \beta\right)  ,\qquad1\leq j\leq2N,\text{\quad
}\beta\gg1 \label{pod10}%
\end{equation}
which split into two distinct classes
\begin{gather}
\text{high-loss}\text{:$\quad$}\zeta_{j}\left(  \beta\right)  ,\text{ }%
w_{j}\left(  \beta\right)  ,\text{$\quad$}1\leq j\leq N_{R};\label{pod11}\\
\text{low-loss}\text{:$\quad$}\zeta_{j}\left(  \beta\right)  ,\text{ }%
w_{j}\left(  \beta\right)  ,\text{$\quad$}N_{R}+1\leq j\leq2N,\nonumber
\end{gather}
with the following properties.

\textbf{The high-loss class:} the eigenvalues have poles at $\beta=\infty$
whereas their eigenvectors are analytic at $\beta=\infty$, having the
asymptotic expansions%
\begin{align}
\zeta_{j}\left(  \beta\right)   &  =-\mathrm{i}b_{j}\beta+\rho_{j}+O\left(
\beta^{-1}\right)  ,\text{$\quad$}b_{j}>0,\text{$\quad$}\rho_{j}=\left(
\mathring{w}_{j},\Omega\mathring{w}_{j}\right)  ,\text{ }\label{pod12}\\
w_{j}\left(  \beta\right)   &  =\mathring{w}_{j}+O\left(  \beta^{-1}\right)
,\text{$\quad$}1\leq j\leq N_{R}.\nonumber
\end{align}

\textbf{The low-loss class:} the eigenvalues and eigenvectors are analytic at
$\beta=\infty$, having the asymptotic expansions%
\begin{align}
\zeta_{j}\left(  \beta\right)   &  =\rho_{j}-\mathrm{i}d_{j}\beta
^{-1}+O\left(  \beta^{-2}\right)  ,\text{$\quad$}d_{j}=\left(  \mathring
{w}_{j},\Theta^{\ast}B_{2}^{-1}\Theta\mathring{w}_{j}\right)  \geq0,\text{
}\label{pod14}\\
w_{j}\left(  \beta\right)   &  =\mathring{w}_{j}+O\left(  \beta^{-1}\right)
,\text{$\quad$}N_{R}+1\leq j\leq2N.\nonumber
\end{align}

By \cite[\S VI.A, Proposition 7]{FigWel1} we know that all the frequencies
$\operatorname{Re}\zeta_{j}\left(  \beta\right)  $ have convergent Taylor
series expansions in only even powers in $z=\beta^{-1}$, whereas the damping
factors $-\operatorname{Im}\zeta_{j}\left(  \beta\right)  $ either have
convergent Laurant series expansions with only odd powers in $z=\beta^{-1}$ or
$-\operatorname{Im}\zeta_{j}\left(  \beta\right)  \equiv0$. And, moreover,
have the asymptotic expansions as $\beta\rightarrow\infty$,
\begin{gather}
\operatorname{Re}\zeta_{j}\left(  \beta\right)  =\rho_{j}+O\left(  \beta
^{-2}\right)  ,\text{ }-\operatorname{Im}\zeta_{j}\left(  \beta\right)
=b_{j}\beta+O\left(  \beta^{-1}\right)  ,\text{ }1\leq j\leq N_{R}%
;\label{pod16}\\
\operatorname{Re}\zeta_{j}\left(  \beta\right)  =\rho_{j}+O\left(  \beta
^{-2}\right)  ,\text{ }-\operatorname{Im}\zeta_{j}\left(  \beta\right)
=d_{j}\beta^{-1}+O\left(  \beta^{-3}\right)  ,\text{ }N_{R}+1\leq
j\leq2N.\nonumber
\end{gather}

The following theorems give a characterization the spectrum $\sigma\left(
A\left(  \beta\right)  \right)  $ of the system operator $A\left(
\beta\right)  $ and the modal dichotomy in the high-loss regime $\beta\gg1$ in
terms of the high-loss and low-loss eigenpairs.

\begin{theorem}
[modal dichotomy III]\label{tmdicII}For the loss parameter $\beta$
sufficiently large, the modal dichotomy occurs as in Theorem \ref{tmdic} with
the following equalities holding:
\begin{gather}
\sigma\left(  A\left(  \beta\right)  |_{H_{\ell\ell}\left(  \beta\right)
}\right)  =\left\{  \zeta_{j}\left(  \beta\right)  :N_{R}+1\leq j\leq
2N\right\}  ,\label{pod18}\\
\sigma\left(  A\left(  \beta\right)  |_{H_{h\ell}\left(  \beta\right)
}\right)  =\left\{  \zeta_{j}\left(  \beta\right)  :1\leq j\leq N_{R}\right\}
,\nonumber\\
H_{\ell\ell}\left(  \beta\right)  =\operatorname*{span}\left\{  w_{j}\left(
\beta\right)  :N_{R}+1\leq j\leq2N\right\}  ,\label{pod19}\\
H_{h\ell}\left(  \beta\right)  =\operatorname*{span}\left\{  w_{j}\left(
\beta\right)  :1\leq j\leq N_{R}\right\}  .\nonumber
\end{gather}
In particular, the high-loss eigenvectors $\left\{  w_{j}\left(  \beta\right)
\right\}  _{j=1}^{N_{R}}$ and the low-loss eigenvectors $\left\{  w_{j}\left(
\beta\right)  \right\}  _{j=N_{R}+1}^{2N}$ are a basis for $H_{h\ell}\left(
\beta\right)  $ and $H_{\ell\ell}\left(  \beta\right)  $, respectively.
\end{theorem}

\begin{theorem}
[modal dichotomy IV]\label{tmddII}Suppose that $N_{R}<N$. Then for $\beta$
sufficiently large, Theorems \ref{tmdic} and \ref{tllspmd} are true as are
Corollaries \ref{cmdic} and \ref{cllspmd} and furthermore,%
\begin{align}
H_{h\ell}\left(  \beta\right)   &  =\operatorname*{span}\left\{  w_{j}\left(
\beta\right)  :\lim_{\beta\rightarrow\infty}\left\vert \zeta_{j}\left(
\beta\right)  \right\vert =\infty\right\}  ,\label{tmdd1b}\\
H_{\ell\ell,0}\left(  \beta\right)   &  =\operatorname*{span}\left\{
w_{j}\left(  \beta\right)  :\lim_{\beta\rightarrow\infty}\zeta_{j}\left(
\beta\right)  =\rho_{j}\text{ and }\rho_{j}=0\right\}  ,\label{tmdd1c}\\
H_{\ell\ell,1}\left(  \beta\right)   &  :=\operatorname*{span}\left\{
w_{j}\left(  \beta\right)  \in H_{\ell\ell}\left(  \beta\right)  :\lim
_{\beta\rightarrow\infty}\zeta_{j}\left(  \beta\right)  =\rho_{j}\text{ and
}\rho_{j}\not =0\right\}  .\label{tmdd1d}%
\end{align}

\end{theorem}

\begin{proof}
The first theorem was proved in \cite{FigWel2}. The second theorem follows
immediately from the perturbation theory above, Theorems \ref{tmdic} and
\ref{tllspmd}, and Corollaries \ref{cmdic} and \ref{cllspmd}.
\end{proof}

Now we will prove an important theorem on the asymptotics of the quality
factor. First, if $\operatorname{Ker}R\cap\operatorname{Ker}\eta=\left\{
0\right\}  $ (such as if the duality condition (\ref{cnddl}) is true) then by
Theorems \ref{tllspmd} and \ref{tmddII} we know that we can reindex the
eigenpairs $\left\{  w_{j}\left(  \beta\right)  ,\text{ }\zeta_{j}\left(
\beta\right)  \right\}  _{j=1}^{2N}$ such that%
\begin{equation}
\text{high-loss}\text{:\quad}\lim_{\beta\rightarrow\infty}\left\vert \zeta
_{j}\left(  \beta\right)  \right\vert =\infty,\text{\quad}1\leq j\leq N_{R};
\label{pod23_0}%
\end{equation}%
\begin{equation}
\text{low-loss/asymp. overdamped:\quad}\lim_{\beta\rightarrow\infty}\zeta
_{j}\left(  \beta\right)  =\rho_{j}=0,\text{\quad}N_{R}+1\leq j\leq2N_{R};
\label{pod24_0}%
\end{equation}%
\begin{equation}
\text{low-loss/asymp. underdamped}:\ \lim_{\beta\rightarrow\infty}\zeta
_{j}\left(  \beta\right)  =\rho_{j}\not =0,\ 2N_{R}+1\leq j\leq2N.
\label{pod25_0}%
\end{equation}

\begin{corollary}
[duality asymptotics]\label{cdasym}If the duality condition (\ref{cnddl}) is
true then we can reindex the eigenpairs $\left\{  w_{j}\left(  \beta\right)
,\text{ }\zeta_{j}\left(  \beta\right)  \right\}  _{j=N_{R}+1}^{2N_{R}}$ that
have the property (\ref{pod24_0}) so that they have the asymptotic expansions
as $\beta\rightarrow\infty$,%
\[
\text{low-loss/asymp. overdamped: }\zeta_{j}\left(  \beta\right)
=-\mathrm{i}\frac{1}{b_{j-N_{R}}^{\flat}}\beta^{-1}+O\left(  \beta
^{-2}\right)  ,\text{\quad}N_{R}+1\leq j\leq2N_{R}.
\]

\end{corollary}

\begin{proof}
Considering the high-loss modes $\left\{  w_{j}^{\flat}\left(  \beta\right)
,\text{ }\zeta_{j}^{\flat}\left(  \beta\right)  \right\}  _{j=1}^{N_{R}}$ of
the dual Lagrangian system and the low-loss/asymp. overdamped modes $\left\{
w_{j}\left(  \beta\right)  ,\text{ }\zeta_{j}\left(  \beta\right)  \right\}
_{j=N_{R}+1}^{2N_{R}}$ of the Lagrangian system. It follows for $\beta$
sufficiently large that there is a one-to-one correspondence between the
functions $-\left[  \zeta_{j}^{\flat}\left(  \beta\right)  \right]  ^{-1}$,
$1\leq j\leq N_{R}$ and the functions $\zeta_{j}\left(  \beta\right)  $,
$N_{R}+1\leq j\leq2N_{R}$ as they are analytic eigenfunctions of $A\left(
\beta\right)  $ in the variable $z=\beta^{-1}$ near $z=0$ and as sets they are
equal by Theorem \ref{tmddI}. From this the proof immediately follows.
\end{proof}

\begin{theorem}
[Quality factor-duality]\label{tmqf}As $\beta\rightarrow\infty$ the quality
factors $Q_{\zeta_{j}\left(  \beta\right)  }$ of the high-loss modes are
decreasing functions of $\beta$, i.e.,%
\[
\text{high-loss modes}\text{:\quad}Q_{\zeta_{j}\left(  \beta\right)  }%
\searrow0,\ \ 1\leq j\leq N_{R}.
\]
If $\operatorname{Ker}R\cap\operatorname{Ker}\eta=\left\{  0\right\}  $ then
as $\beta\rightarrow\infty$ the quality factors $Q_{\zeta_{j}\left(
\beta\right)  }$ [indexed according to (\ref{pod24_0}) and (\ref{pod25_0})]
are either decreasing or increasing as functions of $\beta$ and, in
particular,%
\begin{align*}
\text{low-loss, low-Q modes}  &  \text{:\quad\ }Q_{\zeta_{j}\left(
\beta\right)  }\searrow0,\text{\quad}N_{R}+1\leq j\leq2N_{R};\\
\text{low-loss, high-Q modes}  &  \text{: \quad}Q_{\zeta_{j}\left(
\beta\right)  }\nearrow+\infty,\text{\quad}2N_{R}+1\leq j\leq2N.
\end{align*}

\end{theorem}

\begin{proof}
\noindent\noindent The proof of this theorem follows immediately from the
perturbation theory above, Theorems \ref{tmdic}, \ref{tllspmd}, and
Corollaries \ref{cmdic}, \ref{cllspmd}.
\end{proof}

This theorem is one of the main results of our paper since it says that as
long as the duality condition (\ref{cnddl}) holds (or even the weaker
hypothesis $\operatorname{Ker}R\cap\operatorname{Ker}\eta=\left\{  0\right\}
$) then as $\beta\rightarrow\infty$ (\textit{i.e.}, as losses in the lossy
component approach infinity), all $N_{R}=\operatorname{rank}R$ of the
high-loss modes have their quality factor going to zero and an equal number,
$N_{R}$, of low-loss modes are asymptotically overdamped with quality factor
going to zero, and the remaining $2\left(  N-N_{R}\right)  $ low-loss modes
which are underdamped with quality factor approaching infinity.

\subsection{Overdamping analysis\label{sodgd}}

Overdamping phenomena has already been studied for nongyroscopic-dissipative
systems (i.e., $\theta=0$) in \cite{FigWel2} and some subtleties have already
been discussed in Subsection \ref{sinsubod}). As we will show in this section,
the introduction of gyroscopy, i.e., $\theta\not =0$ [and in the generic case,
i.e., under the generic condition \ref{cndgc} for both the Lagrangian system
(\ref{sintro1}) and its dual system (\ref{dradis2a})], does not change
qualitatively the overdamping phenomena as described in \cite{FigWel2} for the
non-gyroscopic case ($\theta=0$) and the only thing that changes significantly
is the analysis (which is now based on the duality principle which we have
introduced above). Moreover, we will show that the only difference that occurs
is in the nongeneric case and we will demonstrate this by giving an extreme
example showing that when the generic condition \ref{cndgc} is not satisfied
it is possible for all the eigenmodes to be underdamped not only in the
high-loss regime $\beta\gg1$, but for all $\beta>0$.

\subsubsection{Overdamping in the generic case}

The following theorems and their corollaries, along with Corollary
\ref{cllspmdd} in Sec. \ref{sevmd},\ are the main results of our paper on
overdamping phenomena.

\begin{theorem}
[Selective overdamping]\label{tsogc}If the generic condition \ref{cndgc} is
true then all the high-loss eigenvalues $\zeta_{j}\left(  \beta\right)  $,
$1\leq j\leq N_{R}$ (counting multiplicities) of the system operator $A\left(
\beta\right)  $ have the property for $\beta\gg1$ (i.e., for $\beta$
sufficiently large):
\[
\operatorname{Re}\zeta_{j}\left(  \beta\right)  =0,\text{ for }1\leq j\leq
N_{R}.
\]
Moreover, if $\operatorname{Ker}R\cap\operatorname{Ker}\eta=\left\{
0\right\}  $ then all the low-loss eigenvalues $\zeta_{j}\left(  \beta\right)
$, $N_{R}+1\leq j\leq2N$ (counting multiplicities) of the system operator
$A\left(  \beta\right)  $, indexed according to (\ref{pod24_0}) and
(\ref{pod25_0}), have the following properties for $\beta\gg1$:
\begin{align*}
\operatorname{Re}\zeta_{j}\left(  \beta\right)   &  \not =0,\text{ for }%
2N_{R}+1\leq j\leq2N;\\
\lim_{\beta\rightarrow\infty}\operatorname{Re}\zeta_{j}\left(  \beta\right)
&  =0,\text{ \ for }N_{R}+1\leq j\leq2N_{R}.
\end{align*}

\end{theorem}

\begin{corollary}
[Selective overdamping-duality]\label{csogc}If the duality condition
(\ref{cnddl}) is true and the generic condition \ref{cndgc} is true for both
the Lagrangian system (\ref{sintro1}) and its dual system (\ref{dradis2a})
then all of Theorem \ref{tsogc} is true and, in addition, for $\beta\gg1$,
\[
\operatorname{Re}\zeta_{j}\left(  \beta\right)  =0,\ \ \text{for }N_{R}+1\leq
j\leq2N_{R}.
\]

\end{corollary}

\begin{proof}
By Proposition \ref{pevbd} and Theorem \ref{tmdic} we know that for $\beta
\gg1$,
\[
\sigma\left(  A\left(  \beta\right)  |_{H_{h\ell}\left(  \beta\right)
}\right)  =-\overline{\sigma\left(  A\left(  \beta\right)  |_{H_{h\ell}\left(
\beta\right)  }\right)  }.
\]
It follows from this and Theorem \ref{tmdicII} that the high-loss eigenvalues
come in pairs%
\[
\zeta_{j}\left(  \beta\right)  ,-\text{ }\overline{\zeta_{j}\left(
\beta\right)  }\in\sigma\left(  A\left(  \beta\right)  |_{H_{h\ell}\left(
\beta\right)  }\right)
\]
for $1\leq j\leq N_{R}$ and so (since all $\zeta_{j}\left(  \beta\right)  $,
$-$ $\overline{\zeta_{j}\left(  \overline{\beta}\right)  }$ for $1\leq j\leq
N_{R}$ are meromorphic in $\beta^{-1}$) there must exist $j^{^{\prime}}$ with
$1\leq j^{^{\prime}}\leq N_{R}$ such that%
\[
-\text{ }\overline{\zeta_{j}\left(  \beta\right)  }=\zeta_{j^{^{\prime}}%
}\left(  \beta\right)
\]
for $\beta\gg1$. By (\ref{pod8}) and (\ref{pod12}), this implies $B_{2}$ has a
repeated eigenvalue $b_{j}$ unless $-\overline{\zeta_{j}\left(  \beta\right)
}=\zeta_{j}\left(  \beta\right)  $ for $\beta\gg1$. By hypothesis the generic
condition \ref{cndgc} holds so that we must have $-$ $\overline{\zeta
_{j}\left(  \beta\right)  }=\zeta_{j}\left(  \beta\right)  $ for $1\leq j\leq
N_{R}$ and for $\beta\gg1$. \ Therefore, for $\beta\gg1$ we have proven that
$\operatorname{Re}\zeta_{j}\left(  \beta\right)  =0$\ for $1\leq j\leq N_{R}$.
The proof of this theorem now follows immediately from this and Theorems
\ref{tllspmd} and \ref{tmddII}. The corollary follows immediately from Theorem
\ref{tsogc} and duality\ by appealing to Theorems \ref{tmddI}, \ref{tmddII}.
This completes the proof.
\end{proof}

\begin{theorem}
[Estimate of overdamped regime]\label{tbeta0}If the generic condition
\ref{cndgc} is true then the high-loss eigenvalues $\left\{  \zeta_{j}\left(
\beta\right)  \right\}  _{j=1}^{N_{R}}$ of the system operator $A\left(
\beta\right)  $ are meromorphic in $z=\beta^{-1}$ at $z=0$ which all converge
in a punctured disk of radius of $\beta_{0}^{-1}$, where
\[
\beta_{0}=\frac{2\omega_{\max}}{d},\qquad d:=\min_{0\leq i,j\leq N_{R},\text{
}i\not =j}\left\vert b_{i}-b_{j}\right\vert
\]
and $\sigma\left(  B\right)  =\left\{  b_{0},b_{1},\ldots,b_{N_{R}}\right\}  $
with $b_{0}=0$. Furthermore, their corresponding eigenprojections $\left\{
P_{j}\left(  \beta\right)  \right\}  _{j=1}^{N_{R}}$ are analytic in this disc
with $\dim\operatorname{Ran}P_{j}\left(  \beta\right)  =1$ for $1\leq j\leq
N_{R}$, in particular, the high-loss eigenvalues are simple eigenvalues of
$A\left(  \beta\right)  $. Moreover, if $\beta>\beta_{0}$ then
\[
\operatorname{Re}\zeta_{j}\left(  \beta\right)  =0,\text{\quad for }1\leq
j\leq N_{R}\text{.}%
\]

\end{theorem}

\begin{proof}
For the system operator%
\[
A\left(  \beta\right)  =\Omega-\mathrm{i}\beta B,
\]
making the substitution $\varepsilon=\left(  -\mathrm{i}\beta\right)  ^{-1}$
we have that
\[
\varepsilon A\left(  \mathrm{i}\varepsilon^{-1}\right)  =B+\varepsilon\Omega
\]
is an analytic operator in $\varepsilon\in%
\mathbb{C}
$ which is self-adjoint for real $\varepsilon$ in which $b_{1},\ldots
,b_{N_{R}}$ are all the nonzero eigenvalues of $B$ and by the generic
condition \ref{cnddl} they are all simple eigenvalues too. By \cite[pp.
324-326, \S 8.1.3, Theorem 1 \& 2]{Bau85}, for each $j\in\left\{  1,\ldots
N_{R}\right\}  $, there is a unique simple eigenvalue $\lambda_{j}\left(
\varepsilon\right)  $ with $\lambda_{j}\left(  0\right)  =b_{j}$ and
one-dimensional eigenprojection $Q_{j}\left(  \varepsilon\right)  $ which are
analytic in $\varepsilon$ near $\varepsilon=0$ with a radius of convergence
greater than or equal to $r_{j}:=\frac{d_{j}}{2\left\Vert \Omega\right\Vert }%
$, where%
\[
\omega_{\max}=\left\Vert \Omega\right\Vert ,\text{\qquad}d_{j}:=\min_{0\leq
i\leq N_{R},\text{ }i\not =j}\left\vert b_{i}-b_{j}\right\vert ,
\]
i.e., $d_{j}$ is the distance of $b_{j}$ to the rest of the spectrum of $B$.
We now define%
\[
d:=\min_{0\leq j\leq N_{R}}d_{j}=\min_{0\leq i,j\leq N_{R},\text{ }i\not =%
j}\left\vert b_{i}-b_{j}\right\vert .
\]
It follows that the follow eigenprojections are analytic%
\[
Q_{j}\left(  \varepsilon\right)  ,\text{\qquad}1\leq i,j\leq N_{R}%
,\text{\qquad\ }\left\vert \varepsilon\right\vert <\beta_{0}^{-1}%
,\text{\qquad}\beta_{0}:=\frac{2\left\Vert \Omega\right\Vert }{d},
\]
and%
\[
\dim\operatorname{Ran}Q_{j}\left(  \varepsilon\right)  =1,\text{\quad for
}1\leq j\leq N_{R},\text{\qquad}\left\vert \varepsilon\right\vert <\beta
_{0}^{-1}.
\]
The eigenprojection-eigenvalue pairs $\left\{  Q_{j}\left(  \varepsilon
\right)  ,\text{ }\lambda_{j}\left(  \beta\right)  \right\}  _{j=1}^{N_{R}}$
satisfy%
\[
\varepsilon A\left(  \mathrm{i}\varepsilon^{-1}\right)  Q_{j}\left(
\varepsilon\right)  =\left(  B+\varepsilon\Omega\right)  Q_{j}\left(
\varepsilon\right)  =\lambda_{j}\left(  \varepsilon\right)  Q_{j}\left(
\varepsilon\right)  ,\text{\quad}1\leq i,j\leq N_{R},\text{\quad}\left\vert
\varepsilon\right\vert <\beta_{0}^{-1}.
\]
Thus, making the substitution $\varepsilon=\left(  -\mathrm{i}\beta\right)
^{-1}$ and multiplying by $\varepsilon^{-1}$ yields%
\[
A\left(  \beta\right)  Q_{j}\left(  \left(  -\mathrm{i}\beta\right)
^{-1}\right)  =\left(  -\mathrm{i}\beta\right)  \lambda_{j}\left(  \left(
-\mathrm{i}\beta\right)  ^{-1}\right)  Q_{j}\left(  \left(  -\mathrm{i}%
\beta\right)  ^{-1}\right)  ,\text{\quad}1\leq i,j\leq N_{R},\text{\quad
}\left\vert \beta\right\vert >\beta_{0}.
\]
Therefore, defining
\[
\zeta_{j}\left(  \beta\right)  :=\left(  -\mathrm{i}\beta\right)  \lambda
_{j}\left(  \left(  -\mathrm{i}\beta\right)  ^{-1}\right)  ,\text{\quad}%
P_{j}\left(  \beta\right)  :=Q_{j}\left(  \left(  -\mathrm{i}\beta\right)
^{-1}\right)  ,\text{\quad}1\leq i,j\leq N_{R},\text{\quad}\left\vert
\beta\right\vert >\beta_{0}%
\]
the theorem now follows immediately from this for the high-loss eigenvalues
and their eigenprojections and from Theorem \ref{tsogc}. This completes the proof.
\end{proof}

\begin{corollary}
[Estimate of overdamped regime-duality]\label{cbeta0}If the duality condition
(\ref{cnddl}) and the generic condition \ref{cndgc} are true for both the
Lagrangian system (\ref{sintro1}) and its dual system (\ref{dradis2a}) then
Theorem \ref{tbeta0} is true and, for the low-loss eigenvalues of the system
operator $A\left(  \beta\right)  $ indexed according to (\ref{pod24_0}) and
(\ref{pod25_0}), the eigenvalues $\left\{  \zeta_{j}\left(  \beta\right)
\right\}  _{j=N_{R}+1}^{2N_{R}}$ are analytic in $z=\beta^{-1}$ at $z=0$ and
each converges in a punctured disk of radius of $\beta_{1}^{-1}$, where
\[
\beta_{1}=\max\left\{  \beta_{0},\frac{2\omega_{\max}^{\flat}}{d^{\flat}%
}\right\}  ,\text{\qquad}d^{\flat}:=\min_{0\leq i,j\leq N_{R},\text{ }%
i\not =j}\left\vert b_{i}^{\flat}-b_{j}^{\flat}\right\vert ,
\]
and $\sigma\left(  B^{\flat}\right)  =\left\{  b_{0}^{\flat},b_{1}^{\flat
},\ldots,b_{N_{R}}^{\flat}\right\}  $ with $b_{0}^{b}=0$. Furthermore, their
corresponding eigenprojections $\left\{  P_{j}\left(  \beta\right)  \right\}
_{N_{R}+1}^{2N_{R}}$ are analytic in this disc with $\dim\operatorname{Ran}%
P_{j}\left(  \beta\right)  =1$ for $N_{R}+1\leq j\leq2N_{R}$. In particular,
if $\beta>\beta_{1}$ then $\left\{  \zeta_{j}\left(  \beta\right)  \right\}
_{j=N_{R}+1}^{2N_{R}}$ are simple eigenvalues of $A\left(  \beta\right)  $
and
\[
\operatorname{Re}\zeta_{j}\left(  \beta\right)  =0,\text{\quad for }%
N_{R}+1\leq j\leq2N_{R}\text{.}%
\]

\end{corollary}

\begin{proof}
The result follows immediately from Theorems \ref{tmddI} and \ref{tbeta0} by duality.
\end{proof}

\subsubsection{Overdamping in the nongeneric case\label{sodngc}}

If the generic condition \ref{cndgc} doesn't hold (i.e., the nongeneric case)
then as we will show one can build examples where no overdamping occurs in the
high-loss regime $\beta\gg1$ from which one can build mix cases.

\begin{example}
[no overdamping]\label{nodex}Take the $N\times N$ identity matrix $\alpha
=\eta=R=\mathbf{1}$ (so that the duality condition \ref{cnddl} is satisfied),
the loss parameter $\beta\geq0$, and any real $N\times N$ matrix $\theta$
satisfying $\theta^{\mathrm{T}}=-\theta$. Then one can find a $N\times N$
unitary matrix $M$ such that $\mathrm{i}\theta=M\operatorname{diag}\left\{
\lambda_{1},\ldots,\lambda_{N}\right\}  M^{-1}$, with $\sigma\left(
\mathrm{i}\theta\right)  =\left\{  \lambda_{1},\ldots,\lambda_{N}\right\}
\subseteq%
\mathbb{R}
$. Hence, the Lagrangian system (\ref{sintro1}) (which is it's own dual system
in this example) for these matrices is%
\begin{gather*}
0=\alpha\ddot{Q}+\left(  2\theta+\beta R\right)  \dot{Q}+\eta Q=\\
=M\operatorname{diag}\left\{  \partial_{t}^{2}+\left(  -2\mathrm{i}\lambda
_{1}+\beta\right)  \partial_{t}+1,\ldots,\partial_{t}^{2}+\left(
-2\mathrm{i}\lambda_{N}+\beta\right)  \partial_{t}+1\right\}  M^{-1}Q.
\end{gather*}
A calculation of the matrix $B$ and its spectrum for this example is%
\[
B=%
\begin{bmatrix}
\mathbf{1} & 0\\
0 & 0
\end{bmatrix}
,\text{\quad}\sigma\left(  B\right)  =\left\{  0,1\right\}  ,\text{\quad}%
b_{0}=0,1=b_{1}=\cdots=b_{N}.
\]
In particular, the generic condition \ref{cndgc} does not hold if $N>1$. We
now determine the eigenmodes of the system. The eigenmodes of this Lagrangian
system are $Q_{j,i}\left(  t\right)  =q_{j}e^{-\mathrm{i}\zeta_{j,i}t}$ where
$q_{j}=Me_{j}$ ($e_{j}$, $1\leq j\leq N$ are the standard orthonormal vectors
in $%
\mathbb{R}
^{N}$) and $\zeta_{j,i}$, $i=-,+$ are%
\[
\zeta_{j,\pm}=-\frac{2\lambda_{j}+\mathrm{i}\beta}{2}\pm\sqrt{\left(
\frac{2\lambda_{j}+\mathrm{i}\beta}{2}\right)  ^{2}+1}.
\]
Therefore, if $0\not \in \sigma\left(  \mathrm{i}\theta\right)  $ then
$\operatorname{Re}\zeta_{j,\pm}\not =0$ so that all eigenmodes are underdamped
for all $\beta>0$ (according to Definition \ref{defodm}). Now since
$N_{R}=\dim\operatorname{Ran}R=N$ then by Theorems \ref{tmdic} and
\ref{tmddI}, we can only have $\lim_{\beta\rightarrow\infty}\left\vert
\zeta_{j,\pm}\left(  \beta\right)  \right\vert =0$ or $\infty$ with an equal
number of each (counting multiplicities). It is easy to verify that%
\[
\zeta_{j,+}=-\zeta_{j,-}^{-1},\text{\quad}\lim_{\beta\rightarrow\infty}%
\beta^{-1}\zeta_{j,+}\left(  \beta\right)  =0,\text{\quad}\lim_{\beta
\rightarrow\infty}\beta^{-1}\zeta_{j,-}\left(  \beta\right)  =-\mathrm{i},
\]
for $1\leq j\leq N$ which implies%
\[
\lim_{\beta\rightarrow\infty}\left\vert \zeta_{j,-}\left(  \beta\right)
\right\vert =\infty,\text{\quad}\lim_{\beta\rightarrow\infty}\zeta
_{j,+}\left(  \beta\right)  =0,\text{\quad for }1\leq j\leq N.
\]
Thus, $\zeta_{j,-}\left(  \beta\right)  $, $1\leq j\leq N$ and $\zeta
_{j,+}\left(  \beta\right)  $, $1\leq j\leq N$ are the corresponding high-loss
and low-loss eigenvalues, respectively. This allows us to illustrate an
interesting difference between this example of a nongeneric case and the
theory developed for the generic case, namely, for the high-loss eigenvalues
\[
\lim_{\beta\rightarrow\infty}\operatorname{Re}\zeta_{j,-}\left(  \beta\right)
=-2\lambda_{j},\text{\quad for }1\leq j\leq N,
\]
which will never be zero if $0\not \in \sigma\left(  \mathrm{i}\theta\right)
$. This is quite a striking difference between the generic case where due to
overdamping (cf. Theorem \ref{tbeta0}) the real part of the high-loss
eigenvalues will be identically zero for $\beta\gg1$!
\end{example}

\textbf{Acknowledgments:} The research of A. Figotin was supported through
AFOSR MURI Grant FA9550-12-1-0489 administered through the University of New
Mexico. The research of A. Welters was supported by the AFOSR through the Air
Force's Young Investigator Research Program (YIP), under grant number
FA9550-15-1-0086. Both authors would like to thank Marcus Marinho for creating
the Figs. \ref{Fig_hldp}-\ref{Fig_udmq} in this paper when he was an
undergraduate (at Florida Institute of Technology and Pontificial Catholic
University) working with A. Welters.

\end{document}